\documentclass[11pt,english]{article}
\usepackage{geometry}
\usepackage{setspace}
\usepackage[T1]{fontenc}%gut fuer akzente, oe, ue etc.
\usepackage[utf8]{inputenc}%reihenfolge fontenc inputenc wichtig
\usepackage{babel}
\usepackage{varioref}
\usepackage{amsmath}
\numberwithin{equation}{section}
\usepackage{amssymb}
\usepackage{graphicx}
\usepackage{floatrow}%for captions next to pictures
\usepackage[labelfont=bf]{caption}
\usepackage{appendix}
\usepackage{slashed}
\usepackage{svg}
\usepackage [autostyle, english = american]{csquotes}
\MakeOuterQuote{"}
\makeatletter

\usepackage{xcolor}
\usepackage{authblk}%%for title page
\usepackage{amsthm}
\usepackage{graphicx}

\usepackage{mathtools}
\usepackage{comment}

\usepackage{hyperref}
\hypersetup{
    colorlinks,
    linkcolor={blue!60!black},
    citecolor={green!60!black},
    urlcolor={blue!80!black}
}
\newtheorem{thm}{Theorem}[section]
\newtheorem{conj}{Conjecture}[section]
\newtheorem{lemma}{Lemma}[section]

\newtheorem{cor}{Corollary}[section]
\theoremstyle{remark}
\newtheorem{rem}{Remark}[section]
\theoremstyle{plain}
\newtheorem{prop}{Proposition}[section]
\newtheorem*{prop*}{Proposition}

\newenvironment{nalign}{
    \begin{equation}
    \begin{aligned}
}{
    \end{aligned}
    \end{equation}
    \ignorespacesafterend
}
\theoremstyle{remark}
\usepackage[disable]{todonotes}

\newcommand{\pu}{\partial_u}
\newcommand{\pv}{\partial_v}

\newcommand{\T}{\mathcal{T}}

\newcommand{\e}{\mathrm{e}}

\newcommand{\DU}{\mathcal{D}_{U_0}}
\renewcommand{\(}{\left(}
\renewcommand{\)}{\right)}

\newcommand{\ualpha}{\underline{\alpha}}
\newcommand{\ubeta}{\underline{\beta}}

\newcommand{\con}{\text{constant}}
\newcommand{\dd}{\mathop{}\!\mathrm{d}}

\title{The Case Against Smooth Null Infinity I: \\Heuristics and Counter-Examples} % Title

\author[1]{Lionor M. A. Kehrberger} % Author name
\affil[1]{University of Cambridge, Department of Applied Mathematics and Theoretical Physics, Wilberforce Road, Cambridge CB3 0WA, United Kingdom}
\date{October 2, 2023} % Date for the report

\makeatother
\begin{document}
\pagenumbering{gobble}

\maketitle % Insert the title, author and date
%\newpage
\begin{abstract}\small
This paper initiates a series of works dedicated to the rigorous study of the precise structure of gravitational radiation near infinity.

We begin with a brief review of an argument due to Christodoulou~\cite{CHRISTODOULOU2002} stating that Penrose's proposal of smooth conformal compactification of spacetime (or smooth null infinity) fails to accurately capture the structure of gravitational radiation emitted by $N$ infalling masses coming from past timelike infinity $i^-$.

  Modelling gravitational radiation by scalar radiation, we then take a first step towards a rigorous, fully general relativistic understanding of the non-smoothness of null infinity by constructing solutions to the spherically symmetric Einstein-Scalar field equations. Our constructions are motivated by Christodoulou's argument: They arise dynamically from polynomially decaying boundary data, $r\phi\sim |t|^{-p}$ as $t\to-\infty$, on a timelike hypersurface (to be thought of as the surface of a star) and the no incoming radiation condition, $r\partial_v\phi=0$, on past null infinity.
    We show that if the initial Hawking mass at past timelike infinity $i^-$ is non-zero, then there exists a constant $C\neq 0$  such that, in the case $p=1$, we obtain the following asymptotic expansion near $\mathcal{I}^+$, precisely in accordance with the non-smoothness of  $\mathcal{I}^+$: $\partial_v(r\phi)=Cr^{-3}\log r+\mathcal{O}(r^{-3})$. Similarly, if $p>1$, we find constant coefficient logarithmic terms appearing at higher orders in the expansion of $\partial_v(r\phi)$.
    
    Even though these results are obtained in the non-linear setting, we show that the same logarithmic terms appear already in the linear theory, i.e.\ when considering the spherically symmetric linear wave equation on a fixed Schwarzschild background. 
    
    As a corollary, we can apply our results to the scattering problem on Schwarzschild: Putting smooth compactly supported scattering data for the linear (or coupled) wave equation on $\mathcal{I}^-$ and on $\mathcal{H}^-$, we find that the asymptotic expansion of $\pv(r\phi)$ near $\mathcal{I}^+$ generically contains logarithmic terms at second order, i.e.\ at order $r^{-4}\log r$. 
\end{abstract}

\newpage
    \begingroup
\hypersetup{linkcolor=black}
    \tableofcontents{}
    \endgroup
\newpage
\pagenumbering{arabic}
\part{Introduction, motivation and summary of the main results}\label{sec:part1}
\section{Introduction}\label{sec:introduction}
This work is concerned with the rigorous mathematical analysis of gravitational waves near infinity. In particular, it contains various dynamical constructions of physically motivated example spacetimes that violate the well-known \textit{peeling property} of gravitational radiation and, thus, do not possess a smooth null infinity.

The paper aims to be accessible to an audience of both mathematicians and physicists. In hopes of achieving this aim, we divided it into two parts, with only the second one containing the actual mathematical proofs. 

In the first part (Part~\ref{sec:part1}), we give some historical background on the concept of smooth null infinity and review an important argument against smooth null infinity due to Christodoulou, which forms the main motivation for the present  work. This is done in section~\ref{sec:introduction}. Motivated by this argument, we then summarise, explain, and discuss the main results of this work (in the form of mathematical theorems) in section~\ref{sec:results}. 

The  proofs of these results are then entirely contained in Part~\ref{sec:counterexamples} of this paper, which, in principle, can be read independently of Part~\ref{sec:part1}.

\subsection{Historical background}
The first direct detection of gravitational waves a few years ago~\cite{Abbott2016ObservationMerger} may not only well be seen as one of the most important experimental achievements in recent times, but also as one of theoretical physics' greatest triumphs.
The theoretical analysis of gravitational waves "near infinity", i.e.\ far away from an \textit{isolated system} emitting them, has seen its basic ideas set up in the 1960s, in works by Bondi, van der Burg and Metzner~\cite{Bondi1962GravitationalSystem}, Sachs~\cite{R.Sachs1961GravitationalCondition,Sachs1962GravitationalSpace-time}, Penrose and Newman~\cite{Newman1962}, and others.
The ideas developed in these works were combined by Penrose's notion of \textit{asymptotic simplicity}~\cite{Penrose1965ZeroBehaviour}, a concept that can now be found in most advanced textbooks on general relativity. The idea behind this notion is to characterise the asymptotic behaviour of gravitational radiation by the requirement that the conformal structure of spacetime be smoothly\footnote{In fact, smooth here can be replaced by $C^k$ for, say, $k\geq4$. } extendable to "null infinity" (denoted by $\mathcal{I}$ and to be thought of as a "boundary of the spacetime") -- the place where gravitational radiation is observed. This requirement is also referred to as the spacetime possessing a "smooth null infinity".
Implied by this smoothness assumption is, amongst other things, the so-called \textit{Sachs peeling property}. This states that the different components of the Weyl curvature tensor fall off with certain negative integer powers of a certain parameter $r$ (whose role will in our context be played by the area radius function) as null infinity is approached along null geodesics~\cite{Penrose1965ZeroBehaviour}.\footnote{A more precise statement is given in section~\ref{sec:intro:CHR} below.}

Although Penrose's proposal of smooth null infinity has certainly left a notable impact %\footnote{Let us note though that the Penrose-Newman formalism~\cite{Newman1962} which preceded Penrose's proposal was much more important for the field's progress.}
 on the asymptotic analysis of gravitational radiation, its assumptions have been subject to debate ever since.
  In particular, the implied Sachs peeling property has been a cause of early controversy; in fact, it remained unclear for decades whether there even exist non-trivial dynamical solutions to Einstein's equations that exhibit the Sachs peeling behaviour or a smooth null infinity. 
  This question has been answered in the affirmative in the case of hyperboloidal initial data in~\cite{Friedrich1983CauchyRelativity,Friedrich1986OnStructure,Andersson1992OnEquations} and, more recently, also in the more interesting case of asymptotically flat initial data in~\cite{Chrusciel2002ExistenceSpacetimes,Corvino2007OnCompactification}, where a large class of asymptotically simple solutions was constructed by gluing the interior part of initial data to e.g.\  Schwarzschild initial data in the exterior (using the gluing results of~\cite{Corvino2000ScalarEquations}) and then exploiting the domain of dependence property combined with the fact that Schwarzschild initial data lead to a smooth null infinity.  See also the recent~\cite{Kroon} or the survey article~\cite{Friedrich2017Peeling} and references therein for related works.
  A similar result with a different approach (based on~\cite{Christodoulou2014}) was obtained in~\cite{Klainerman2003PeelingEquations}, where it was shown that if the initial data decay fast enough towards spatial infinity\footnote{So fast as to force the angular momentum of the initial data set to vanish.}, then the evolution of those data satisfies peeling.

While the analyses above show that the class of solutions with smooth $\mathcal{I}$  is non-trivial, they tell us very little about the \textit{physical} relevance of that class. 
Moreover, several heuristic works~\cite{Bardeen1973RadiationBackground, Schmidt1979TheSpace-Time,Walker1979RelativisticPast,Isaacson1984ExtensionFormula, Winicour1985LogarithmicFlatness} have hinted at  Penrose's regularity assumptions being too rigid to admit physically relevant systems, and a relation between the non-vanishing of the quadrupole moment of the radiating mass distribution and the failure of   $\mathcal{I}$ to be smooth was suggested by Damour using perturbative methods~\cite{Damour1986AnalyticalRadiation.}. 
In fact, there is a much stronger argument against the smoothness of $\mathcal{I}$ due to Christodoulou~\cite{CHRISTODOULOU2002}, which we will review now.
The core contents and results of the present paper (which are logically independent from Christodoulou's argument, but heavily motivated by it) will then be introduced in section~\ref{sec:results}, where we will present various classes of physically motivated counter-examples to smooth null infinity. The reader impatient for the results may wish to skip to section~\ref{sec:results} directly.

\subsection{Christodoulou's argument against smooth null infinity }\label{sec:intro:CHR}
Perhaps the most striking argument against smooth null infinity comes from the monumental work of Christodoulou and Klainerman on the proof of the global non-linear stability of the Minkowski spacetime~\cite{Christodoulou2014}. 
The results of this work do not confirm the Sachs peeling property; moreover, an argument by Christodoulou~\cite{CHRISTODOULOU2002}, which adds to the proof~\cite{Christodoulou2014} a physical assumption on the radiative amplitude on $\mathcal{I}$, shows that this failing of peeling is not a shortcoming of the proof but is, instead, likely to be a true physical effect. 
It is this argument~\cite{CHRISTODOULOU2002} which gives the present section its name, and which forms the main motivation for the present paper. 
Since it does not appear to be widely known, we will give a brief review of it now.

First, let us outline the setup.
In the work~\cite{Christodoulou2014}, given asymptotically flat vacuum initial data sufficiently close to the Minkowski initial data, two foliations of the dynamical vacuum solution $(M,g)$ -- which is shown to remain globally close and quantitatively settle down to the Minkowski spacetime -- are constructed: A foliation of maximal hypersurfaces, which are level sets $\mathcal{H}_t$ of a canonical time function $t$, as well as a foliation of outgoing null hypersurfaces,  level sets $\mathcal{C}_u^+$ of a canonical optical function $u$ (to be thought of as retarded time and tending to $-\infty$ as $i^0$ is approached). 

Let now $e_4$ be a suitable choice of the corresponding generating (outgoing) null geodesic vector field of $\mathcal{C}_u^+$ and $e_3$ a suitable choice of conjugate incoming null normal s.t.\ $g(e_4,e_3)=-2$, let $X,Y$ be vector fields on the spacelike 2-surfaces $S_{t,u}=\mathcal{H}_t\cap\mathcal{C}_u^+$, and let $\slashed{\epsilon}$ be the volume form induced on $S_{t,u}$.
Then, under the following null decomposition\footnote{This decomposition is closely related to the decomposition into the Newman--Penrose scalars $\Psi_0,\dots,\Psi_4$.} of the Riemann tensor $R$,
    \begin{align}\label{eq:intro:nulldecomp}
            \alpha(X,Y):=R(X,e_4,Y,e_4),&&\beta(X):=R(X,e_4,e_3,e_4)\nonumber,\\
            \underline{\alpha}(X,Y):=R(X,e_3,Y,e_3),&&\underline{\beta}(X):=R(X,e_3,e_3,e_4),\\
            4\rho:= R(e_4,e_3,e_4,e_3),	&&2\sigma \slashed{\epsilon}(X,Y):=R(X,Y,e_3,e_4)\nonumber,
    \end{align}
Penrose's regularity requirements would require the Sachs peeling property to hold, i.e., they would require along each $\mathcal{C}_u^+$ the following decay rates, $r$ denoting the area radius of $S_{t,u}$:
    \begin{nalign}\label{eq:intro:SachsPeeling}
        \alpha=\mathcal{O}(r^{-5}), &&\beta=\mathcal{O}(r^{-4}),&&
        \rho=\mathcal{O}(r^{-3}),\\\sigma=\mathcal{O}(r^{-3}),&&\ubeta=\mathcal{O}(r^{-2}),&&\ualpha=\mathcal{O}(r^{-1}).
    \end{nalign}
        In addition, depending on the precise regularity under which the conformal structure of spacetime is assumed to be extendable, Penrose's proposal would imply that all of the above quantities will admit higher-order power series expansions in $1/r$. 
%Similarly, for the past development, it would be required that along each past analogue %of $\mathcal{C}_u^+$, call it $\mathcal{C}_v^-$:
%    \begin{nalign}
%        \ualpha=\mathcal{O}(r^{-5}) &&\ubeta=\mathcal{O}(r^{-4})&&
%        \rho=\mathcal{O}(r^{-3})\\\sigma=\mathcal{O}(r^{-3})&&\beta=\mathcal{O}(r^{-2})&&%\alpha=\mathcal{O}(r^{-1})
%    \end{nalign}
However, the results of~\cite{Christodoulou2014} only confirm the last four rates of  \eqref{eq:intro:SachsPeeling}, whereas, for $\alpha$ and $\beta$, the following weaker decay results are obtained:
\begin{equation}\label{eq:introd:nonpeeling}
   \alpha, \beta=\mathcal{O}(r^{-\frac72}),
\end{equation}
so the peeling hierarchy is \textit{chopped off} at $r^{-7/2}$.

Now, on the one hand, the rates \eqref{eq:introd:nonpeeling} are only shown in~\cite{Christodoulou2014} to be upper bounds  (i.e.\ not asymptotics). Moreover, one might think that these upper bounds can be improved if one imposes further conditions on the initial data -- for, the data considered in~\cite{Christodoulou2014} are only required to have $\alpha, \beta=\mathcal{O}(r^{-7/2})$ on the initial hypersurface. Indeed, %while it is the case that the proof of~\cite{Christodoulou2014} is not able to directly improve the rates \eqref{eq:introd:nonpeeling} even if the initial data decay faster,
 one can slightly adapt the methods of Christodoulou--Klainerman to show that if the initial data decay much faster than assumed in~\cite{Christodoulou2014}, the peeling rates \eqref{eq:intro:SachsPeeling} can indeed be recovered~\cite{Klainerman2003PeelingEquations}. We will return to this at the end of this section.

On the other hand, as remarked before, the fundamental question is not whether there exist initial data which lead to solutions satisfying peeling, but whether \textbf{\textit{physically relevant spacetimes satisfy peeling.}} Evidently, any answer to this latter question must appeal to some additional physical principle. This is exactly what Christodoulou does in \cite{CHRISTODOULOU2002}. There, he shows that, indeed, the rates \eqref{eq:intro:SachsPeeling} cannot be recovered in several physically relevant systems, making the idea of smooth $\mathcal{I}$ physically implausible. 
%On the other hand, shifting the discussion away from the question of which initial data lead to solutions that exhibit peeling, Christodoulou~\cite{CHRISTODOULOU2002} gave an argument (using the framework of~\cite{Christodoulou2014}) which suggests that, indeed, the rates \eqref{eq:intro:SachsPeeling} cannot be recovered in several physically relevant systems, making the idea of smooth $\mathcal{I}$ physically implausible. 
%As the paragraph above shows, such an argument cannot follow from sufficiently strong Cauchy data assumptions, but must instead appeal to some additional physical principle. 
At the core of Christodoulou's argument lies the assumption that the \textit{Bondi mass} along $\mathcal I^+$ decays with the rate predicted by the quadrupole approximation for a system of $N$ infalling masses coming from past infinity\todo{leave out the following}%with initially neither vanishing nor relativistic relative velocities
, combined with the assumption that there be no incoming radiation from past null infinity.

\begin{rem}\label{rem1}
Before we move on to explain Christodoulou's argument, we shall make an important remark. 
Even though we stressed that one should not derive arguments for or against peeling from sufficiently strong Cauchy data assumptions, but rather appeal to some physical ingredients,  
we still want to make \underline{some} initial data assumptions in order to have access to the results of~\cite{Christodoulou2014}.
These results, a priori, only hold for evolutions of asymptotically flat vacuum initial data sufficiently close to Minkowski initial data, i.e.\ data for which, in particular, a certain Sobolev norm $||\cdot||_{\text{CK}}$ is small. We shall call such data \textbf{C--K small data}. 

Of course, C--K small data are not directly suited to describe the evolutions of spacetimes with $N$ infalling masses.
However, consider now initial data which are only required to have finite (as opposed to small) $||\cdot||_{\text{CK}}$-norm and to be vacuum only in a neighbourhood of spatial infinity (as opposed to everywhere). We shall call such data \textbf{C--K compatible}. Let us explain this terminology: 
One can now restrict these data to a region, let's call it the exterior region, sufficiently close to spacelike infinity in a way so that the data in this exterior region are vacuum and have arbitrarily small $||\cdot||_{\text{CK}}$-norm. 
By a gluing argument, one can then  extend these exterior data to interior data whose $||\cdot||_{\text{CK}}$-norm can also be chosen sufficiently small so that the resulting glued data are C--K small. % gluing results~\cite{BieriChru16gluing,Chru17}, one can then extend these exterior data to interior data whose $||\cdot||_{\text{CK}}$-norm can also be chosen sufficiently small so that the resulting glued data are C--K small.
 Therefore, the results of~\cite{Christodoulou2014} apply to the (C--K small) glued data, and, thus, by the domain of dependence property, they apply to the domain of dependence of the exterior part of the (C--K compatible) original data, i.e.\ in a neighbourhood of spacelike infinity containing a piece of null infinity.\footnote{We note that one should be able to avoid this gluing argument by appealing to the results of~\cite{Klainerman2003TheRelativity}, see the first remark below Definition~3.6.4 therein. Alternatively, one could also use the results of the more general~\cite{Keir}, since that work does not require the constraint equations to be satisfied on data.}
%By the gluing results~\cite{Bieri2016Future-completeMass,Chrusciel2017LongGluing}, one can then extend these exterior data to interior data whose $||\cdot||_{\text{CK}}$-norm can also be chosen sufficiently small so that the resulting glued data are C--K small.
% Therefore, the results of~\cite{Christodoulou2014} apply to the (C--K small) glued data, and, thus, by the domain of dependence property, they apply to the domain of dependence of the exterior part of the (C--K compatible) original data, i.e.\ in a neighbourhood of spacelike infinity containing a piece of null infinity.\footnote{We note that one should be able to avoid this gluing argument by appealing to the results of~\cite{Klainerman2003TheRelativity}, see the first remark below Definition~3.6.4 therein. Alternatively, one could also use the results of the more general~\cite{Keir}, since that work does not require the constraint equations to be satisfied on data.}
It is evolutions of C--K compatible data that we shall make statements on. One can reasonably expect that such evolutions contain a large class of physically interesting systems such as that of $N$ infalling masses from the infinite past.
\end{rem}

We can now paraphrase Christodoulou's result~\cite{CHRISTODOULOU2002}:
\par\smallskip\noindent
\centerline{\fbox{\begin{minipage}{0.95\textwidth}
\textit{Consider all evolutions of C--K compatible initial data which \textbf{a)} satisfy on $\mathcal I^-$ the no incoming radiation condition and \textbf{b)} behave on $\mathcal I^+$ as predicted by the quadrupole approximation for  $N$ infalling masses. These evolutions do not admit a smooth conformal compactification. 
}
\end{minipage}}}
\par\smallskip
More precisely, the failure of these evolutions to admit a smooth conformal compactification  manifests itself in the asymptotic expansion of $\beta$ near future null infinity containing logarithmic terms at leading order (namely, at order $r^{-4}\log r$).\footnote{As will be discussed in~\cite{IV}, this interpretation of Christodoulo's statement, particularly with the decay rate of $\beta$ being of order $r^{-4}\log r$, is not quite correct, at least not in the context of linearised gravity. Furthermore, it is explained in~\cite{IV} that the class of C--K compatible data is too small to capture relevant physics. We added this footnote to the third version of the paper so as to already make the reader aware of this, for details, see~\cite{IV}.}

Let us briefly expose the main ideas of the proof of the above statement:
We recall from~\cite{Christodoulou2014} that the traceless part $\hat{\underline{\chi}}$ of the connection coefficient
\begin{equation}
    \underline{\chi}(X,Y)=g(\nabla_X e_3,Y)
\end{equation}
tends along any given $\mathcal{C}_u^+$ to
    \begin{equation}
        \lim_{\mathcal{C}_u^+,r\to\infty}r \hat{\underline{\chi}}=\Xi(u)
    \end{equation}
 as the area radius function $r$ associated to $S_{t,u}$ tends to infinity. 
 Here, $\Xi(u)$ is a 2-form on the unit sphere $\mathbb S^2$ that should be thought of as living on future null infinity and which defines the radiative amplitude per solid angle. The quantity $\hat{\underline{\chi}}$ is often called the ingoing shear of the 2-surfaces $S_{t,u}$, and the limit $\Xi$ is sometimes referred to as \textit{Bondi news}.
 Indeed, one of the many important corollaries of~\cite{Christodoulou2014} is the \emph{Bondi mass loss formula}: If $M(u)$ denotes the Bondi mass along $\mathcal I^+$, then we have
 \begin{equation}\label{eq:Argument:Bondimass}
 \frac{\partial}{\partial u	}M(u)=-\frac{1}{16\pi}\int_{\mathbb S^2}|\Xi(u,\cdot)|^2.
 \end{equation}
 
 Now, the quadrupole approximation for $N$ infalling masses predicts that $\pu M\sim -|u|^{-4}$ as $u\to-\infty$ (it is assumed that the relative velocities tend to constant values near the infinite past and that the mass distribution has non-vanishing quadrupole moment) and, thus, in view of \eqref{eq:Argument:Bondimass}, that
 \begin{equation}\label{eq:intro:Xi-}
 \lim_{u\to -\infty} u^2\Xi=:\Xi^-\neq 0.
 \end{equation}

Christodoulou's two core observations then are the following: Even though $\beta$ itself only decays like $r^{-7/2}$ (see \eqref{eq:introd:nonpeeling}), its derivative in the $e_3$-direction  decays like $r^{-4}$ as a consequence of the differential Bianchi identities. Schematically, an analysis of Einstein's equations on $\mathcal I^+$ moreover reveals that, assuming \eqref{eq:intro:Xi-}, 
\begin{equation}
 \lim_{\mathcal{C}_u^+,r\to\infty}\partial_u(r^4\beta)=\frac{\slashed{\mathcal{D}}^{(3)}\Xi^-}{|u|}+\dots,
\end{equation}
where $\slashed{\mathcal{D}}^{(3)}$ is a third-order differential operator on $\mathbb S^2$.
The most difficult part of the argument then consists of obtaining a similar estimate for $\pu(r^4\beta)$ away from null infinity. Once this is achieved, one can integrate $\pu(r^4\beta)$ from initial data ($t=0$) to obtain schematically (see Figure~\ref{fig:1} below):
\begin{equation}\label{eq:Argument:betadifference}
 (r^4\beta)(u_2,t)-(r^4\beta)(u_1(t),0)\sim \int_{u_1(t)}^{u_2} \frac{\slashed{\mathcal{D}}^{(3)}\Xi^-}{|u|}\dd u\sim (\log r_{t,u_2}-\log |u_2|)\cdot \slashed{\mathcal{D}}^{(3)}\Xi^-.
\end{equation}
%\begin{SCfigure}
%\caption{Schematic depiction of Christodoulou's argument.}\label{fig:1}
%{\includegraphics[width = 200pt]{Christodoulouproofwitht.pdf}}
%\end{SCfigure}
%
%%%%%%%%%%%
\begin{figure}[htbp]
\floatbox[{\capbeside\thisfloatsetup{capbesideposition={right,top},capbesidewidth=4cm}}]{figure}[\FBwidth]
{\caption{Schematic depiction of Christodoulou's argument. Integrating $\pu(r^4\beta)\sim|u|^{-1}$ from initial data gives rise to logarithmic terms.}\label{fig:1}}
{\includegraphics[width = 200pt]{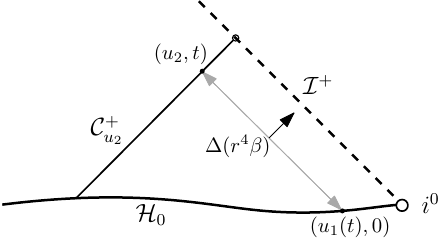}}
\end{figure}
%%%%%%%%%%%
%\begin{figure}[htbp]
%  \centering
%  \includegraphics[width = 200pt]{Christodoulouproofwitht.pdf}
%\end{figure}
Here, $r_{t,u_2}$ denotes the area radius of $S_{t,u_2}$, and we used that $u_1(t)\sim r_{t,u_2}$. 

Finally, Christodoulou argues that $r^4\beta$ remains finite on $t=0$ as a consequence of the \textit{no incoming radiation condition}, which is the statement that the Bondi mass remains constant along past null infinity.

He thus concludes that the peeling property is violated by $\beta$, and that one instead has that
\begin{equation}
\beta=B^*r^{-4}(\log r-\log|u|)+\mathcal O(r^{-4})
\end{equation}
for a 1-form $B^*$ which encodes physical information about the quadrupole distribution of the infalling matter and which is independent of $u$. 

Similarly, he shows show that $\alpha=O(r^{-4})$, in contrast to the $r^{-5}$-rate predicted by peeling.

Now, rather than \textit{imposing} \eqref{eq:intro:Xi-}, it would of course be desirable to \emph{dynamically derive} the rate \eqref{eq:intro:Xi-} (and thus the failure of peeling) from a suitable scattering setup resembling that of $N$ infalling masses. 

In fact, this is exactly what we present in section~\ref{sec:introduction:scalar}, albeit for a simpler model. In this context, we will also be able to motivate the following simpler conjectures (cf.\ Thms.~\ref{thm.intro:nullcase} and~\ref{thm.intro:scattering}):
\begin{conj}\label{conj1}
Consider the scattering problem for the Einstein vacuum equations with conformally regular data on an ingoing null hypersurface and no incoming radiation from past null infinity. Then, generically, the future development fails to be conformally smooth near $\mathcal I^+$.
\end{conj}
\begin{conj}\label{conj2}
Consider the \textit{scattering problem} for the Einstein vacuum equations with compactly supported data on $\mathcal{I}^-$ and a Minkowskian $i^-$. Then, generically, the future development fails to be conformally smooth near $\mathcal{I}^+$. 
\end{conj}
To clarify, we do not explicitly conjecture that the leading-order peeling behaviour \eqref{eq:intro:SachsPeeling} is violated, but that there will be logarithmic terms in the expansion of e.g.~$\beta$ or $\alpha$ at some finite, potentially higher order.

Before we move on to the next section, we feel that it may be helpful to comment on the work~\cite{Klainerman2003PeelingEquations}. There, it is shown that if one works with faster decaying $r^{\frac12+\epsilon}$-weighted  C--K data (which have finite $||r^{\frac12+\epsilon}\cdot||_{\text{CK}}$-norm), then peeling holds for $\beta$ if $\epsilon>0$, and also for $\alpha$ if $\epsilon>1$. 
So how is this consistent with the above result? Well, one of the results of~\cite{Klainerman2003PeelingEquations} implies that $r^{\frac12+\epsilon}$-weighted C--K data lead to solutions which have $|\Xi|\leq |u|^{-2-\epsilon}$, hence the data considered in~\cite{Klainerman2003PeelingEquations} are incompatible with eq.~\eqref{eq:intro:Xi-} or, in other words, with the quadrupole approximation of $N$ infalling masses. The same applies to~\cite{Chrusciel2002ExistenceSpacetimes,Corvino2007OnCompactification}.

\section{Overview of the main results (Thms.~\ref{thm.intro:timelikecase}--\ref{thm.intro:scattering}) and of upcoming work }\label{sec:results}
\subsection{Construction of counter-examples to smooth null infinity within the  Einstein-Scalar field system in spherical symmetry} \label{sec:introduction:scalar}
While the argument~\cite{CHRISTODOULOU2002} presented above already forms a serious obstruction to peeling, one would ultimately -- in order to develop a fully general relativistic understanding of the non-smoothness of null infinity -- like to actually construct solutions to Einstein's equations that resemble the setup of $N$ infalling masses from past infinity (and which lead to \eqref{eq:intro:Xi-} \textit{dynamically}). That is to say, one would like to understand the semi-global evolution of a configuration of $N$ masses at past timelike infinity with no incoming radiation from $\mathcal{I}^-$. More concretely, one would like to understand the asymptotics of such solutions in a neighbourhood of $i^0$ containing a piece of $\mathcal{I}^+$.

Of course, the resolution of this problem seems to be quite difficult.

We will therefore, in this paper, take only a first step towards the resolution of said problem by explicitly constructing a fully general relativistic example system 
 that is based on a simple realisation of infalling masses from past timelike infinity and the no incoming radiation condition; 
 %(which is simple enough to be studied rigorously but rich enough to capture the failure of peeling near infinity);
 namely, we consider the Einstein-Scalar field  equations for a chargeless and massless scalar field under the assumption of spherical symmetry:
\begin{align}
    R_{\mu\nu}-\frac 1 2 R g_{\mu\nu}=2T_{\mu\nu}=2T^{sf}_{\mu\nu},%+T^{em}_{\mu\nu}),
\end{align}
with the matter content\footnote{We can also include a Maxwell field that is coupled to the geometry (and not to the scalar field) in the equations.} given by
\begin{align}
    T^{sf}_{\mu\nu}= \phi_{;\mu} \phi_{;\nu} - \frac 1 2 g_{\mu\nu}\phi^{;\xi} \phi_{;\xi}. %\\
   % T^{em}_{\mu\nu}=F_{\mu\xi}F^{\xi}_\nu-\frac 1 4 g_{\mu\nu} F_{\xi o}F^{\xi o}
\end{align}
Here, $\phi$ denotes the scalar field,
%and $F$ the electromagnetic field tensor,
$R_{\mu\nu}$ the Ricci tensor,  $R$ the scalar curvature of the metric $g_{\mu\nu}$, and "$;$" denotes covariant differentiation. 
 
The assumption of spherical symmetry essentially allows us to write the unknown metric in double null coordinates $(u,v)$ as
\begin{equation}
g=-\Omega^2\dd u\dd v+r^2\,\gamma,
\end{equation}
where $\gamma$ is the standard metric on the unit sphere $\mathbb S^2$, and where $\Omega$ and $r$ (the area radius function) are functions depending only on $u$ and $v$. The spherically symmetric Einstein-Scalar field system thus reduces to a system of hyperbolic partial differential equations for the unknowns $\Omega$, $r$ and $\phi$ in two dimensions. In practice, it is often convenient to replace $\Omega$ in this system with the \textit{Hawking mass} $m$, which is defined in terms of $\Omega$ and~$r$.

We construct for this system data resembling the assumptions of Christodoulou's argument that lead to a non-smooth future null infinity in the following way:

On past null infinity, to resemble the no incoming radiation condition (for more details on the interpretation of this, see Remark~\ref{remark:null:bondimass2}), we set 
\begin{equation}\label{eq:intro:noincomingradiation}
    \pv(r\phi)|_{\mathcal{I}^-}=0,
\end{equation}
where $v$ is advanced time, see Figure~\ref{fig:2} below.
Note that, in spherical symmetry, it is not possible to have $N$ infalling masses for $N>1$. We thus have to restrict to a single infalling mass. In particular, there can be no non-vanishing quadrupole moment.
To still have some version of "infalling masses" that emit (scalar) radiation, we therefore impose decaying boundary data on a smooth timelike hypersuface\footnote{The precise conditions on $\Gamma$ are fairly general and, in particular, admit cases where $r|_\Gamma$ tends to a finite or infinite limit. For the derivation of upper bounds, we only require $r|_\Gamma>2M$, where $M$ is defined in \eqref{eq:intro:Hawking}. For the derivation of lower bounds, we require the slightly stronger assumption $r|_{\Gamma}> 2.95 M$. We expect that this bound can be improved.} $\Gamma$ (to be thought of as the surface of a single star) such that 
\begin{align}\label{eq:intro:Gammafields}
    r\phi|_\Gamma = \frac{C}{|t|^{p-1}}+\mathcal{O}\left(\frac{1}{|t|^{p-1+\epsilon}}\right),&&\boldsymbol{T}(r\phi|_\Gamma)= \frac{(p-1)C}{|t|^{p}}+\mathcal{O}\left(\frac{1}{|t|^{p+\epsilon}}\right),
\end{align}
where $C\neq0$ and $p>1$ are constants, $\boldsymbol{T}$ is the normalised future-directed vector field generating $\Gamma$, and $t$ is its corresponding parameter ($\boldsymbol{T}(t)=1$), tending to $-\infty$ as $i^-$ is approached. %\footnote{Here and in the remainder of the paper, we will write $f\sim g$ if there exist positive constants $A,B$ such that $Af\leq g\leq Bf$.}
It will turn out that, in the case $p=2$,  this condition implies the precise analogue of eq.~\eqref{eq:intro:Xi-}, i.e.\ the prediction of the quadrupole approximation (see also the Remark~\ref{long remark ZetaXi}). This motivates the case $p=2$ to be the most interesting one.

Finally, we need the "infalling mass" to be non-vanishing; we thus set the Hawking mass $m$ to be positive initially, i.e.\ at $i^-$: 
\begin{equation}\label{eq:intro:Hawking}
    m(i^-)=M>0.
\end{equation}
\begin{rem}
Note already that conditions \eqref{eq:intro:noincomingradiation} and \eqref{eq:intro:Hawking} are to be understood in a certain limiting sense; indeed, we will construct solutions where $\mathcal{I}^-$ is replaced by an outgoing null hypersurface $\mathcal{C}^+_{u_n}$ at finite retarded time $u_n$ and then show that the solutions to these mixed characteristic-boundary value problems converge to a unique limiting solution as $u_n\to-\infty$, that is, as $\mathcal{C}_{u_n}^+$ "approaches" $\mathcal{I}^-$. We will then show that the solution constructed in this way is the unique solution to our problem, cf.\ Remark~\ref{rem:intro:unique}.
\end{rem}

%\begin{figure}[htbp]
%  \centering
%  \includegraphics[width = 140pt]{PenroseDiagramTimeLikeCaseIntrouv.pdf}
%\end{figure} 
To more clearly state the following rough versions of our results, we remark that, throughout most parts of this work, we work in a globally regular double null coordinate system $(u,v)$ (see Figure~\ref{fig:2} below) in which $\mathcal{I}^+$ can be identified with $v=\infty$, $\mathcal{I}^-$ can be identified with $u=-\infty$, and which satisfies $u=v$ on $\Gamma$ and $\pv r=1$ along $\mathcal{I}^-$ (in a limiting sense).

We then have the following theorem (see Thms.~\ref{thm:timelike:final!!!} and~\ref{thm:timelike:logs} for the precise statement):
    \begin{thm}\label{thm.intro:timelikecase}
        For sufficiently regular initial/boundary data on $\mathcal{I}^-$ and $\Gamma$ as above, i.e.\ obeying eqns.\ \eqref{eq:intro:noincomingradiation}, \eqref{eq:intro:Gammafields}, \eqref{eq:intro:Hawking},  a unique semi-global solution to the spherically symmetric Einstein-Scalar field system exists for sufficiently large negative values of $u$. Moreover, if $p=2$, we get the following asymptotic behaviour for the outgoing derivative of the radiation field:\footnote{Here, and in the remainder of the paper, we write $f\sim g$ if there exist positive constants $A,B$ s.t.\ $Af\leq g\leq Bf$. Similarly, we write $f=\mathcal O(g)$ if there exists a constant $A$ s.t.\ $|f|\leq A g$. }
\begin{equation}
     |\pv(r\phi)|\sim  
\begin{cases}
\frac{\log r}{r^3}, & u=\con,\,\, v \to \infty ,\\
\frac{1}{r^3}, & v=\con,\,\, u \to -\infty,\\
\frac{1}{r^3}, & v+u=\con,\,\, v\to \infty.
\end{cases}
\end{equation}
    More precisely, for fixed values of $u$, we obtain the following asymptotic expansion as $\mathcal{I}^+$ is approached:
        \begin{equation}\label{eq:intro:thmtimelikeasymptotics}
            \pv(r\phi)(u,v)=B^* \frac{\log r-\log|u|}{r^3}+\mathcal{O}(r^{-3}).
        \end{equation}
        Here, $B^*\neq 0$ is a constant independent of $u$ given by 
        \begin{equation}
        B^*=-2M\lim_{u\to - \infty}|u|r\phi(u,v),
        \end{equation}
    and the limit above exists and is independent of $v$.
    \end{thm}
    %%%%%%%%%%%
\begin{figure}[htbp]
\floatbox[{\capbeside\thisfloatsetup{capbesideposition={right,top},capbesidewidth=4.4cm}}]{figure}[\FBwidth]
{\caption{The Penrose diagram of the solution of Theorem~\ref{thm.intro:timelikecase}. We impose polynomially decaying data on a timelike boundary $\Gamma$ and no incoming radiation from past null infinity $\mathcal I^-$. Note that, with our choice of coordinates ($u=v$ on $\Gamma$), $\Gamma$ becomes a straight line.}\label{fig:2}}
{ \includegraphics[width = 135pt]{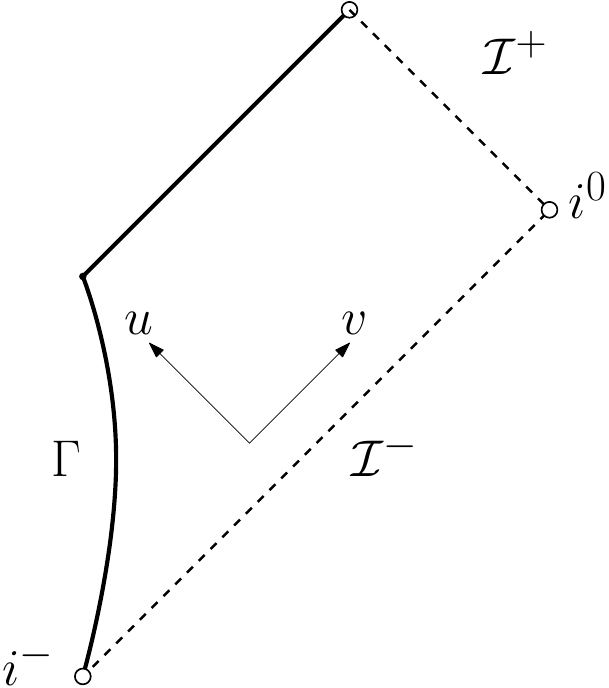}}
\end{figure}
%%%%%%%%%%%
   \begin{rem}\label{rem:intro:unique}
The uniqueness in the above statement is with respect to the class of solutions with uniformly bounded Hawking mass. See also Remark \ref{rem:uniqueness}. 
    %Roughly speaking, we show uniqueness w.r.t.\ the class of solutions that have uniformly bounded Hawking mass and which also satisfy $\pv^2(r\phi)|_{\mathcal I^-}=\pv^2 r|_{\mathcal I^-}=0$. We expect that there do not exist solutions that restrict correctly to the data on $\mathcal I^-$ and $\Gamma$ and do not satisfy these assumptions, but we do not show this here.
    \end{rem}
%    \begin{rem}\label{rem:intro:unique}
%    For the precise sense in which these solutions are unique, see Remark \ref{rem:uniqueness}. Roughly speaking, the uniqueness shown in this paper is w.r.t.\ the class of solutions that have uniformly bounded Hawking mass and which also satisfy $\pv^2(r\phi)|_{\mathcal I^-}=\pv^2 r|_{\mathcal I^-}=0$. We expect that there do not exist solutions that restrict correctly to the data on $\mathcal I^-$ and $\Gamma$ and do not satisfy these assumptions, but we do not show this here.
%    \end{rem}
     Theorem \ref{thm.intro:timelikecase} shows that the asymptotic expansion of $\pv(r\phi)$ near $\mathcal{I}^+$, which should be thought of as the analogue to $\beta$ for the wave equation, contains logarithmic terms and, thus, fails to be regular in the conformal picture (i.e.\ in the variable $1/r$), whereas the expansion near $\mathcal{I}^-$ remains regular.\footnote{\label{fn:1}Since various ideas regarding the relation between the conformal regularity of $\mathcal I^-$ and  that of $\mathcal I^+$ have been entertained in the literature, we want to point out that, in our setting, the smoothness or non-smoothness of $\mathcal I^-$ is completely inconsequential to the smoothness of $\mathcal I^+$.  See also footnote~\ref{fn:2}.}
    
  One can moreover show that, for general integer $p>2$ , one instead gets the following expansion for fixed values of $u$:
         \begin{equation}
            \pv(r\phi)(u,v)=B(u)\frac{1}{r^3}+\dots +B'\frac{\log r}{r^{p+1}}+\mathcal{O}(r^{-p-1}),
        \end{equation}
        where the $\dots $-terms denote negative integer powers of $r$, and where $B'\neq 0$ is a constant determined by $M$ and $\lim_{u\to-\infty}|u|^{p-1}r\phi$, the latter limit again being independent of $v$.

We can also state the precise analogue of the argument~\cite{CHRISTODOULOU2002} presented in section~\ref{sec:intro:CHR} for the Einstein-Scalar field system (see Remark~\ref{rem:proofofthmintrocor}):
   \begin{thm}\label{thm.intro.corollary}
       Suppose a semi-global solution to the spherically symmetric Einstein-Scalar field system with Hawking mass $m\geq c>0$ for some constant $c$ and $m(\mathcal I^-)\equiv M>0$ and obeying the no incoming radiation condition exists such that, on $\mathcal{I}^+$, $r\phi =\Phi^- |u|^{-1}+\mathcal{O}(|u|^{-1-\epsilon})$. 
       Then, for fixed values of $u$, we obtain the following asymptotic expansion of $\pv(r\phi)$ as $\mathcal{I}^+$ is approached:
        \begin{equation}
            \pv(r\phi)(u,v)=B^* \frac{\log r-\log|u|}{r^3}+\mathcal{O}(r^{-3}),
        \end{equation}
        where $B^*$ is a constant independent of $u$ given again by $-2M\Phi^-$.
    \end{thm}
Indeed, the main work of this paper consists of showing that both the lower and upper bounds on the $u$-decay of $r\phi$ imposed on $\Gamma$ are propagated all the way up to $\mathcal{I}^+$.\footnote{The propagation of $u$-decay is somewhat special to spherical symmetry, see also section~\ref{sec:intro:higherl} and the upcoming~\cite{Kerrburger3}.} The limit $\Phi^-$ then plays a similar role to $\Xi^-$ from \eqref{eq:intro:Xi-}, see already Remark~\ref{long remark ZetaXi}.

We remark that, even though the above theorems are proved for the coupled problem, the methods of the proofs can also be specialised to the linearised problem (see section~\ref{sec:linear} of the present paper or  section~11 of~\cite{Dafermos2005}), i.e.\ the problem of the wave equation on a fixed Schwarzschild (or Reissner--Nordstr\"om) background:
   \begin{thm}\label{thm.intro:linearcase}
        Consider the spherically symmetric wave equation
        \begin{equation}
            \nabla^\mu\nabla_\mu \phi=0
        \end{equation}
       on a fixed Schwarzschild background with mass $M\neq 0$, where $\nabla$ is the connection induced by the Schwarzschild metric
        \begin{equation}
            g^{\text{Schw}}=-\left(1-\frac{2M}{r}\right)\dd t^2+\left(1-\frac{2M}{r}\right)^{-1}\dd r^2+r^2 \dd\Omega^2,
        \end{equation}
        and consider sufficiently regular initial/boundary data as above, i.e.\ obeying eqns.~\eqref{eq:intro:noincomingradiation} and \eqref{eq:intro:Gammafields}.
        Then the results of Theorems~\ref{thm.intro:timelikecase},~\ref{thm.intro.corollary} apply.
    \end{thm}
    Notice that the same result does not hold on Minkowski, as we need the spacetime to possess some mass near spatial infinity.
    
   Let us now explain, both despite and due to its simplicity, the main cause for the logarithmic term (focusing now on $p=2$):  The wave equation (derived from $\nabla^\mu T^{sf}_{\mu\nu}=\nabla^\mu R_{\mu\nu}=0$) then reads
\begin{equation}\label{eq:intro:waveequation}
    \pu\pv(r\phi)=-2m\frac{(-\pu r) \pv r}{1-\frac{2m}{r}}\frac{r\phi}{r^3}.
\end{equation}
 Assuming that we can propagate upper and lower bounds for $r\phi$ from $\Gamma$ to null infinity, we have that $r\phi\sim |u|^{-1}$ everywhere. 
For sufficiently large $r$, and for sufficiently large negative values of $u$, we then have that $r(u,v)\sim (v-u)$ and that all other terms appearing in front of the $\frac{r\phi}{ r^{3}}$-term remain bounded from above, and away from zero, such that integrating \eqref{eq:intro:waveequation} from $\mathcal{I}^-$ gives (we decompose into fractions)
\begin{align}\begin{split}\label{eq:intro:heuristic}
    &\pv(r\phi)(u,v)\sim -\int_{-\infty}^u \frac{1}{r(u',v)^3|u'|}\dd u' \sim \int_{-\infty}^u \frac{1}{(v-u')^3 u'}\dd u'  \\
   &= \int_{-\infty}^u \frac{1}{v^3}\left(\frac{1}{u' }+\frac{1}{ v-u'}+\frac{v}{ (v-u')^2}+\frac{v^2}{(v-u')^3}\right)\dd u'=       \frac{\log|u|-\log(v-u)}{v^3}+\frac{3v-2u}{2v^2(v-u)^2}.
\end{split}\end{align}
Taking the limit of $v\to\infty$ while fixing $u$ then, already, suggests the logarithmic term in the asymptotic expansions of Thms.~\ref{thm.intro:timelikecase} and~\ref{thm.intro.corollary}. Of course, the calculation above  is only a sketch, and many details have been left out.\footnote{\label{fn:2}To relate to footnote~\ref{fn:1}, note that one can do a similar calculation if $r\phi\sim|u|^{-\epsilon}$ for some $\epsilon>0$. In this case, there will be an $r^{-2-\epsilon}$-term in the asymptotic expansion of $\pv(r\phi)$, unless $\epsilon\in\mathbb{Z}$. In other words, one cannot recover conformal regularity near $\mathcal I^+$ from a lack of conformal regularity near $\mathcal I^-$.}

Let us remark that posing polynomially decaying boundary data on a timelike hypersurface comes with various technical difficulties. For instance, one cannot a priori prescribe the Hawking mass on $\Gamma$ -- in fact, even showing local existence will come with some difficulties -- and $r$-weights cannot be used to infer decay when integrating in the outgoing direction from $\Gamma$ since $r$ is, in general, allowed to remain bounded on $\Gamma$.
Both of these difficulties disappear in the characteristic initial value problem, i.e., when one prescribes initial data on an ingoing null hypersurface $\mathcal{C}_{\mathrm{in}}$ terminating at past null infinity (see Figure~\ref{fig:3}) according to 
\begin{align} \label{eq:intro:nullrayfields}
 r\phi|_{\mathcal{C}_{\mathrm{in}}} = \frac{\Phi^-}{r^{p-1}}+\mathcal{O}\left(\frac{1}{r^{p-1+\epsilon}}\right),&&\left.\frac{\pu(r\phi)}{\pu r}\right|_{\mathcal{C}_{\mathrm{in}}}= \frac{(p-1)\Phi^-}{r^{p}}+\mathcal{O}\left(\frac{1}{r^{p+\epsilon}}\right),
\end{align}
where $\Phi^-$ and $p>1$ are constants, one again sets $\pv(r\phi)$ to vanish on past null infinity,
and makes the obvious modification to condition \eqref{eq:intro:Hawking}:
\begin{equation}\label{eq:intro:Hawking2}
    m(\mathcal{C}_{\mathrm{in}}\cap \mathcal{I}^-)=M>0.
\end{equation}
%\begin{figure}[htbp]
%  \centering
%  \includesvg[width = 120pt, svgpath = pictures/somepicture]{PenroseDiagramNullCase}
%\end{figure}
We then obtain the following theorem (see Thm.~\ref{thm:null:asymptotics of dvrphi} for the precise statement):
    \begin{thm}\label{thm.intro:nullcase}
        For sufficiently regular characteristic initial data on $\mathcal{I}^-$ and $\mathcal{C}_{\mathrm{in}}$ as above, i.e.\ obeying eqns.~\eqref{eq:intro:noincomingradiation}, \eqref{eq:intro:nullrayfields}, \eqref{eq:intro:Hawking2}, a unique semi-global solution to the Einstein-Scalar field system in spherical symmetry exists for sufficiently large negative values of $u$. Moreover, in the case $p=2$, we obtain the following asymptotic expansion of $\pv(r\phi)$ as $\mathcal{I}^+$ is approached along hypersurfaces of constant $u$:
        \begin{equation}
            \pv(r\phi)(u,v)=B^* \frac{\log r-\log|u|}{r^3}+\mathcal{O}(r^{-3}),
        \end{equation}
        where $B^*$ is a constant independent of $u$ given by $B^*=-2M\Phi^-$.
        On the other hand, the expansion near $\mathcal{I}^-$ remains regular, i.e.\ $\pv(r\phi)=\mathcal{O}(r^{-3})$ near $\mathcal I^-$.
        
        As before, the same holds true for the linear case, cf.\ Thm.~\ref{thm.intro:linearcase}.
         \end{thm}
%%%%%%%%%%%
\begin{figure}[htbp]
\floatbox[{\capbeside\thisfloatsetup{capbesideposition={right,top},capbesidewidth=4cm}}]{figure}[\FBwidth]
{\caption{The Penrose diagram of the solution of Theorem~\ref{thm.intro:nullcase}. We impose polynomially decaying data on an ingoing null hypersurface $\mathcal C_{\mathrm{in}}$ and no incoming radiation from past null infinity $\mathcal I^-$.}\label{fig:3}}
{  
\includegraphics[width = 155pt]{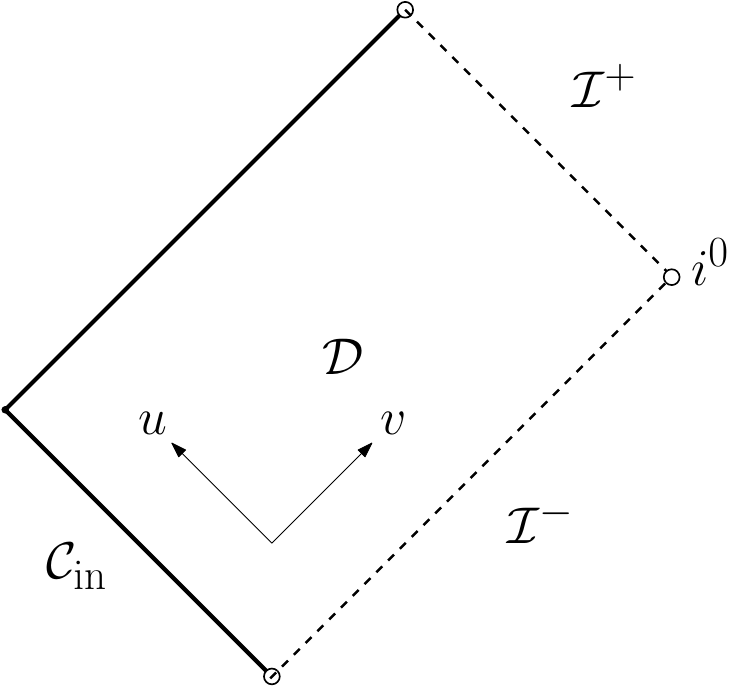}
}
\end{figure}
%%%%%%%%%%%
%    \begin{figure}[htbp]
%  \centering
%  \includegraphics[width = 160pt]{PenroseDiagramNullCaseIntroD.pdf}
%\end{figure}
It is this result which motivates Conjecture~\ref{conj1} from section~\ref{sec:intro:CHR}.
   
    Since the characteristic setup above is much simpler to deal with compared to the case of boundary data on $\Gamma$, we shall prove Thm.~\ref{thm.intro:nullcase} first such that the technically more involved timelike case can be understood more easily afterwards.
    Moreover, it turns out that this setting allows for another interesting motivation or interpretation of our choice of polynomially decaying initial data, namely in the context of the \textit{scattering problem} of scalar perturbations of Minkowski or Schwarzschild. We will discuss this in the next section (section~\ref{intro:scattering}).
    
    On the other hand, the problem of timelike boundary data is interesting precisely because of its difficulties and the methods used to deal with them. 
    Indeed, we develop a quite complete understanding of the evolutions of such data in Thm.~\ref{thm:timelike:final!!!}. 
    Let us point out again that we are not able to work directly with such data, but rather need to consider a sequence of smooth compactly supported data that lead to solutions which can be extended to the past by the vacuum solution. 
    We will show uniform bounds and sharp decay rates for this sequence of solutions. We will then show that these bounds carry over to the limiting solution, which then restricts correctly to the (non-compactly supported) initial boundary data. A major obstacle in obtaining the necessary bounds will be proving decay for $\pu(r\phi)$, for which we will need to commute with the timelike generators of $\Gamma$. 
    The limiting argument itself proceeds via a careful Gr\"onwall-type argument on the differences of two solutions, thus establishing that the sequence is Cauchy. This method is then also used to infer the uniqueness of the limiting solution. Notice that the logarithmic term of~\eqref{eq:intro:thmtimelikeasymptotics} only appears in the limiting solution, whereas the actual sequence of solutions satisfies peeling. This can be understood already from the heuristic computation \eqref{eq:intro:heuristic}.
    
    We refer the reader to the introduction of section~\ref{sec:timelike} as well as Theorems~\ref{thm:timelike:final!!!} and~\ref{thm:timelike:logs} (which together contain Thm.~\ref{thm.intro:timelikecase}) for details. 
    
\subsection{An application: The scattering problem}\label{intro:scattering}
 \subsubsection{The scattering problem  on Minkowski, Schwarzschild and Reissner--Nordstr\"om}  
In the setting of data on an ingoing null hypersurface, the case $p=3$ is of independent interest in view of its natural appearance in the scattering problem "on" Minkowski or Schwarzschild (or Reissner--Nordstr\"om). 
If one puts compactly supported data for the scalar field $r\phi=G(v)$ on $\mathcal{I}^-$ and\footnote{Note that, since all our results only apply in a region sufficiently close to $\mathcal{I}^-$, it does not make a difference whether we consider compactly supported or vanishing data on $\mathcal{H}^-$.} on the past event horizon $\mathcal{H}^-$, it is not difficult to see that there exists an ingoing null hypersurface ${\mathcal{C}_{\mathrm{in}}}$, "intersecting" $\mathcal{I}^-$ to the future of the support of $r\phi|_{\mathcal{I}^-}$, on which eq.~\eqref{eq:intro:nullrayfields} generically holds with $p=3$ and such that eq.~\eqref{eq:intro:Hawking2} holds on $\mathcal{C}_{\mathrm{in}}\cap \mathcal{I}^-$. See Figure~\ref{fig:4}.
This puts us in the situation of Thm.~\ref{thm.intro:nullcase}.
%%%%%%%%%%%
\begin{figure}[htbp]
\floatbox[{\capbeside\thisfloatsetup{capbesideposition={right,top},capbesidewidth=4.4cm}}]{figure}[\FBwidth]
{\caption{The Penrose diagram of Schwarzschild. By Theorem~\ref{thm.intro:scattering}~\textbf{c)}, smooth compactly supported scattering data on $\mathcal H^-$ and $\mathcal I^-$ generically lead to the setup of Theorem~\ref{thm.intro:nullcase} with $p=3$. The region $\mathcal D$ as depicted corresponds to Figure~\ref{fig:3}. As a consequence, the solution fails to be conformally regular on $\mathcal I^+$. }\label{fig:4}}
{  
 \includegraphics[width = 190pt]{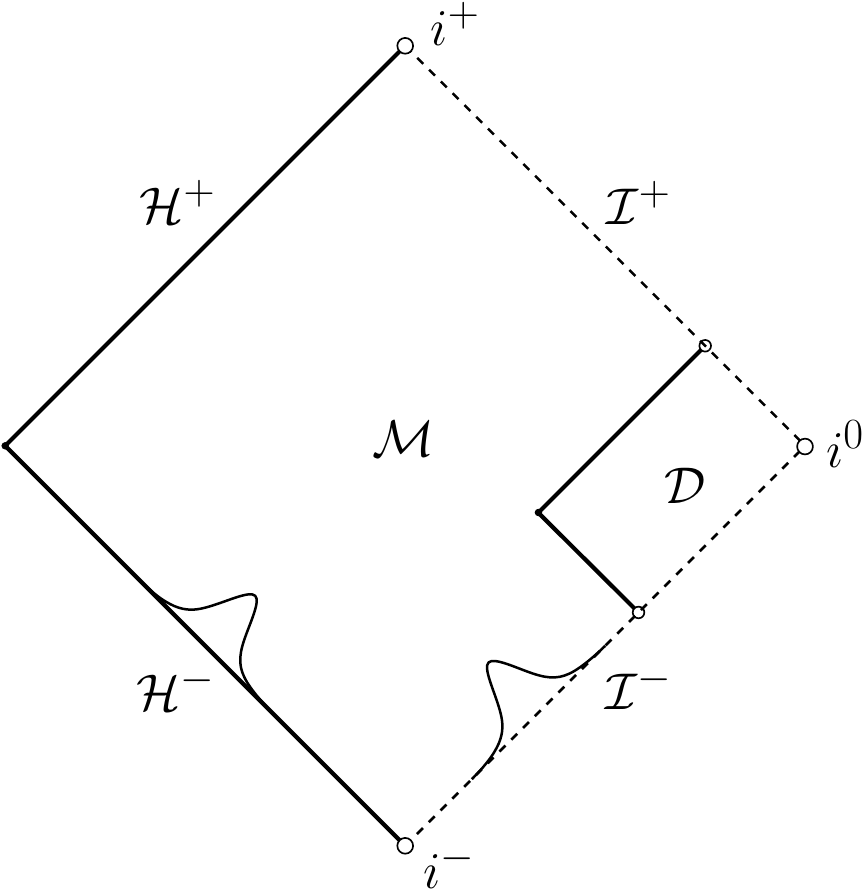}
}
\end{figure}
%%%%%%%%%%%
%\begin{figure}[htbp]
%  \centering
%  \includegraphics[width = 190pt]{SchwarzschildDoubleSupportDU.pdf}
%\end{figure} 

However, recall that we required $M$ from eq.~\eqref{eq:intro:Hawking2} to be strictly positive in order for the $\log$-terms in $\pv(r\phi)$ to be non-vanishing: 
Therefore, while $M$ is positive in both the coupled and the linear problem on Schwarzschild, \textit{one needs to consider the coupled problem} when considering the corresponding problem with a Minkowskian $i^-$ since one needs the scalar field to generate mass along $\mathcal{I}^-$.\footnote{In general, the spherically symmetric linear wave equation on a Minkowski background is not very exciting in view of the exact conservation law $\pu\pv(r\phi)=0$. Combined with the no incoming radiation condition, this conservation law would  force $\pv(r\phi)$ to vanish everywhere.}
\begin{rem}
Let us quickly explain our terminology: Since we only consider \textit{compactly supported} scattering data, the arising solutions will be identically vacuum in a neighbourhood of $i^-$. Depending on the setting, we then say that the arising spacetimes either have a \textit{Minkowskian} or a \textit{Schwarzschildean} (with mass $M>0$) $i^-$.
\end{rem}
%\begin{align} \label{eq:intro:nullrayfields2}
 %   r\phi|_{\mathcal{C}_{\mathrm{in}}}=M \int_{\mathcal{I}^-}G(v)\dd v|u|^{-2}.
%\end{align}

We therefore obtain the following result (see Thm.~\ref{null:thm:scattering} for the precise version):
  \begin{thm}\label{thm.intro:scattering}
        Consider either 
        
      \textbf{  a)} the non-linear scattering problem for the spherically symmetric Einstein-Scalar field system with a Schwarzschildean $i^-$  (with mass $M>0$), with vanishing data on $\mathcal{H}^-$ and with smooth compactly supported data $r\phi=G(v)$ on $\mathcal{I}^-$,
        
        or 
        
      \textbf{  b)} the non-linear scattering problem for the spherically symmetric Einstein-Scalar field system with a Minkowskian $i^-$, with smooth compactly supported data $r\phi=G(v)$ on $\mathcal{I}^-$,
      
      or
      
       \textbf{ c)} the linear scattering problem for the wave equation on a fixed Schwarzschild background with mass $M>0$,   with vanishing data on $\mathcal{H}^-$ and with smooth compactly supported data $r\phi=G(v)$ on $\mathcal{I}^-$.
     
       Then, a unique smooth semi-global solution exists (in fact, in case \textbf{c)}, this smooth solution exists globally in the exterior of Schwarzschild), and we get, along hypersurfaces of constant $u$, for sufficiently large negative values of $u$,  the following asymptotic expansion near $\mathcal{I}^+$:
        \begin{equation}
            \pv(r\phi)(u,v)=B(u)\frac{1}{r^3}+B'\frac{\log r-\log|u|}{r^4}+\mathcal{O}(r^{-4}),
        \end{equation}
        where $B'$ is a constant which, in each case, can  be explicitly computed from $G(v)$ and is generically non-zero.
    \end{thm}
\begin{rem}\label{rem:BV}
We note that, in the case \textit{\textbf{b)}}, if an additional smallness assumption on the data on $\mathcal I^-$ is made, then, in fact, the solution is causally geodesically complete,  globally regular, and has a complete $\mathcal I^+$. This follows from~\cite{ChristodoulouBV}. See also Theorem~1.7 of~\cite{LukOhYangBV} and the dichotomy of \cite{Luk2015QuantitativeSymmetry}.
\end{rem}
Theorem \ref{thm.intro:scattering} suggests that, in the context of the scattering problem, one should generically expect logarithmic terms to appear at the latest at second order in the asymptotic expansions of $\pv(r\phi)$ near $\mathcal{I}^+$. This is precisely what motivates our statement of Conjecture~\ref{conj2} in section~\ref{sec:intro:CHR}.

We can make an even stronger statement in the case of the linear wave equation on Schwarzschild.
There, the condition that $G$ needs to satisfy so that no logarithmic terms appear in the expansion of $\pv(r\phi)$ up to order $r^{-(4+n)}$ is that
\begin{equation}
    \int G(v) v^m \dd v=0
\end{equation}
for all $m\leq n$.
In particular, \textit{all} non-trivial smooth compactly supported scattering data on Schwarzschild lead to expansions of $\pv(r\phi)$ which eventually fail to be conformally smooth. See already Theorem~\ref{scatteringthm2}. We also refer the reader to~\cite{DRSR} for a general treatment of the scattering problem on Kerr.

%\begin{rem}
%Remark that the same result does not hold true for the corresponding linear scattering problem on Minkowski due to the absence of a mass term near $i^0$. However, in the non-linear scattering setup on Minkowski, the initial data on $\mathcal{I}^-$ will generate a mass term so that the above result holds with the mass term $M$ given by 
%$$2M=\int_{\mathcal{I}^-}G'(v)^2\dd v. $$
%\end{rem}

\subsubsection{The conformal isometry on extremal Reissner--Nordstr\"om}\label{sec:222}
Our results can also be applied to the linear wave equation on extremal Reissner--Nordstr\"om.\footnote{More generally, recall that our results also apply to the coupled Einstein-Maxwell-Scalar field system with a chargeless and massless scalar field and  can then -- \textit{a fortiori} -- be specialised to the linear setting (i.e.\ to the linear wave equation on Reissner--Nordstr\"om) as in Theorem \ref{thm.intro:linearcase}.}
In this setting, let us finally draw the reader's attention to the well-known conformal "mirror" isometry~\cite{CouchTorrence} on extremal Reissner--Nordstr\"om, which implies that all results on the radiation field are essentially invariant under interchange of
\begin{align*}
 u\longleftrightarrow v, &&\frac1r\longleftrightarrow (r-r_+),
\end{align*}
where $r_+$ is the value of $r$ at the event horizon. To make this more precise, we recall from~\cite{Lucietti2013} (see also~\cite{Bizo__2013}) that if $\phi$ is a solution to the linear wave equation in \textit{outgoing} Eddington--Finkelstein coordinates $(u,r)$, then, in \textit{ingoing} Eddington--Finkelstein coordinates $(v,r')$,
\begin{equation}
\tilde{\phi}(v,r'):=\frac{r_+}{r'-r_+}\phi\left(u=v,r=\frac{r_+r'}{r'-r_+}\right)=\frac{r-r_+}{r_+}\phi\left(u=v,r=\frac{r_+r'}{r'-r_+}\right)
\end{equation}
also is a solution to the wave equation, where, in the above definition, the LHS is evaluated in ingoing and the RHS in outgoing null coordinates. 
One can directly read off from this that regularity in $r'$ of $\tilde \phi$ near the future event horizon $\mathcal H^+$ is equivalent to regularity of $r\phi$ in the conformal variable $1/r$ near $\mathcal I^+$.
In other words, applying this conformal isometry to Theorems~\ref{thm.intro:timelikecase}--\ref{thm.intro:scattering}, which made statements on the conformal regularity of $r\phi$ near $\mathcal I^+$, now produces statements on the \emph{physical} regularity of $\tilde\phi$ near the event horizon. 

For instance, the mirrored version of Theorem~\ref{thm.intro:scattering} shows that smooth compactly supported scattering data on $\mathcal{H}^-$ and on $\mathcal{I}^-$ for the linear wave equation on extremal Reissner--Nordstr\"om generically lead to solutions $\phi$ which not only fail to be conformally smooth near $\mathcal I^+$, but also fail to be in $C^4$ near $\mathcal H^+$. (See also~\cite{AAGscattering} for a general scattering theory on extremal Reissner--Nordstr\"om.)
This is in stark contrast to the scattering problem on Schwarzschild, where, under the same setup, the solution remains smooth up to and including the future event horizon. 
One can relate this to the absence of a bifurcation sphere in extremal Reissner--Nordstr\"om (see also~\cite{Luebbe_2014}). 
Indeed, if one, instead of posing data on all of $\mathcal H^-$, poses compactly supported data on a null hypersurface which coincides with $\mathcal H^-$ up to some finite time and which, for sufficiently large $u$, becomes a timelike boundary intersecting $\mathcal H^+$ at some finite $v$, then the corresponding solution remains smooth.

We will not explore potential implications of this on Strong Cosmic Censorship in this paper (see however also~\cite{Gajic2017LinearI}, where the importance of logarithmic asymptotics for extendibility properties near the \underline{inner} Cauchy horizon of extremal Reissner--Nordstr\"om is discussed). 

\subsection{Translating asymptotics near \texorpdfstring{$i^0$}{i0} into asymptotics near \texorpdfstring{$i^+$}{i+}}
All the results presented so far hold true in a neighbourhood of $i^0$. 
In our companion paper~\cite{Kerrburger2}, we answer the question how the asymptotics for $\pv(r\phi)$ obtained near \textit{spacelike infinity} translate into asymptotics for $\phi$ near \textit{future timelike infinity}.  
In that work, we restrict to the analysis of the linear wave equation on a fixed Schwarzschild background and focus on the case $p=2$ of Theorem~\ref{thm.intro:timelikecase} (so $r\phi\sim |t|^{-1}$ on data). 
Smoothly extending the boundary data to the event horizon, we prove in~\cite{Kerrburger2} that the logarithmic asymptotics~\eqref{eq:intro:thmtimelikeasymptotics} imply that the leading-order asymptotics of $\phi$ on $\mathcal{H}^+$ and of $r\phi$ on $\mathcal{I}^+$ are also logarithmic and entirely determined by the constant $-2M\Phi^-$. For instance, we obtain that $r\phi|_{\mathcal I^+}=-2M\Phi^-u^{-2}\log u+\mathcal O(u^{-2})$ along $\mathcal I^+$ as $u\to\infty$.
 In particular, the leading-order asymptotics are independent of the extension of the data to (and towards) the horizon. 
This gives rise to \textit{a logarithmically modified Price's law} and, in principle, provides a tool to \textit{directly measure the non-smoothness of} $\mathcal{I}^+$.

 The paper~\cite{Kerrburger2} crucially uses methods and results from~\cite{Angelopoulos2018ASpacetimes},~\cite{Angelopoulos2018Late-timeSpacetimes}.

It would be an interesting problem to show a similar statement for the \textit{coupled} Einstein-Scalar field system considered in the present paper. See also the works~\cite{Dafermos2005} and~\cite{Luk2015QuantitativeSymmetry} in this context.
\subsection{Future directions}\label{sec:intro:future}
As this paper constitutes the first of a series of papers, we here outline some further directions which we will pursue in the future and which build on the present work.
\subsubsection{Going beyond spherical symmetry: Higher \texorpdfstring{$\ell$}{l}-modes}\label{sec:intro:higherl}
It is natural to ask what happens outside of spherical symmetry in the case of the linear wave equation on a fixed Schwarzschild background (as the coupled problem would be incomparably more difficult):
If one decomposes the solution to the wave equation by projecting onto spherical harmonics and works in double null Eddington--Finkelstein coordinates $(u,v)$, one gets the following generalisation of the spherically symmetric wave equation \eqref{eq:intro:waveequation}:
\begin{equation}\label{eq:intro:waveequationhigherl}
    \pu\pv(r\phi_\ell)=-\ell(\ell+1)\left(1-\frac{2M}{r}\right)\frac{r\phi_\ell}{r^2}+   2M\frac{\pu r \pv r}{1-\frac{2M}{r}}\frac{r\phi_\ell}{r^3},
\end{equation}
where $\phi_\ell$ is the projection onto the $\ell$-th spherical harmonic. 
The difference from the spherically symmetric case treated so far is obvious: 
The RHS decays slower for $\ell\neq0$. 
Since the good $r^{-3}$-weight for $\ell=0$ plays a crucial rule in the proofs of all theorems in the present paper, one might think that this renders the methods of this paper useless for higher $\ell$-modes. 
However, one can recover the good $r^{-3}$-weight by commuting $\ell$ times with vector fields which, in Eddington--Finkelstein coordinates, to leading order all look like $r^2\pv$.\footnote{Recall that on the Minkowski background, i.e.\ for $M=0$, the strong Huygens' principle manifests itself in form of the conservation law $\pu(r^{-2\ell-2}(r^2\pv)^{\ell+1}(r\phi_\ell))=0$ for the $\ell$-th spherically harmonic mode.} 
Using these commuted wave equations, one can then adapt the methods of this paper to obtain similar results for higher $\ell$-modes, with logarithms appearing in the expansions of $\pv(r\phi_\ell)$ at orders which depend in a more subtle way on the precise setup.
 We note that these commuted wave equations, which we will dub \textit{approximate conservation laws}, are closely related to the higher-order Newman--Penrose quantities for the scalar wave equation (see also the introduction of~\cite{Kerrburger2} or the recent~\cite{Dejantobepublished}). 

We will dedicate an upcoming paper to the discussion of higher $\ell$-modes~\cite{Kerrburger3}. Similarly to~\cite{Kerrburger2}, we will also discuss the issue of late-time asymptotics in~\cite{Kerrburger3}. 
It will turn out that, in certain physically reasonable scenarios (such as the scattering problem of Theorem \ref{thm.intro:scattering}), the usual expectation that higher $\ell$-modes decay faster towards $i^+$ is partially violated. 
\subsubsection{The wave equation on a fixed Kerr background} Similarly, it would be interesting to understand how the results obtained in this paper would differ if one were to consider the linear wave equation on a fixed Kerr background. See also the recent~\cite{DejanKerr} and~\cite{Hintz}, where a generalisation of the well-known Price's law is obtained for Kerr backgrounds. 
\subsubsection{Going from scalar to tensorial waves: The Teukolsky equations}
Once the behaviour of higher $\ell$-modes is understood in~\cite{Kerrburger3}, the natural next step towards a resolution of Conjectures~\ref{conj1} and~\ref{conj2} would be an analysis of the Teukolsky equations of linearised gravity (e.g.\ in the context of the scattering problem of the recent~\cite{Hamedalpha}). We believe that an understanding of the approximate conservation laws associated to the Newman--Penrose constants of the Teukolsky equations will play a crucial role here, similarly to~\cite{Kerrburger3}.
\subsubsection{A resolution of Conjectures~\ref{conj1},~\ref{conj2}}
In turn, once a detailed understanding of the Teukolsky equations is obtained, we will attempt to resolve Conjectures~\ref{conj1},~\ref{conj2} for the Einstein vacuum equations themselves. This will, in particular, require a detailed understanding of the scattering problem for the Einstein vacuum equations with a Minkowskian $i^-$, which we hope to obtain in the not too distant future.
In the context of resolving the above conjectures, we will also give a detailed explanation and enhancement of Christodoulou's argument~\cite{CHRISTODOULOU2002}, in which we hope to obtain \eqref{eq:intro:Xi-} \emph{dynamically}.

Once this program is completed, one could finally attempt to tackle the actual $N$-body problem in a fully general relativistic setting, i.e.\ one could attempt to obtain a result similar to Thm.~\ref{thm.intro:timelikecase}. Let us however not yet speculate how this would look like. For now, we hope that it suffices to say that one of the most interesting aspects of such a problem would be a rigorous justification for the quadrupole approximation, arguably one of the most important tools in general relativity.

\subsection{Structure of the paper}
The remainder of this paper (corresponding to the "Counter-Examples" part of the title) is structured as follows: 
 We first reduce the spherically symmetric Einstein-Scalar field system to a system of first-order equations and set up the notation that we shall henceforth work with in section~\ref{sec:system of equations}. 
 We sketch the specialisation to the linear case, i.e.\ to the case of the wave equation on a fixed Schwarzschild background, in section~\ref{sec:linear}.
We construct characteristic initial data as outlined above in section~\ref{sec:null} and prove Theorem~\ref{thm.intro:nullcase} in section
\ref{sec:nul:asymptotics}.
We deal with the problem of timelike boundary data in section~\ref{sec:timelike} and prove Theorem~\ref{thm.intro:timelikecase} in section~\ref{sec:refine:limitingargumentover}. Section~\ref{sec:timelike} can, in principle, be read independently of section~\ref{sec:null}, though we recommend reading it after section~\ref{sec:null}.

The scattering results and, in particular, Theorem~\ref{thm.intro:scattering}  are proved in section~\ref{sec:scattering}. This section can be read immediately after section~\ref{sec:null}.

\section*{Acknowledgements}
\addcontentsline{toc}{section}{Acknowledgements}
The author would like to express their gratitude to their supervisor Mihalis Dafermos for suggesting the problem and for many discussions on drafts of the present work. The author would also like to express their gratitude towards the Analysis and Relativity groups at Cambridge and Princeton University, and, in particular, towards John Anderson, Dejan Gajic, Christoph Kehle and Claude Warnick for various helpful discussions and/or comments on drafts of the present work. 
%\paragraph{Real Acknowledgements}
%The author would like to express their gratitude to their supervisor Mihalis Dafermos for \textit{repeatedly} drawing the author's attention to a spelling mistake in the food consumed by the author while the author was writing the original transcript.
\newpage
\part{Construction of spherically symmetric counter-examples to the smoothness of \texorpdfstring{$\mathcal{I}^+$}{I+} }\label{sec:counterexamples}
%\section{Construction of spherically symmetric counter-examples to the smoothness of $\mathcal{I}^+$ }\label{sec:counterexamples}
In this part of the paper, we will construct two classes of initial data that have a non-smooth future null infinity in the sense that the outgoing derivative of the radiation field $\pv(r\phi)$ has an asymptotic expansion near $\mathcal{I}^+$ that contains logarithmic terms at leading order. 
These examples will be for the spherically symmetric Einstein-Maxwell\footnote{We emphasise that the presence of a Maxwell field is not important to most results, we mainly include it in view of the remarks on the scattering problem on extremal Reissner--Nordstr\"om made in the introduction.}-Scalar field system, with no incoming radiation from $\mathcal{I}^-$ and polynomially decaying initial/boundary data on an ingoing null hypersurface or a timelike hypersurface, respectively. They are motivated by Christodoulou's argument against smooth null infinity, see the introductory remarks in section~\ref{sec:introduction:scalar}.

This part of the paper is structured as follows: 

We first reduce the spherically symmetric Einstein-Maxwell-Scalar field system to a system of first-order equations in section~\ref{sec:system of equations}.

We then construct counter-examples to the smoothness of null infinity that have polynomially decaying data on an ingoing null hypersurface in section~\ref{sec:null}. 

In section~\ref{sec:timelike}, we construct counter-examples with polynomially decaying data on a general timelike hypersurface (e.g.\  on a hypersurface of constant area radius). This latter case will be strictly more difficult than the former, so we advise the reader to first understand the former. Nevertheless, each of the sections can be understood independently of the respective other one.

Our constructions will be fully general relativistic, however, we remark that the non-smoothness of null infinity can already be observed in the linear setting, which we present in sections~\ref{sec:Schwarzschild} and~\ref{sec:linear}. 

We finally discuss implications of our results on the scattering problem on Schwarzschild; in particular, we find that it is essentially impossible for solutions to remain conformally smooth near $\mathcal{I}^+$ if they come from compactly supported scattering data. This is discussed in section~\ref{sec:scattering}. The reader can skip to this section immediately after having read section~\ref{sec:null}.

More detailed overviews will be given at the beginning of each section.

\newpage
\section{The Einstein-Maxwell-Scalar field equations in spherical symmetry}\label{sec:system of equations}
In this section, we introduce the systems of equations that are considered in this paper. We write down the spherically symmetric Einstein-Maxwell-Scalar equations in double null coordinates and transform them into a particularly convenient system of first-order equations in section~\ref{sec:subsec:system of equations}. We then briefly introduce the Reissner--Nordstr\"om family of solutions and discuss the linear setting in sections~\ref{sec:Schwarzschild} and~\ref{sec:linear}.
\subsection{The coupled case}\label{sec:subsec:system of equations}
Throughout this section, we will use the convention that upper case Latin letters denote coordinates on the sphere, whereas lower case Latin letters denote "downstairs"-coordinates. For general spacetime coordinates, we will use Greek letters.

In any double null coordinate system $(u,v)$, the Einstein equations 
    \begin{equation}
        R_{\mu\nu}-\frac 1 2 R g_{\mu\nu}=2T_{\mu\nu}\label{eq:sys:einstein}
    \end{equation}
in spherical symmetry (see~\cite{Dafermos2003} and section 3 of~\cite{Christodoulou1995} for details on the notion of spherical symmetry in this context) can be re-expressed into the following system of equations for the metric 
    \begin{equation}
        g=-\Omega ^2 \dd u\dd v + r^2\,\gamma ,\label{eq:metricindoublenull}
    \end{equation}
where $\gamma$ is the metric on the unit sphere $\mathbb S^2$, $r$ is the area radius function, $\Omega$ is a positive function, and where we assume that $r, \Omega$ are $C^2$:
    \begin{align}
         \pu\pv r &=-\frac{\Omega^2}{4r}\left(1+4\frac{\pv r \pu r}{\Omega^2}\right)+r T_{uv}  ,\label{eq:sys:dudvr}\\
        \pu\pv \log \Omega &= \frac{\Omega^2}{4r^2}\left(1+4\frac{\pv r \pu r}{\Omega^2}\right) - T_{uv} -\frac{\Omega^2}{4 }g^{AB}T_{AB} , \label{eq:sys:dudvlogOmega}\\
        \pu(\Omega^{-2}\pu r) &= -r \Omega^{-2} T_{uu} ,\label{eq:sys:dudur}\\
        \pv(\Omega^{-2}\pv r) &= -r \Omega^{-2} T_{vv}. \label{eq:sys:dvdvr}
    \end{align}
The matter system considered in this paper is represented by the sum of the following two energy momentum tensors:
    \begin{align}
        T^{sf}_{\mu\nu}&= \phi_{;\mu} \phi_{;\nu} - \frac 1 2 g_{\mu\nu}\phi^{;\xi} \phi_{;\xi} ,\\
        T^{em}_{\mu\nu}&=F_{\mu\xi}F^{\xi}{}_{\nu}-\frac 1 4 g_{\mu\nu} F_{\xi o}F^{\xi o}.
    \end{align}
    These are in turn governed by the wave equation and the Maxwell equations, respectively, which can compactly be written as $\nabla^\mu T^{sf}_{\mu\nu}=0=\nabla^\mu F_{\mu\nu}=\nabla^\mu {}^\ast F_{\mu\nu} $.
    
One can show that, in spherical symmetry (assuming no magnetic monopoles\footnote{This is an evolutionary consistent assumption. This is to say that if we exclude magnetic monopoles on data,  then they cannot arise dynamically. More mathematically, this is the statement that if $F_{AB}=0$ on data, then $F_{AB}=0$ everywhere.}), the electromagnetic contribution decouples and can be computed in terms of a constant $e^2$ (the electric charge) and  $r$: 
    \begin{align}
         T^{em}_{ab}=-\frac{e^2}{2r^4}g_{ab},&& T^{em}_{AB}=\frac{e^2}{2r^4}g_{AB}.
    \end{align}
For more details, see~\cite{Dafermos2003}. 
On the other hand, for the scalar field,  one computes directly  
    \begin{align}
        T^{sf}_{uu}=(\pu \phi)^2,&&T^{sf}_{vv}=(\pv \phi)^2 ,&&      T^{sf}_{uv}=0,&&g^{AB}T^{sf}_{AB}=4\Omega^{-2}\pu\phi\pv\phi.
    \end{align}
In particular,  equations \eqref{eq:sys:dudvr}, \eqref{eq:sys:dudvlogOmega} now read:
    \begin{align}
        \pu\pv r &=-\frac{\Omega^2}{4r}\left(1+4\frac{\pv r \pu r}{\Omega^2}\right) -\frac{\Omega^2}{4r} \frac{e^2}{r^2},\label{eq:sys:dudvr2} \\
            \pu\pv \log \Omega &= \frac{\Omega^2}{4r^2}\left(1+4\frac{\pv r \pu r}{\Omega^2}\right) - \frac{e^2\Omega^2}{2r^4}-\pu\phi\pv\phi. \label{eq:sys:dudvlogOmega2}
    \end{align}
Moreover, one derives the following wave equation for the scalar field from $\nabla^\mu T^{sf}_{\mu\nu}=0$:\footnote{Note that if $T^{em}_{\mu\nu}=0$, this would be a consequence of \eqref{eq:sys:einstein} and the Bianchi identities.}
    \begin{align}
        r \pu\pv\phi+\pu r \pv \phi+\pv r \pu\phi=0 .\label{eq:sys:dudvphi}
    \end{align}

We can transform this second-order system into a system of first-order equations by introducing the \textit{renormalised Hawking mass}:
    \begin{equation}
        \varpi:=m+\frac{e^2}{2r}:=\frac r 2 (1-g'(\nabla r,\nabla r))+\frac{e^2}{2r} \label{eq:sys:varpi},
    \end{equation}
where $g'$ is the projected metric $g'=-\Omega^2\dd u\dd v$ and $m$ denotes the \textit{Hawking mass}. In the remainder of the paper, we shall write $g$ instead of $g'$.
As we shall see, the renormalised Hawking mass obeys important monotonicity properties and will essentially allow us to do energy (i.e.\  $L^2$-) estimates, which will usually form the starting point for our estimates, which will otherwise be $L^1$- or $L^\infty$-based. 

Let us now recall the notation introduced by Christodoulou:
    \begin{align}\pu r=\nu, &&\pv r =\lambda\end{align}
and
     \begin{align}r\pu \phi=\zeta,&&r\pv\phi=\theta.\end{align}
Moreover, we write
    \begin{align}\mu:=\frac{2m}{r},&&\kappa:=\frac{\lambda}{1-\mu}=-\frac14\Omega^2\nu^{-1},\end{align}
where the last equality comes from the definition of $m$. 
It is then straightforward to derive equivalence between the system of second-order equations \eqref{eq:sys:dudur}, \eqref{eq:sys:dvdvr}, \eqref{eq:sys:dudvr2}--\eqref{eq:sys:dudvphi} and the following system of first-order equations:\footnote{Notice already that, e.g.,  by controlling $\varpi$ in the $u$-direction, eq.~\eqref{eq:puvarpi} gives us an $L^2$-estimate for $\pu\phi$, assuming sufficient control over $\nu$ and $1-\mu$.}
    \begin{align}
        \pu\varpi&=\frac12(1-\mu)\frac{\zeta^2}{\nu}, \label{eq:puvarpi}\\
        \pv\varpi&=\frac12 \frac{\theta^2}{\kappa},\label{eq:pvvarpi}\\
        \pu\kappa&=\frac{1}{r}\frac{\zeta^2}{\nu}\kappa,\label{eq:pukappa}\\
        \pu\theta&=-\frac{\zeta\lambda}{r},\label{eq:putheta}\\
        \pv \zeta &= -\frac{\theta\nu}{r}.\label{eq:pvzeta}
    \end{align}
From these equations, one derives the following two useful wave equations for $r$ and the radiation field $r\phi$:
    \begin{align}
        \pv\nu=\pu\lambda=\pu\pv r=\frac{2\nu\kappa}{r^2}\left(\varpi-\frac{e^2}{r}\right), \label{eq:pupvr}\\
        \pu\pv(r\phi)=2\nu\kappa\left(\varpi-\frac{e^2}{r}\right)\frac{r\phi}{r^3}. \label{eq:wave}
    \end{align}
In the sequel, we shall mostly work with equations \eqref{eq:puvarpi}--\eqref{eq:wave}.
\subsection{The Reissner--Nordstr\"om/Schwarzschild family of solutions}\label{sec:Schwarzschild}
If one sets $\phi$ to vanish identically in the system of equations \eqref{eq:puvarpi}--\eqref{eq:pvzeta}, then, by (a generalisation of) Birkhoff's theorem -- which essentially follows from equations \eqref{eq:puvarpi}, \eqref{eq:pvvarpi} --  all \textit{asymptotically flat} solutions\footnote{In contrast to the $e^2=0$-case, there exist spherically symmetric solutions to the Einstein-Maxwell equations which are not asymptotically flat.} belong to the well-known Reissner--Nordstr\"om family of solutions, which contains as a subfamily the Schwarzschild family  (corresponding to the case where also $e^2=0$).

Let us, for the moment, go back to the four-dimensional picture and restrict to the physical parameter range $M\geq 0$, $|e|\leq M$ ($|e|=M$ corresponding to the \textit{extremal} case). Then, the exteriors of this family of spacetimes are given by the family of Lorentzian manifolds $(\mathcal M_{M,e},g_{M,e})$, with $$\mathcal M_{M,e}=\mathbb{R}\times (M+\sqrt{M^2-e^2},\infty)\times \mathbb S^2$$ covered by the coordinate chart $(t,r,\vartheta,\varphi)$, where $t\in\mathbb R$, $r\in(M+\sqrt{M^2-e^2},\infty)$, and where $\vartheta$, $\varphi$  are the standard coordinates on the sphere, and with $g_{M,e}$ given in these coordinates by
\begin{equation}
g_{M,e}=-D(r)\dd t^2+\frac{1}{D(r)}\dd r^2+r^2\(\dd\vartheta^2+\sin^2\vartheta \dd \varphi^2\).
\end{equation}
Here, $D(r)$ is given by $D(r)=1-\frac{2M}{r}+\frac{e^2}{r^2}$. By introducing the tortoise coordinate 
\begin{equation}
r^*(r):=R+\int_{R}^r D^{-1}(r')\dd r'
\end{equation}
for some $R> M+\sqrt{M^2-e^2}$ and further introducing the (Eddington--Finkelstein) coordinates $2u=t-r^*(r)$, $2v=t+r^*(r)$, one can bring the metric into the double null form \eqref{eq:metricindoublenull} with $\Omega^2=4D(r)$. One then has $\varpi\equiv M$ and $\lambda=-\nu=D(r)$.
\subsection{Specialising to the linear case}\label{sec:linear}
We claimed in the introduction that the results that we will obtain for the coupled Einstein-Maxwell-Scalar field system can also be applied to the linear case, i.e.\ to the case of the wave equation \eqref{eq:wave} on a fixed Reissner--Nordstr\"om background. 

In that case, the right-hand sides of eqns.~\eqref{eq:puvarpi}--\eqref{eq:pukappa} are replaced by zero, whereas the remaining equations remain unchanged, with $\varpi\equiv M$ a constant.
This severely simplifies most proofs in the present paper. 
However, there is one ingredient that seems to be lost at first sight: the energy estimates (see \eqref{eq:ebinvdirection}, \eqref{eq:ebinudirection})!
These are, for instance, used for obtaining preliminary decay, $|\phi|\lesssim r^{-1/2}$, for the scalar field in \eqref{usageofenergyestimates}. 
However, in the linear case, one can obtain these very estimates  \eqref{eq:ebinvdirection}, \eqref{eq:ebinudirection} by an application of the divergence theorem to $\nabla^\mu (T_{\mu\nu}^{sf}\boldsymbol{K}^\nu)$ in a null rectangle, where $\boldsymbol{K}$ is the static Killing vector field of the Reissner--Nordstr\"om metric (given by $\partial_t$ in $(t,r)$-coordinates). 
In fact, the divergence theorem implies that the 1-form (for details, see  section~11 of~\cite{Dafermos2005})
\begin{equation}
\eta:=\frac{1}{2}(1-\mu)\frac{\zeta^2}{\nu}\dd u+\frac{1}{2}\frac{\theta^2}{\kappa}\dd v
\end{equation}
is closed, $\dd \eta=0$, and one can thus define a 0-form $\varpi'$ via $\dd \varpi'=\eta$ and by demanding that $\varpi'=M$ on past null infinity. The quantity $\varpi'$ then obeys the exact same equations as $\varpi$ does in the coupled case. This means that one can repeat all the estimates of the present paper, \textit{mutatis mutandis,} in the uncoupled case. In particular, once we show Theorem~\ref{thm.intro:timelikecase} from Part~\ref{sec:part1}, Theorem~\ref{thm.intro:linearcase} will follow \textit{a fortiori}.

%However, estimates like \eqref{eq:ebinvdirection}, \eqref{eq:ebinudirection} can easily be recovered by using the static Killing vector field $\boldsymbol{K}$ of the Reissner--Nordstr\"om metric (given by $\partial_t$ in $(t,r)$-coordinates) and applying the divergence theorem to $T_{\mu\nu}^{sf}\boldsymbol{K}^\nu$ in a null rectangle. Indeed, this gives rise to an energy $\varpi'$ that obeys the same equations as $\varpi$ does in the coupled case. This means that all the following estimates can be repeated \textit{mutatis mutandis} in the uncoupled case. In fact, there is only one proof which is simplified significantly in the linear setting, and that is the proof of Thm.~\ref{thm:refinements:TR}.  More details on converting to the linear case can be found in section 11 of~\cite{Dafermos2005}. 
%
%
%On the other hand, we will also make statements (for example Theorem~\ref{scatteringthm2}) which only hold in the linear case.
\subsection{Conventions}
In the remainder of the paper, we shall typically consider functions defined on some set $\mathcal D$. We then write $f\sim g$ if there exist uniform positive constants $A,B$ such that $Af\leq g\leq Bf$ on $\mathcal D$. Similarly, we write $f=\mathcal O(g)$ if there exists a uniform constant $A>0$ such that $|f|\leq Ag$. Occasionally, we shall write that $f\sim g$ on some subset of $\mathcal D$. In this case, the constants $A,B$ may also depend on the subset. Similarly for $f=\mathcal O(g)$. 

\newpage
\section{Case 1: Initial data posed on an ingoing null hypersurface}\label{sec:null}

In this section, we consider the semi-global characteristic initial value problem with polynomially decaying data on an ingoing null hypersurface and no incoming radiation from past null infinity to the future of that null hypersurface. 

As the case of initial data on an ingoing null hypersurface is significantly simpler than that with boundary data on a timelike hypersurface presented in section~\ref{sec:timelike}, and since, in particular, the relevant local existence theory is well-known, we will only present \textit{a priori estimates} in this section, i.e., we will assume that a sufficiently regular solution that restricts correctly to the initial data, and that "possesses" past and future null infinity as well as no anti-trapped or trapped surfaces, exists, and then show the relevant estimates on this assumed solution. 
We hope that this will allow the reader to more easily develop a tentative understanding of the main argument. The left out details of the proof of existence will then be dealt with in section~\ref{sec:timelike}.
    
We shall first explicitly state our assumptions in section~\ref{sec:null:assumptions}. 
The middle part of the section will be devoted to showing that the geometric quantities $\nu,\lambda,\kappa,\varpi$ etc.\ remain bounded for large enough negative values of $u$ in section~\ref{sec:null:energyboundedness}.
 We then use the wave equation \eqref{eq:wave} to derive sharp decay rates for the scalar field and its derivatives in section~\ref{sec:null:purphi}.
 Equipped with these sharp rates, we can then upgrade all the previous estimates on $\nu,\lambda$ etc.\ to asymptotic estimates. 
 This will finally allow us to obtain an asymptotic expansion of $\pv(r\phi)$ near future null infinity in section~\ref{sec:nul:asymptotics}. This last section is thus also the section where Thm.~\ref{thm.intro:nullcase} is proved (see Thm.~\ref{thm:null:asymptotics of dvrphi}).

 \subsection{Assumptions and initial data}\label{sec:null:assumptions}
 \subsubsection{Global \textit{a priori} assumptions}
 Let $\mathbb{R}^2$ denote the standard plane, and call its double null coordinates $(u,v)$. Fix a constant $M>0$, and assume
  that we have a rectangle (see Figure~\ref{fig:5} below)
 \begin{equation}\label{eq:null:DU}
     \mathcal{D}_{U}:=(-\infty,U]\times [1,\infty)\subset{\mathbb{R}^2}
 \end{equation}
 with $U<-2M$, and denote, for $u\in (-\infty,U]$, the sets $\mathcal{C}_u:=\{u\}\times[1,\infty)$  as \textit{outgoing null rays} and, for $v\in [1,\infty)$, the sets $\mathcal{C}_v:=(-\infty,U]\times \{v\}$ as \textit{ingoing null rays}. 
 We furthermore write $\mathcal{C}_{v=1}:=\mathcal{C}_{\mathrm{in}}$, and we \textit{colloquially refer to} $\{-\infty\}\times[1,\infty)$ as $\mathcal{I}^-$ or \textit{past null infinity}, to $(-\infty,U]\times\{\infty\}$ as $\mathcal{I}^+$ or \textit{future null infinity}, and to $\{-\infty\}\times\{\infty\}$ as $i^0$ or \textit{spacelike infinity}.
 
 On this rectangle $\mathcal{D}_{U}$, we assume that a strictly positive $C^3$-function $r(u,v)$, a non-negative $C^2$-function $ m(u,v)$, a $C^2$-function $\phi(u,v)$ and a constant $e^2>0$ are defined and obey the following properties:

The function $r$ is such that, along each of the ingoing and outgoing null rays, it tends to infinity, i.e., $\sup_{C_u}r(u,v)=\infty$ for all $u\in(-\infty,U]$, and $\sup_{C_v}r(u,v)=\infty$ for all $v\in[1,\infty)$.
 We moreover assume that, throughout $\mathcal{D}_{U}$,
 \begin{align}
    \pu r=	\nu<0\label{eq:nu-},\\
    \pv r=	\lambda>0\label{eq:lambda+},
\end{align}
 that $\nu=-1$ along $\mathcal{C}_{\mathrm{in}}$, and that $r(U,1)=r_1=-U>0$. We also assume that $\lim_{u\to-\infty}\lambda(u,v)=1$ for all $v\in[1,\infty).$
 
 Concerning $m$, we assume that  
 \begin{equation}\label{eq:null:ass:kappapositive}
     \frac{\lambda}{1-\mu}=\kappa>0
 \end{equation}
 is a strictly positive quantity and that 
 \begin{equation}
     \lim_{u\to -\infty}m(u,v)=M>0\label{eq:null:Hawking}
 \end{equation}
 for all $v\in[1,\infty)$.
 
 On the function $\phi$, we make the assumptions that, along $\mathcal{C}_{\mathrm{in}}$, it obeys
 \begin{equation}\label{eq:nullcase:ass:limit}
    r^p\frac{\pu (r\phi)}{\nu}+\frac{\Phi^-}{p-1}=\mathcal{O}(r^{-\epsilon}).
\end{equation}
for some constants $\Phi^-\neq 0$, $p>1$ and $\epsilon\in(0,1)$, and that
  \begin{align}
    \lim_{u\to-\infty}	 r\phi(u,v)=0=\lim_{u\to-\infty}\pv(r\phi)(u,v) \label{eq:null:idI-}
\end{align}
for all  $v\in[1,\infty)$.

Finally, we assume that, throughout $\mathcal{D}_{U}$, equations \eqref{eq:puvarpi}--\eqref{eq:pvzeta} hold \textit{pointwise}.
 %%%%%%%%%%%
\begin{figure}[htbp]
\floatbox[{\capbeside\thisfloatsetup{capbesideposition={right,top},capbesidewidth=4cm}}]{figure}[\FBwidth]
{\caption{The Penrose diagram of $\mathcal D_{U}$. It contains no black or white holes and, correspondingly, no trapped or anti-trapped surfaces (cf.\ \eqref{eq:nu-}, \eqref{eq:lambda+}). See also~\cite{Dafermos2005b} for an explanation of these notions.}\label{fig:5}}
{  
 \includegraphics[width = 155pt]{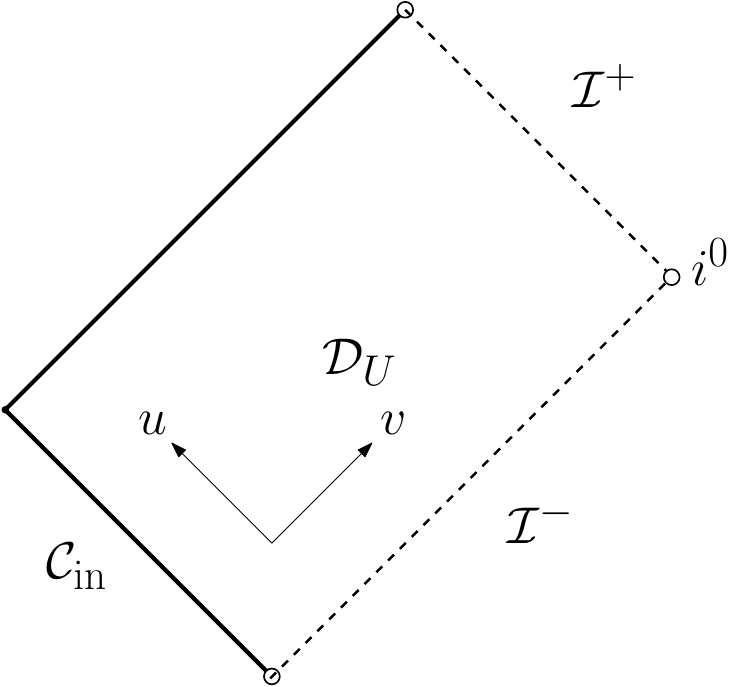}
}
\end{figure}
%%%%%%%%%%%
% \begin{figure}[htbp]
%  \centering
%  \includegraphics[width = 160pt]{PenroseDiagramNullCaseDU0.pdf}
%\end{figure}

The reader familiar with Penrose diagrams may refer to the Penrose diagram above (Figure~\ref{fig:5}), where the geometric content of these assumptions is summarised. The reader unfamiliar with Penrose diagrams may either ignore this remark of refer to the appendix of~\cite{Dafermos2005} for a gentle introduction to Penrose diagrams.

\subsubsection{Retrieving the assumptions}
By essentially considering solutions to the spherically symmetric Einstein-Maxwell-Scalar field system with characteristic initial data which satisfy $\nu=-1$ as well as \eqref{eq:nullcase:ass:limit} on an ingoing null hypersurface $\mathcal{C}_{\mathrm{in}}$, and which satisfy $\lambda=1$, \eqref{eq:null:Hawking} and \eqref{eq:null:idI-} on $\mathcal{I}^-$ (and by a limiting argument), we will, in section~\ref{sec:timelike} (cf.\ Thm.~\ref{thm:timelike:final}), prove the following:
\begin{prop}\label{prop:exis4}
Given a set $\mathcal D_U$ as in \eqref{eq:null:DU}, there exists a unique triplet of functions $(r,\phi,m)$ such that the above assumptions are satisfied, with the uniqueness being understood in the sense  of Remark~\ref{rem:uniqueness}.
\end{prop}
%In this case, $\mathcal{D}_{U}$ corresponds to the quotient of the solution under the action of $SO(3)$,
The metric associated to this solution is then given by~\eqref{eq:metricindoublenull}, with $ -\Omega^2=4\nu\kappa$. 
 
%This is also the reason for assumption \eqref{eq:null:ass:kappapositive}. The assumptions \eqref{eq:nu-}, \eqref{eq:lambda+} can then be interpreted as the statement that $\mathcal{D}_{U}$ is contained in the regular region of the spacetime (so it contains no black holes or white holes, see also~\cite{Dafermos2005b}).
%
%One can easily show local existence for this system, which, combined with the uniform estimates that will be derived in this section, immediately yields global-in-$v$ existence. One can then push the finite outgoing null ray to past null infinity by a simple limiting argument, or by density considerations. The limiting quotient of this limiting solution then corresponds to all of $\mathcal{D}_{U}$. 

\begin{rem}\label{remark:null:bondimass2}
In order to see why condition \eqref{eq:null:idI-} is the correct interpretation of the no incoming radiation condition, we recall from section~\ref{sec:intro:CHR} that the statement of no incoming radiation should be interpreted as the Bondi mass along past null infinity being a conserved quantity. In spherical symmetry, the definition of the Bondi mass as limit of the Hawking mass $m$ (or $\varpi$) is straight-forward. The analogue to the Bondi mass loss formula \eqref{eq:Argument:Bondimass} then becomes \eqref{eq:puvarpi}, or, formulated with respect to the past, \eqref{eq:pvvarpi}.
We thus see that the analogue to $\Xi$ is given by $\lim_{\mathcal I^+}\zeta$ (or by $\lim_{\mathcal I^-}\theta$ in the past).
\end{rem}

\subsection{Coordinates and energy boundedness}\label{sec:null:energyboundedness}

%\subsubsection*{The $u$-coordinate}
Note that the following consistency calculation
$$r(u,1)-r_1=-\int_{u}^{U}\nu \dd u'=U-u$$ 
confirms that, as $\mathcal{I}^-$  is approached along $v=1$,\footnote{In the sequel, we will show that $\nu\sim -1$ and $\lambda\sim 1$ throughout $\mathcal D_{U_0}$, so one can do a similar consistency calculation for any $\mathcal C_v$ or $\mathcal C_u$.} $r$ tends to infinity. 
Moreover, one sees from this equation that, with this choice of $u$-coordinate, $u=-r$ along $\mathcal{C}_{\mathrm{in}}$. 

In the remainder of the section, we will want to restrict to sufficiently large negative values of $u$ in order to be able to make asymptotic statements. Therefore, we introduce the set
\begin{equation}\label{eq:DutoDU0}
\mathcal D_{U_0}:=\mathcal D_{U}\cap\{u\leq U_0\}
\end{equation}
for some sufficiently large negative constant $U_0$ whose choice will only depend on $M$, $e^2$, $\Phi^-$, $p$, $\epsilon$ and the implicit constant in the RHS of \eqref{eq:nullcase:ass:limit}. % Similarly, we will introduce the notation $\mathcal C_{\mathrm{in},U_0}:=\mathcal C_{\mathrm{in}}\cap\{u\leq U_0\}$. 
Our first restriction on $U_0$ will be the following: Since $|u|=r$  along $\mathcal{C}_{\mathrm{in}}$, we shall from now on assume that
    \begin{equation}\label{eq:nullcase:ass:limit2}
    \frac{|\Phi^-|}{2(p-1)}|u|^{-p}\leq|\pu(r\phi)|\leq 2\frac{|\Phi^-|}{(p-1)} |u|^{-p}
    \end{equation}
   % \[d'_{\mathrm{in}}|u|^{-p+1}\leq|r\phi|\leq C'_{\mathrm{in}} |u|^{-p+1}\]
 along $\mathcal{C}_{\mathrm{in}}\cap\{u\leq U_0\}$. In view of  assumption \eqref{eq:nullcase:ass:limit}, this indeed holds for sufficiently large values of $U_0$.

In a first step, we will now prove energy boundedness along $\mathcal{C}_{\mathrm{in}}$, i.e.\ bounds on $m$ and $\varpi$. 
 
\begin{prop}[Energy boundedness on $\mathcal C_{\mathrm{in}}\cap\{u\leq U_0\}$]
    For sufficiently large negative values of $U_0$, we have along $\mathcal{C}_{\mathrm{in}}\cap\{u\leq U_0\}$, where  $\mathcal{C}_{\mathrm{in}}$ is  as described in the assumptions of section~\ref{sec:null:assumptions}:
        \begin{equation}
            \frac{M}{2}<m,\varpi< 2M .\label{eq:null:boundednessofvarpialongCin}
        \end{equation}  
\end{prop}
\begin{proof}
 We recall the transport equation for $\varpi$ \eqref{eq:puvarpi}:
    \[\pu\varpi=\frac12\(1+\frac{e^2}{r^2}\)\frac{\zeta^2}{\nu} -\frac{\varpi}{r}\frac{\zeta^2}{\nu} .\]
Upon integrating, one sees that, along $\mathcal{C}_{\mathrm{in}}$, the energy $\varpi$ is given by
    \begin{equation}
        \varpi(u,1)=M \e^{-\int_{-\infty}^u\frac{\zeta^2}{\nu r}\dd \tilde{u}}+\e^{-\int_{-\infty}^u\frac{\zeta^2}{\nu r}\dd \tilde{u}}  \int_{-\infty}^u \frac12\(1+\frac{e^2}{r^2}\)\frac{\zeta^2}{\nu}\e^{\int_{-\infty}^{u'}\frac{\zeta^2}{\nu r}\dd \tilde{u}}\dd u' . \label{eq:null:varpiintegral}
    \end{equation}
Now, observe that, on $\mathcal{C}_{\mathrm{in}}$, we have 
    $$\zeta=r\pu\phi=\pu(r\phi)+\phi.$$
Moreover, by \eqref{eq:null:idI-} and \eqref{eq:nullcase:ass:limit2}, we have that
    $$ |r\phi(u,1)|=\left|\int_{-\infty}^u \pu(r\phi)\dd u'\right|\leq\frac{2|\Phi^-|}{1-p}|u|^{-p+1}.$$
Combining the estimate above with \eqref{eq:nullcase:ass:limit2}  and applying the triangle inequality, we thus find 
    \[|\zeta|\leq C_\zeta|u|^{-p},\]
where $C_\zeta=2|\Phi^-|+\frac{2|\Phi^-|}{p-1}$.  Inserting this estimate back into \eqref{eq:null:varpiintegral}, and using that $r(u,v)\geq |u|$ for all $v\geq 1$ as a consequence of $\lambda>0$, we thus find
    \begin{align*}
       \varpi(u,1)  &\leq M \e^{\int_{-\infty}^{U_0} C_\zeta^2 |\tilde{u}|^{-2p-1}\dd \tilde{u}} \\
       +&\e^{\int_{-\infty}^{U_0} C_\zeta^2  |\tilde{u}|^{-2p-1}\dd \tilde{u}}  \int_{-\infty}^{U_0} \frac12\(1+\frac{e^2}{r^2}\)C_\zeta^2  d_\nu^{-1}|u'|^{-2p}\e^{\int_{-\infty}^{u'} C_\zeta^2  |\tilde{u}|^{-2p-1}\dd \tilde{u}} \dd u' .
    \end{align*}
For sufficiently large values of $U_0$, the RHS can be chosen smaller than $2M$. Similarly, one can make the second term in the RHS of \eqref{eq:null:varpiintegral} small enough such that the lower bound for $\varpi$ also follows. The bounds for $m$ then follow from \eqref{eq:sys:varpi} by again choosing $U_0$ sufficiently large.
\end{proof}

Equipped with these energy bounds on $\mathcal C_{\mathrm{in}}$ (to be thought of as initial data), we can now exploit the monotonicity properties of the (renormalised) Hawking mass to extend these bounds into all of $\mathcal D_{U_0}$:

    \begin{prop}[Energy boundedness in $\mathcal D_{U_0}$]\label{prop:null mass boundedness}
    For sufficiently large negative values of $U_0$, we have the following bounds in all of $\DU$, where $\DU$, $r$, $m$ and $\phi$ are as described in \eqref{eq:DutoDU0} and in the assumptions of section~\ref{sec:null:assumptions}: 
    \begin{align}
        	\pu\varpi\leq0 ,\label{eq:null:puvarpi-}\\
        	\pv\varpi\geq0. \label{eq:null:pvvarpi+}
    \end{align}
     In particular, we have
        \begin{equation}
        \frac{M}{2}<m,\varpi,\varpi-\frac{e^2}{r}\leq M \label{eq:null:boundednessofvarpi}.
    \end{equation} 
       Moreover, we have 
    \begin{equation}
    0<d_\mu:=\frac12<1-\mu\leq 1.
    \end{equation}
    \end{prop}

\begin{proof}
Observe that 
$$\kappa=\frac{\lambda}{1-\mu}$$
 is positive by \eqref{eq:null:ass:kappapositive}. We thus obtain that $1-\mu>0$, so \eqref{eq:puvarpi}, \eqref{eq:pvvarpi}  imply $\pu\varpi \leq 0$, $ \pv\varpi\geq 0$, respectively. From these monotonicity properties, we obtain the following global energy bounds for all $(u,v)\in \DU$ (we recall assumption \eqref{eq:null:Hawking}):
\[ \varpi(U_0,1)\leq \varpi (u,v)\leq M,\]
so the estimate \eqref{eq:null:boundednessofvarpi} for $\varpi$ follows from \eqref{eq:null:boundednessofvarpialongCin}.
Boundedness of $m$ and $\varpi-\frac{e^2}{r}$ again follows by choosing $U_0$ sufficiently large. 
To find the positive lower bound for $1-\mu$, we simply insert the upper bound $m\leq M$ into the definition $\mu=1-\frac{2m}{r}$. This gives $1-\frac{2M}{|U_0|}\leq 1-\mu$. The bound then follows by choosing $U_0\leq-4M$.
\end{proof}
\subsubsection*{The energy estimates}
Energy boundedness (Prop.~\ref{prop:null mass boundedness}) in particular implies the following two crucial energy estimates by the fundamental theorem of calculus (simply integrate eqns.~\eqref{eq:puvarpi}, \eqref{eq:pvvarpi}), which hold throughout $\mathcal{D}_{U_0}$:\todo{Make propo?}
\begin{align}
0\leq\int_{v_1}^{v_2}\frac{1}{2}\frac{\theta^2}{\kappa}(u,v)\dd v\leq \frac{M}{2}, \label{eq:ebinvdirection}\\
0\leq -\int_{u_1}^{u_2}\frac{1}{2}\frac{(1-\mu)}{\nu}\zeta^2(u,v)\dd u\leq \frac{M}{2}. \label{eq:ebinudirection}
\end{align}

%\subsubsection*{Fixing the $v$-coordinate}
%Notice that, in stating our assumptions in section~\ref{sec:null:assumptions}, we have not fixed the choice of $v$. In fact, the assumptions are invariant under any change of $v$-coordinate $v\mapsto v'$ with $\frac{dv'}{dv}>0$, $v'(1)=1$. We now fix the $v$-coordinate by setting
%\begin{equation}\label{eq:vcoordinate}
%\kappa=\frac{\lambda}{1-\mu}=1
%\end{equation}
%on $\mathcal{I}^-$ in a limiting sense. We will show below that $\kappa$ remains comparable to 1 away from $\mathcal{I}^-$ as well (cf.\ \eqref{eq:null:kappacomp1}). Hence, by the same argument as for the $u$-coordinate, and by the bounds on $1-\mu$, we find that, for fixed $u$,  $v$ tends to $\infty$ as $r$ tends to infinity, i.e., as $\mathcal{I}^+$ is approached. In other words, $V=\infty$ in the definition \eqref{eq:null:DU} of $\mathcal D_{U}$.

Equipped with these energy estimates, we can now control the geometric quantities $\nu,\lambda,\kappa$  in $L^\infty$. 

\begin{prop} \label{prop:null geometric quantities boundedness}
For sufficiently large negative values of $U_0$, there exist positive constants $d_\kappa,C_\lambda,C_\nu, d_\lambda, d_\nu$, depending only on initial data\footnote{By this, we henceforward mean that the constants depend only on $\Phi^-$, $p$, $M$ and $e^2$. We also recall that the choice of $U_0$ only depends on the constants $\Phi^-$, $p$, $\epsilon$, $M$, $e^2$ and the implicit constant in the RHS of \eqref{eq:nullcase:ass:limit2}. }, such that the following inequalities hold throughout all of $\DU$, where $\DU$, $r$, $m$ and $\phi$ are as described in \eqref{eq:DutoDU0} and in the assumptions of section~\ref{sec:null:assumptions}:%, with the $v$-coordinate having been fixed in \eqref{eq:vcoordinate}:
\begin{align}
\pu\kappa\leq0,\\
d_\kappa\leq\kappa\leq 1,\label{eq:null:kappacomp1}\\
d_\lambda\leq\lambda\leq C_\lambda,\\
-d_\nu\geq\nu\geq-C_\nu.\label{eq:null:nukomc}
\end{align}
\end{prop}

\begin{proof}
It is clear that $\pu\kappa\leq0$ (see eq.~\eqref{eq:pukappa}). Since $\lim_{u\to-\infty} \kappa=\lim_{u\to-\infty}\lambda=1$ by assumption, $\kappa\leq 1$ follows by monotonicity.
Moreover, integrating the equation \eqref{eq:pukappa} for $\pu\kappa$ in $u$, we find, for $(u,v)\in \DU$,
\begin{align*}
\kappa(u,v)&=\kappa(-\infty,v)\e^{-\int_{-\infty}^{u}\frac{1}{r}\frac{\zeta^2}{\nu}\dd u'}\\
		&\geq \e^{\frac{2}{U_0 }\int_{-\infty}^{U_0}\frac{1}{2}\frac{\zeta^2}{\nu}\frac{1-\mu}{d_\mu}\dd u'} \geq \e^{-\frac{1}{U_0d_\mu }M}\geq \e^{-\frac{1}{2d_\mu}}=:d_\kappa,
\end{align*}
where we used $r\geq- U_0$  and $1-\mu>d_\mu$ in the second step, and the energy estimate \eqref{eq:ebinudirection} in the last step. 
We now immediately get bounds on $\lambda=\kappa(1-\mu)$: $d_\mu d_\kappa\leq \lambda \leq 1 $.

To finally show boundedness for $\nu$, we integrate eq.~\eqref{eq:pupvr} from $\mathcal{C}_{\mathrm{in}}$. We find, for $(u,v)\in \DU$:
\begin{align}\label{eq:null:nuestimate}
|\nu(u,v)|&=|\nu(u,1)|\e^{\int_1^v\frac{2\kappa}{r^2}\(\varpi-\frac{e^2}{r}\)\dd v'}.
\end{align}
The  bound \eqref{eq:null:nukomc} then follows in view of
\begin{align*}
\left|\int_1^v\frac{\kappa}{r^2}\(\varpi-\frac{e^2}{r}\)\dd v'\right|\leq \frac{1}{d_\mu}\int_1^v\frac{\lambda}{r^2 }M\dd v'
			\leq \frac{1}{d_\mu}M\(\frac{1}{r(u,1)}-\frac{1}{r(u,v)}\)
            			\leq \frac{1}{d_\mu}M\cdot \frac{1}{|U_0|}.
\end{align*}\end{proof}

\subsection{Sharp upper and lower bounds for \texorpdfstring{$\pu(r\phi)$}{d/du(r phi)} and \texorpdfstring{$r\phi$}{r phi}}\label{sec:null:purphi}
In this section, we will use the previous results, in particular the energy estimates, to derive sharp upper and lower bounds for $r\phi$ and $\pu(r\phi)$.

\begin{thm} \label{thm:null boundedness of phi}
For sufficiently large negative values of $U_0$, there exist positive constants $b_1,b_2,B_1,B_2$, depending only on initial data, such that the following estimates hold throughout $\DU$, where $\DU$, $r$, $m$ and $\phi$ are as described in \eqref{eq:DutoDU0} and in the assumptions of section~\ref{sec:null:assumptions}: %with the $v$-coordinate having been fixed in \eqref{eq:vcoordinate}:
	\begin{equation}
    b_1|u|^{-p}\leq|\pu(r\phi)|\leq B_1 |u|^{-p}\label{eq:null:thm:boundsforpurphi}
    \end{equation}
and
    \begin{equation}
    b_2|u|^{-p+1}\leq|r \phi| \leq B_2 |u|^{-p+1}. \label{eq:null:thm:boundsonrphi}
    \end{equation}
In particular, both quantities have a sign.
\end{thm}
    
 \begin{proof}
We will prove this by integrating the wave equation \eqref{eq:wave} along characteristics and using the energy estimates.
In a first step, we will integrate $\pu\phi=\zeta/r$ along an ingoing null ray starting from $\mathcal{I}^-$ and use Cauchy--Schwarz and the energy estimate to infer weak decay for the scalar field: $|\phi|\lesssim r^{-1/2}$. 
In a second step, we will integrate the wave equation \eqref{eq:wave} along an outgoing null ray starting from $\mathcal{C}_{\mathrm{in}}$, using the decay obtained in step 1, to then infer bounds on $|\pu(r\phi)|\lesssim |u|^{-3/2}$.
In a third step, we integrate $\pu(r\phi)  $  from  $\mathcal{I}^-$ to improve the decay of the radiation field: $|r\phi|\lesssim |u|^{-1/2}$. 
We then reiterate steps 2 and 3 until the decay matches that of the initial data on $\mathcal{C}_{\mathrm{in}}$. (Note that one could replace this inductive procedure by a continuity argument; this will be the approach of section~\ref{sec:timelike}.)
    
Let now $(u,v)\in \DU$. Recalling the no incoming radiation condition \eqref{eq:null:idI-}, we obtain
    \begin{align}\begin{split}
        |\phi(u,v)|&=\left|\int_{-\infty}^u\frac{\zeta}{r}\dd u'\right|\\
        &\leq \(\int_{-\infty}^u\frac{\zeta^2}{\nu}(1-\mu)\dd u'\)^{\frac12} \(\int_{-\infty}^u\frac{-\nu}{1-\mu}\frac{1}{r^2}\dd u'\)^{\frac12}\\
        &\leq \sqrt{\frac{M}{2}}\sqrt{\frac{1}{d_\mu}\frac{1}{r(u,v)}}\leq C_1 r(u,v)^{-\frac12},
\end{split}    \label{usageofenergyestimates}\end{align}
where we used the energy estimate \eqref{eq:ebinudirection} in the last estimate.

 Next, by integrating the wave equation \eqref{eq:wave} from $v=1$, we get
    \begin{nalign}\label{eq:null:thm:proof1}
      |\pu(r\phi)(u,v)|&\leq  \frac{2|\Phi^-|}{p-1}|u|^{-p}+\left|\int_{1}^v 2\frac{\nu}{1-\mu} \(\varpi-\frac{e^2}{r} \)\frac{\lambda\phi}{r^2}\dd v'\right|\\
      &\leq  \frac{2|\Phi^-|}{p-1}|u|^{-p}+\int_{1}^v 2\frac{C_\nu}{d_\mu} M C_1\pv\(-\frac{2}{3r^{\frac32}}\)\dd v'\\
      &\leq  \frac{2|\Phi^-|}{p-1}|u|^{-p}+ \frac43\frac{C_\nu C_1}{d_\mu} M r(u,1)^{-\frac{3}{2}},
    \end{nalign} where we used \eqref{eq:nullcase:ass:limit2} to estimate the boundary term.
But now recall that on $\mathcal{C}_{\mathrm{in}}\cap\{u\leq U_0\}$, i.e.\ on $v=1$, we have that  $r$ and $|u|$ are comparable\footnote{This is a property that we will not be able to exploit in the timelike case!}, so we indeed get, for some constant $C_2$:
    \[| \pu(r\phi)|(u,v)\leq C_2 |u|^{-\min(\frac32,p)}.\]
In a third step, we integrate this estimate in the $u$-direction:
    \[|r\phi|(u,v)\leq\int_{-\infty}^u |\pu(r\phi)|\dd u'\leq \max(2,(p-1)^{-1})C_2|u|^{-\min(\frac12,p-1)}.\]
This is an improvement over the decay obtained from the energy estimates. We can plug it back into the second step, i.e.\ into \eqref{eq:null:thm:proof1}, to get improved decay for $\pu(r\phi)$, from which we can then improve the decay for $r\phi$  again. 
The upper bounds \eqref{eq:null:thm:boundsforpurphi}, \eqref{eq:null:thm:boundsonrphi} then follow inductively.

Moreover, we can use the upper bound  $|r\phi|\leq B_2|u|^{-p+1}$ to infer a lower bound on $\pu(r\phi)$:
Integrating again the wave equation \eqref{eq:wave} as in \eqref{eq:null:thm:proof1}, and estimating  the arising integral according to
\begin{align*}
    \left|\int_{1}^v 2\frac{\nu}{1-\mu} \(\varpi - \frac{e^2}{r}\)\frac{\lambda r\phi}{r^3}\dd v'\right|&\leq \int_{1}^v \frac{C_\nu}{d_\mu} M B_2|u|^{-p+1}\pv\(-\frac{1}{r^{2}}\)\dd v'\\
    &\leq \frac{C_\nu B_2}{d_\mu} \frac{M}{r(u,1)^{2}}|u|^{-p+1}\leq \frac{|\Phi^-|}{4(p-1)}|u|^{-p},
\end{align*}
where the last inequality holds true for large enough $U_0$, we obtain the lower bound 
\begin{equation}
    |\pu(r\phi)(u,v)|\geq \frac{|\Phi^-|}{2(p-1)}|u|^{-p}-\frac{|\Phi^-|}{4(p-1)}|u|^{-p}=\frac{|\Phi^-|}{4(p-1)}|u|^{-p}.
\end{equation}
In fact, we get the asymptotic statement that
\begin{equation}
    \pu(r\phi)(u,v)-\pu(r\phi)(u,1)=\mathcal{O}(|u|^{-p-1}).
\end{equation}
The lower bound for $r\phi$ then follows by integrating the lower bound for $\pu(r\phi)$. 
\end{proof}

We have now obtained bounds over all relevant quantities. Plugging these back into the previous proofs allows for these bounds to be refined. This is done by following mostly the same steps but replacing all energy estimates with the improved pointwise bounds we now have at our disposal. 
\begin{cor}\label{cor:nullcase:asymptotics}
For sufficiently large values of $U_0$, we have the following asymptotic estimates throughout $\DU$, where $\DU$, $r$, $m$ and $\phi$ are as described in \eqref{eq:DutoDU0} and in the assumptions of section~\ref{sec:null:assumptions}:%, with the $v$-coordinate having been fixed in \eqref{eq:vcoordinate}:
\begin{align}
    |\varpi(u,v)-M|&=\mathcal{O}(|u|^{-2p+1}),\\
    |\nu(u,v)+1|&=\mathcal{O}(|u|^{-1}),\\
    |\kappa(u,v)-1|&=\mathcal{O}(r^{-1}|u|^{-2p+1}),\\
    |\lambda(u,v)-1|&=\mathcal{O}(r^{-1}),\label{eq:cor:lambda}\\
    |\pu(r\phi)(u,v)-\pu(r\phi)(1,v)|&=\mathcal{O}(|u|^{-p-1}).
\end{align}
In particular, since  $|u|^{p-1} r\phi$ takes a limit on initial data as $u\to-\infty$, it takes the same limit everywhere, that is:
\begin{equation}
    \lim_{u\to -\infty} |u|^{p-1} r\phi(u,1)=\lim_{u\to -\infty} |u|^{p-1} r\phi(u,v)=\Phi^-
\end{equation}
for all $v\geq 1$. In particular, we then have
\begin{equation}
    r\phi(u,v)=\frac{\Phi^-}{|u|^{p-1}}+\mathcal{O}(u^{-p+1-\epsilon}).
\end{equation}
\begin{rem}
Notice, in particular, that we obtain that $\zeta \sim |u|^{-2}$ if $p=2$. Using the results below, one can also show that $\zeta=r\pu\phi$ attains a limit on $\mathcal I^+$ and that $\lim_{\mathcal I^+}\zeta\sim |u|^{-2}$ (see Remark \ref{long remark ZetaXi}). 
Comparing \eqref{eq:Argument:Bondimass} with \eqref{eq:puvarpi} (see Remark~\ref{remark:null:bondimass2}), this can be recognised as the direct analogue of the condition that $|\Xi|\sim |u|^{-2}$ from assumption \eqref{eq:intro:Xi-} from section~\ref{sec:intro:CHR}. 
In turn,  \eqref{eq:intro:Xi-}  was motivated by the quadrupole approximation. 
Thus, the case $p=2$ reproduces the prediction of the quadrupole approximation. 
It is therefore the most interesting one from the physical point of view. 
\end{rem}
\begin{rem}\label{rem:proofofthmintrocor}
Note that one can still prove the above corollary if one demands assumption \eqref{eq:nullcase:ass:limit} to hold on $\mathcal I^+$ rather than on $\mathcal C_{\mathrm{in}}$, and if one assumes a positive lower bound on the Hawking mass $m$. In fact, the only calculation that changes in that case is \eqref{eq:null:thm:proof1}: One now integrates $\pu(r\phi)$ from $v=\infty$ rather than from $v=1$. Combined with Theorem~\ref{thm:null:asymptotics of dvrphi} below, this explains the statement of Theorem~\ref{thm.intro.corollary}.
\end{rem}
\end{cor}
We conclude this subsection with the following observation:
\begin{lemma} \label{prop: null comparibilty of r, v-u}
For sufficiently large values of $U_0$, we have throughout $\DU$, where $\mathcal D_{U_0}$ and $r$ are as described in \eqref{eq:DutoDU0} and in the assumptions of sec.~\ref{sec:null:assumptions}:
    \begin{equation}
   |r(u,v)-(v-u)|=\mathcal{O}(\log(r)) .\label{eq:null:comp of r v-u}
    \end{equation}
    \end{lemma}
    
\begin{proof}
This follows from $r(u,1)-r(U_0,1)=U_0-u$ and the following estimate:
\[r(u,v)-r(u,1)=\int_1^v \lambda \dd v'=(v-1)+\mathcal{O}(\log r(u,v)),\]
where we used the asymptotic estimate \eqref{eq:cor:lambda} for $\lambda$.
\end{proof}

\subsection{Asymptotics of \texorpdfstring{$\pv(r\phi)$}{d/dv(r phi)} near \texorpdfstring{$\mathcal{I}^+$}{I+}, \texorpdfstring{$i^0$}{i0} and \texorpdfstring{$\mathcal{I}^-$}{I-} (Proof of Thm.~\ref{thm.intro:nullcase})}\label{sec:nul:asymptotics}
We are now ready to state the main result of this section, namely the asymptotic behaviour of $\pv(r\phi)$. 
Let us first focus on the most interesting case $p=2$. We have the following theorem:
   
\begin{thm}\label{thm:null:asymptotics of dvrphi}
Let $p=2$ in eq.~\eqref{eq:nullcase:ass:limit}, i.e., let  $\lim_{u\to -\infty} |u| r\phi(u,1)=\Phi^-\neq0$.
Then, for sufficiently large negative values of $U_0$, we obtain the following asymptotic behaviour for $\pv(r\phi)$ throughout $\DU$,  where $\DU$, $r$, $m$ and $\phi$ are as described in \eqref{eq:DutoDU0} and in the assumptions of section~\ref{sec:null:assumptions} (in particular, $M\neq0$):%, with the $v$-coordinate having been fixed in \eqref{eq:vcoordinate}:
\begin{equation}
     |\pv(r\phi)|\sim  
\begin{cases}
\frac{\log r-\log|u|}{r^3}, & u=\con,\,\, v \to \infty ,\\
\frac{1}{r^3}, & v=\con,\,\, u \to -\infty,\\
\frac{1}{r^3}, & v+u=\con,\,\, v\to \infty.
\end{cases}
\end{equation}
More precisely, for fixed $u$, we have the following asymptotic expansion  as $\mathcal{I}^+$ is approached:
\begin{equation}
\left|\pv(r\phi)(u,v)+2M \Phi^- r^{-3} \left(\log r-\log(|u|)-\frac32\right)\right|  =\mathcal{O}(r^{-3}|u|^{-\epsilon}). \label{eq:null:thm:B*constant}
\end{equation}
\end{thm}
    Combined with Proposition~\ref{prop:exis4} (and the specialisation to the linear case from section~\ref{sec:linear}), this theorem proves Thm.~\ref{thm.intro:nullcase} from the introduction.
\begin{proof}
 
We plug the asymptotics from Corollary~\ref{cor:nullcase:asymptotics} as well as the estimate \eqref{eq:null:comp of r v-u} into the wave equation~\eqref{eq:wave} to obtain
\[\pu\pv(r\phi)=\frac{-2M \Phi^-}{r^3|u|}+\mathcal{O}(r^{-3}u^{-1-\epsilon})=\frac{2M \Phi^-}{(v-u)^3 u}+\mathcal{O}\left(r^{-3}|u|^{-1-\epsilon}+r^{-4}\log r|u|^{-1}\right).\]
Integrating the above estimate from past null infinity then gives
\begin{align}
    \pv(r\phi)(u,v)-\int_{-\infty}^u\frac{2M \Phi^-}{(v-u')^3 u'}\dd u'= \mathcal{O}(r^{-3}|u|^{-\epsilon}).\label{eq:proof dvrphi1}
\end{align}
We can calculate the integral on the LHS by decomposing the integrand into fractions:
\begin{align}
    \pv(r\phi)(u,v)- 2M \Phi^-\left(\frac{\log|u|-\log(v-u)}{v^3}+\frac{3v-2u}{2v^2(v-u)^2}\right)=\mathcal{O}(r^{-3}|u|^{-\epsilon}).\label{eq:null.precise expansion}
\end{align}
It is then clear that, for fixed $u$ and $v\to\infty$, we have, to leading order,
\begin{equation}
   \pv(r\phi)\sim -\frac{\log r-\log|u|}{r^3}.
\end{equation}
% To see that $B^*$ in eq.~\eqref{eq:null:thm:B*constant} is indeed constant, note that from the above limiting behaviour, we have near $\mathcal{I}^+$:
% \begin{equation}
% \pu\left(r^3\pv(r\phi)\right)=r^3\pu\pv(r\phi)+\mathcal{O}(r^{-1})\sim r\phi
% \end{equation}
% Thus, if $\pu B^*$ were not 0, the above equation would lead to a contradiction to the fact that $r\phi$ remains finite along $\mathcal{I}^+$.
On the other hand, for fixed $v$ and $u\to-\infty$, we have
\begin{equation}
     \pv(r\phi)\sim -\frac{1}{r^3},
\end{equation}
which can be seen by expanding the logarithm $\log(1-v/u)$ to third order in powers of $v/u$.
    
Lastly, if we take the limit along a spacelike hypersurface, e.g.\  along $u+v=0$, we get
\begin{equation}
     \pv(r\phi)\sim -\frac{1}{r^3}.
\end{equation}
    
\end{proof}
\begin{rem}[Similarities to Christodoulou's argument]\label{long remark ZetaXi}
Notice that $\Phi^-$ here plays the same role as $\slashed{\mathcal{D}}^{(3)}\Xi^-$ does in Christodoulou's argument.
Indeed, recall from Remark~\ref{remark:null:bondimass2} that, in our case, the analogue of the radiative amplitude $\Xi$ is $ \lim_{\mathcal I^+}\zeta=\lim_{\mathcal I^+}r\pu\phi$ (this limit exists in view of estimate \eqref{eq:null:thm:B*constant} and the wave equation \eqref{eq:wave}), and that we moreover have 
\begin{equation}
   \lim_{u\to-\infty}\lim_{v\to\infty}u^2\zeta (u,v)=   \lim_{u\to-\infty}\lim_{v\to\infty}u^2(\pu(r\phi)-\nu\phi) (u,v)=\Phi^-;
\end{equation}
and compare equations \eqref{eq:proof dvrphi1}, \eqref{eq:null.precise expansion} to equation \eqref{eq:Argument:betadifference}.
%Notice the similarities between the behaviour of $\pv(r\phi)$ here and $\beta$ from the previous section, in particular the similarity between $(r\phi)(u,\infty)$ and $R(u)$ which both go like $1/u$. 
%Indeed, repeating the argument for the asymptotic expansion of $\beta$ in the proof of Thm.~\ref{thm:Argument:Christodoulou's theorem}, we can establish a precise analogy to the statement that the constant coefficient in front of the logarithmic term in the expansion of $\beta$ is determined by $\Xi$.
%Recall from remark~\ref{remark:null:bondimass2} that in our case, the analogue of $\Xi$ is $\zeta=r\pu\phi$.
%By our previous result and the monotonicity of $\pu(r\phi)$, we have
%\begin{equation}
%    \lim_{u=const.,\,v\to\infty}\zeta (u,v)= \lim_{u=const.,\,v\to\infty}\pu(r\phi) (u,v):=Z(v).
%\end{equation}
%Now, if we choose our initial data on $\mathcal{C}_{\mathrm{in}}$ s.t.\ $u^2\pu(r\phi)$ tends to a limit there, we %also have that
%\begin{equation}
%    \lim_{u\to-\infty}u^2Z(u)=Z^>\neq0;
%\end{equation}
%in fact, these two limits will agree. Similarly, we can also get that (we take $p=2$)
%\begin{equation}
%    \lim_{u\to-\infty}ur\phi(u,\infty)=-Z^-.
%\end{equation}
%Now, $Z^-$ is the analogue of $\Xi^-$.
%On the other hand, observe that we have from the wave equation that
%\begin{equation}
%    \frac{\partial B}{\partial u}=-2M r\phi(u,\infty)
%\end{equation}
%Repeating now the argument in the proof of Thm.~\ref{thm:Argument:Christodoulou's theorem}, we find that 
%\begin{equation}
%    B^*=-2M Z^-.    
%\end{equation}
\end{rem}

One can generalise the above proof to integer $p>1$ to find that the asymptotic expansion of $\pv(r\phi)$ will contain a logarithmic term with constant coefficient at $(p+1)$st order.
Here, we will  demonstrate this explicitly only for the case $p=3$ since this case  is of relevance for the black hole scattering problem, as will be explained in  section~\ref{sec:scattering}. However, we provide a full treatment of general integer $p$ \textit{for the uncoupled problem} in the appendix~\ref{APPendixB}, see Theorem~\ref{thm:B}. We also note that, by considering integrals of the type $\int\frac{1}{(v-u)|u|^{p}}\dd u$ for non-integer $p$, one can obtain similar results for non-integer $p$, cf.~footnote~\ref{fn:2}. For instance, if $p\in(1,2)$, we would obtain that $\pv(r\phi)=Cr^{-1-p}+\dots$.

\begin{thm}\label{thm:null:asymptotics of dvrphi,p=3}
Let $p=3$ in eq.~\eqref{eq:nullcase:ass:limit}, i.e., let  $\lim_{u\to -\infty} |u|^2 r\phi(u,1)=\Phi^-\neq 0$.

Then, throughout $\DU$ and for sufficiently large negative values of $U_0$, where $\DU$, $r$, $m$ and $\phi$ are as described in \eqref{eq:DutoDU0} and in the assumptions of section~\ref{sec:null:assumptions} (in particular, $M\neq 0$),  we obtain for fixed $u$ the following asymptotic expansion for $\pv(r\phi)$ along each $\mathcal{C}_u$ as $\mathcal{I}^+$ is approached:
\begin{equation}
\left|\pv(r\phi)(u,v)-\frac{F(u)}{r^3}-6M\Phi^-\frac{\log(r)-\log|u|}{r^4}\right|  =\mathcal{O}(r^{-4}) ,
\end{equation}
where $F(u)$ is given by 
\begin{equation}\label{eq:thmF(u)}
    F(u)=\int_{-\infty}^u \lim_{v\to\infty}(2m\nu r\phi)(u',v)\dd u'=\frac{2M\Phi^-}{u}+\mathcal{O}(|u|^{-1-\epsilon}).
\end{equation}
%In particular, the asymptotic expansion again remains conformally regular near $\mathcal{I}^-$ and becomes irregular near $\mathcal{I}^+$.
\end{thm}

\begin{proof}
Following the same steps as in the previous proof, we find that 
\begin{equation}\label{orksi}
\pv(r\phi)(u,v)=\mathcal{O}(r^{-3}|u|^{-1}).
\end{equation}
In order to write down higher-order terms in the expansion of $\pv(r\phi)$, we commute the wave equation with $r^3$ and integrate:
\begin{multline}
    r^3\pv(r\phi)(u,v)=r^3\pv(r\phi)(-\infty,v)+\int_{-\infty}^u\pu(r^3\pv(r\phi))(u',v)\dd u'\\
    =\int_{-\infty}^u 3\nu r^2\pv(r\phi)(u',v)\dd u'+\int_{-\infty}^u 2\left(\varpi-\frac{e^2}{r}\right)\nu\kappa r\phi(u',v)\dd u'.\label{pubd}
\end{multline}
Here, we used that, by the above \eqref{orksi}, $r^3\pv(r\phi)$ vanishes as $u\to -\infty$.

Let's first deal with the second integral from the second line of \eqref{pubd}.
Observe that each of the quantities $\varpi$, $\nu$ and $r\phi$ attain a limit on $\mathcal{I}^+$ by monotonicity, and that $\kappa\to 1$ by Cor.~\ref{cor:nullcase:asymptotics}. 
We write these limits as $\varpi(u,\infty)$ etc.
Note, moreover, that $\pv \varpi\lesssim r^{-2}u^{-4}$ and $\pv \nu\lesssim r^{-2}$ by \eqref{eq:pvvarpi} and \eqref{eq:pupvr}, respectively. 
We can further show that $\pv \kappa \lesssim r^{-2}u^{-6}$ by integrating $\pu\pv\log \kappa$ in $u$ from $\mathcal{I}^-$, where $\pv\log\kappa$ vanishes:\footnote{The computation below is the only place where we use that $r$ is $C^3$. If one wishes to compute higher-order asymptotics, then more regularity needs to be assumed.}
\[\pu\pv\log \kappa=\pv\pu\log \kappa=\pv\left(\frac{\zeta^2}{\nu r}\right)=-\frac{\lambda\zeta^2}{\nu r^2}-2\frac{\zeta\theta}{r^2}-\frac{2\left(\varpi-\frac{e^2}{r}\right)\kappa\zeta^2}{\nu r^3}.\]
We can thus apply the fundamental theorem of calculus to write
\begin{nalign}
 &\int_{-\infty}^u 2\left(\varpi-\frac{e^2}{r}\right)\nu\kappa r\phi(u',v)\dd u'\\
 =&\int_{-\infty}^u 2\varpi\nu r\phi(u',\infty)\dd u'-\int_{-\infty}^u \int_v^\infty \pv\left(2\left(\varpi-\frac{e^2}{r}\right)\nu\kappa r\phi\right)(u',v')\dd v'\dd u'.
\end{nalign}
The first integral on the RHS equals $F(u)$ from \eqref{eq:thmF(u)} and asymptotically evaluates to $F(u)=\frac{2M\Phi^-}{u}+\mathcal{O}(|u|^{-1-\epsilon})$ as a consequence of Corollary~\ref{cor:nullcase:asymptotics}.
On the other hand, by the above estimates for the $v$-derivatives of $\varpi$, $\nu$, $\kappa$ and $r\phi$, we can estimate the double integral above according to
\begin{multline}
\int_{-\infty}^u \int_v^\infty \pv\left(2\left(\varpi-\frac{e^2}{r}\right)\nu\kappa r\phi\right)(u',v')\dd v'\dd u'\\
\lesssim \int_{-\infty}^u \int_v^\infty \frac{1}{r^2 |u'|^6}+\frac{1}{r^2 |u'|^2}+\frac{1}{r^2 |u'|^8}+\frac{1}{r^3|u'|}\dd v'\dd u'\lesssim \frac{1}{r |u|}.
\end{multline}

Let us now turn our attention to the first integral in the second line of eq.~\eqref{pubd}. Plugging in our preliminary estimate \eqref{orksi} for $\pv(r\phi)$, we obtain:
\begin{equation}\label{orksii}
 \int_{-\infty}^u 3\nu r^2\pv(r\phi)(u',v)\dd u'
 \lesssim\int_{-\infty}^u \frac{1}{(v-u') |u'|}\dd u'  = \frac{\log(v-u)-\log|u|}{v}.
\end{equation}
Therefore, combining the three estimates above, we obtain from \eqref{pubd} the asymptotic estimate:
\begin{equation}
r^3\pv(r\phi)(u,v)-\int_{-\infty}^u 2m\nu r\phi(u',\infty)\dd u'=\mathcal{O}\left(\frac{\log(v-u)-\log|u|}{v}\right).
\end{equation}
This is an improvement over the estimate \eqref{orksi}. By inserting this into \eqref{orksii}, we can further improve the estimate \eqref{orksii} to
\begingroup
\allowdisplaybreaks
\begin{align}
%\begin{split}
& \int_{-\infty}^u 3\nu r^2\pv(r\phi)(u',v)\dd u'\nonumber\\
 =&\int_{-\infty}^u \frac{-6M\Phi^-}{(v-u') u'}+\mathcal{O}(r^{-1}|u'|^{-1-\epsilon})+\mathcal{O}\left(\frac{\log(v-u')-\log|u'|}{v(v-u')}\right)\dd u'  \nonumber \\
 =&6M\Phi^-\left(\frac{\log(v-u)-\log|u|}{v}\right)+\mathcal{O}(r^{-1}).\label{orksit}
%\end{split}
\end{align}
\endgroup
Here, we used that 
\begin{equation}
\int_{-\infty}^u \frac{\log(v-u')-\log|u'|}{v(v-u')}\dd u'=\frac{1}{v}\mathrm{Li}_2\left(\frac{v}{v-u}\right),
\end{equation}
where $\mathrm{Li}_2$ denotes the \textit{dilogarithm}\footnote{See e.g.~\cite{AbramowitzStegun}, page 1004.}, which has the two equivalent definitions for $|x|\leq 1$:
\begin{equation}\label{dilog}
-\int_0^x \frac{\log(1-y)}{y}\dd y=:\mathrm{Li}_2(x):=\sum_{k=1}^\infty \frac{x^k}{k^2}.
\end{equation}
In particular, we thus have, since $0<v/(v-u)<1$, 
\begin{equation}
\int_{-\infty}^u \frac{\log(v-u')-\log|u'|}{v(v-u')}\dd u'\leq \frac{1}{v}\frac{v}{v-u}\sum_{k=1}^\infty\frac{1}{k^2}=\frac{\pi^2}{6}\frac{1}{v-u}.
\end{equation}

Plugging the above asymptotics \eqref{orksit} back into eq.~\eqref{pubd} and dividing by $r^3$ completes the proof.
\end{proof}

\begin{rem}[Higher derivatives]
We remark that one can commute the two wave equations for $r$ and $r\phi$ with $\pv$ to obtain similar results for higher derivatives.
For instance, one gets that $\pv\pv r\sim r^{-2}$ and, thus, asymptotically,
\begin{equation}
    \pv^2(r\phi)=-\frac{3}{r}\pv(r\phi)+\dots .
\end{equation}
This fact is of importance for proving higher-order asymptotics for general $p$ using \textit{time integrals}, see also the proof of Theorem~\ref{scatteringthm2}, in particular eq.~\eqref{eq:scatteringproof:fact2}.
\end{rem}
\begin{rem}
Comparing these results to those of~\cite{CHRISTODOULOU2002} presented in section~\ref{sec:intro:CHR}, one can of course also compute the Weyl curvature tensor $W$ and relate it to $\phi$ using the Einstein equations. Since we work in spherical symmetry, $\rho$ is the only non-vanishing component of the Weyl tensor $W_{\mu\nu\xi o}$ under the null decomposition \eqref{eq:intro:nulldecomp}. We derive the following formula in Appendix~\ref{sec:app:weyl}:
\begin{equation}\label{eq:null:W3434}
W_{vuvu}=-\frac{\Omega^4}{2}\frac{m}{r^3}+\frac83 \Omega^2 \pu\phi\pv\phi.
\end{equation}
Using the results above, it is thus easy to see that, in the case $p=2$,  the asymptotic expansion of $\rho$ contains a logarithmic term at order $r^{-5}\log r$ (coming from the $\pu\phi\pv\phi$-term).
 We stress, however, that the point of working with the Einstein-Scalar field system is to model the more complicated Bianchi equations (which encode the essential hyperbolicity of the Einstein vacuum equations) by the simpler wave equation (and thus gravitational radiation by scalar radiation), replacing e.g.\ $\beta$ with $\pv(r\phi)$. It is therefore not the behaviour of the curvature coefficients we are directly interested in, but the behaviour of the scalar field.
\end{rem}

\newpage
\section{Case 2: Boundary data posed on a timelike hypersurface }\label{sec:timelike}

 In this section, we construct solutions 
 %(as opposed to only proving a priori estimates as we did in the null case)
  for the setup with vanishing incoming radiation from past null infinity and with polynomially decaying \textit{boundary data} (as opposed to characteristic data as considered in the previous section), 
 \begin{equation}
 r\phi|_\Gamma \sim |t|^{-p+1},
\end{equation}  
posed on a suitably regular timelike curve $\Gamma$, where $t$ is time measured along that curve. 

For instance, the reader can keep the example of a curve of constant $r=R$ in mind, with $R$ being larger than what is the Schwarzschild radius in the linear case: $2 M$. In general, however, $r$ need not be constant and is also permitted to tend to infinity.

In contrast to the characteristic problem of section~\ref{sec:null}, this type of boundary problem does not permit \textit{a priori estimates} of the same strength as those of section~\ref{sec:null}. Instead, we will need to develop the existence theory for these problems simultaneously to our estimates.
The existence theory developed in the present section can then, \textit{a fortiori,} be used to prove Proposition~\ref{prop:exis4} of the previous section (i.e.\ to show the existence of solutions satisfying the assumptions of section~\ref{sec:null:assumptions}).

\subsection{Overview}\label{sec:timelike:overview}
 The problem considered in this section differs from the previous problem in that there are two additional difficulties: In the null case, the two crucial observations were \textbf{a)} boundedness of the Hawking mass (see Prop.~\ref{prop:null mass boundedness}) and \textbf{b)} that $|u|$  is comparable to $r$  along $\mathcal{C}_{\mathrm{in}}$, which we used to propagate $|u|$-decay from $\mathcal{C}_{\mathrm{in}}$ outwards (see \eqref{eq:null:thm:proof1}). It is clear that \textbf{b)} will in general not be true in the timelike case (where $\mathcal{C}_{\mathrm{in}}$  is replaced by~$\Gamma$).
 
 Regarding \textbf{a)}, we recall that the boundedness of the Hawking mass $\varpi$ was just a simple consequence of its monotonicity properties, which allowed us to essentially bound $\varpi$ from above and below by its prescribed values on the initial ingoing or outgoing null ray, respectively. These values were in turn determined by the constraint equations for $\pu \varpi$ and $\pv \varpi$ (eqns. \eqref{eq:puvarpi}, \eqref{eq:pvvarpi}). However, when prescribing \textit{boundary data} for $(r, \phi)$ on $\Gamma$, it is no longer possible to derive bounds for $\varpi$  just in terms of the data.\footnote{It is this fact which forms the main difficulty in proving local existence for this system.}
 
 Instead, we will therefore, inspired by the results for the null case,  \textit{bootstrap} both the boundedness of the Hawking mass and the $|u|$-decay of $r\phi$.
 Unfortunately, the need to appeal to a bootstrap argument comes with a technical subtlety: To show the non-emptiness part of the bootstrap argument, we will need to exploit continuity in a compact region! This forces us to first consider \textit{boundary data of compact support}, $\mathrm{supp}(\phi_\Gamma) \cap (\{u\leq u_0\}\cap\Gamma)  =\emptyset$ for some $u_0$, which obey a  polynomial decay bound that is independent of $u_0$.
   Then, by the domain of dependence property\footnote{We can "extend" to the past, i.e.\ towards $\mathcal{I}^-$, by the Reissner--Nordstr\"om solution for $u\leq u_0$.}, we can consider the finite problem where we set  
\begin{align}
 r\phi=0(=\pv(r\phi)=\pv\varpi),\\
\varpi=M>0
\end{align}
on the outgoing null ray $\mathcal{C}_{u_0}$ of constant $u=u_0$ emanating from $\Gamma$.
The goal is to show uniform bounds in $u_0$ for the solutions arising from this and, ultimately, to push $u_0$ to $-\infty$ using a limiting argument.
\subsubsection*{Structure}
After showing local existence for this initial boundary value problem in section~\ref{sec:local existence}, we first make a couple of restrictive but severely simplifying assumptions (such as smallness of initial data).
These allow us to prove a slightly weaker version of our final result, namely Thm.~\ref{thm:timelike:final}, in which we construct global solutions arising from \textit{non--compactly supported} boundary data. 
This is done in the way outlined above: We first consider solutions arising from \textit{compactly supported} boundary data and prove uniform bounds on $\varpi$  and uniform decay of $r\phi$ for these in section~\ref{sec:timelike bounds} (both uniform in $u_0$). 
Subsequently, we send $u_0$ to $-\infty$ using a Gr\"onwall-based limiting argument in section~\ref{sec:timelike:limit}, hence removing the assumption of compact support.

 Now, while the proof of Theorem~\ref{thm:timelike:final} already exposes many of the main ideas, it does not show sharp decay for certain quantities. As a consequence, the theorem is not sufficient for showing that $\pv(r\phi)=Cr^{-3}\log r+\mathcal{O}(r^{-3})$ (unless the datum for $r$ tends to infinity along $\Gamma$); 
 instead, it only shows that $|\pv(r\phi)|\sim r^{-3}\log r$.
 We will overcome this issue by proving various refinements -- the crucial ingredient to which is commuting with the generator of the timelike boundary $\Gamma$ -- which allow us to not only remove the aforementioned restrictive assumptions but also to show sharp decay on all quantities and, in particular, derive an asymptotic expression for $\pv(r\phi)$. 
 This is done in section~\ref{sec:Refinements}. 
 The main results of this section, namely Theorems~\ref{thm:timelike:final!!!} and~\ref{thm:timelike:logs}, are then proved in section~\ref{sec:refine:limitingargumentover}.
  In particular, these theorems together prove Thm.~\ref{thm.intro:timelikecase} from the introduction.
 The confident reader may wish to skip to  section~\ref{sec:Refinements} immediately after having finished reading section~\ref{sec:timelike bounds}.  
   
 Since the construction of our final solution will span the next 40 pages, we feel that it may be helpful to immediately give a description of the final solutions of section~\ref{sec:refine:limitingargumentover}.  This is done  in section~\ref{sec:timelikepenrose}. The sole purpose of this is for the reader to see already in the beginning what kind of solutions we will construct, it is in no way part of the mathematical argument.
 
 Furthermore, in order for the limiting argument to become more concrete, we will work on an ambient background manifold with suitable coordinates. All solutions constructed in the present section will then be subsets of this manifold. This ambient manifold is introduced in section~\ref{sec:ambient}.

\subsubsection*{The Maxwell field}
 As we have seen in the previous section, the inclusion of the Maxwell field does not change the calculations in any notable way. We will thus, from now on, consider, in order to make the calculations less messy, the case $e^2=0$; however, all results of the present section can be recovered (with some minor adaptations) for $e^2\neq 0$ as well.
 
 \subsection{Preliminary description of the final solution}\label{sec:timelikepenrose}
In this section, we describe the solutions that we will ultimately construct in section~\ref{sec:refine:limitingargumentover}. %We emphasise that even if one were to read the following as \textit{a priori} assumptions, our (continuity-based) methods would still not allow us to derive useful estimates for these solutions.

Let $U_0$ be a negative number, $-\infty< U_0<0$, and define the set
\begin{equation}
    \mathcal{D}_{U_0}:=\{(u,v)\in\mathbb{R}^2\, |\, -\infty< u\leq v< \infty \,\,\text{and}\,\, u\leq U_0\}.
\end{equation}
We denote, for $u\in (-\infty,U_0]$, the sets $\mathcal{C}_u:=\{u\}\times[u,\infty)$  as \textit{outgoing null rays} and, for $v\in [u,\infty)$, the sets $\mathcal{C}_v:=(-\infty,U_0]\times \{v\}$ as \textit{ingoing null rays}. 
 We \textit{colloquially refer to} $\{-\infty\}\times(-\infty,\infty)$ as $\mathcal{I}^-$ or \textit{past null infinity}, to $(-\infty,U_0]\times\{\infty\}$ as $\mathcal{I}^+$ or \textit{future null infinity}, to $\{-\infty\}\times\{-\infty\}$ as $i^-$ or \textit{past timelike infinity}, and to $\{-\infty\}\times\{\infty\}$ as $i^0$ or \textit{spacelike infinity}.
Furthermore, we denote the timelike part of the boundary of $\mathcal{D}_{U_0}$ by
\begin{equation}
    \Gamma:=\{(u,v)\in \mathcal{D}_{U_0} \,|\, u=v \}.
\end{equation}
We denote the generator of $\Gamma$ by $\boldsymbol{T}=\pu+\pv$.
 We will, in section~\ref{sec:timelike:limit}, show the following statement:\footnote{As we mentioned before, we will develop the existence theory at the same time as our estimates. The proposition below, however, extracts from this only the existence theory and the form of the final boundary/scattering data.}
 \begin{prop*}
 Prescribe boundary data $r(u,u)$, $\phi(u,u)$ along $\Gamma$ as follows:
 Let $M>0$. Assume that $r(u,u)>2M$ either tends to a finite limit $R>2M$, or that it tends to an infinite limit, as $u\to -\infty$.
In the case where it tends to a finite limit, we further assume that
\begin{equation}
    \boldsymbol{T}r(u,u)=\mathcal{O}(|u|^{-s})
\end{equation}
for some $s=1+\epsilon_r>1$.
In the case where it tends to an infinite limit, we instead assume that
\begin{equation}
-    \boldsymbol{T}r(u,u)\sim |u|^{-s}
\end{equation}
for some $s\in(0,1]$.
For the scalar field, we assume that
\begin{align}
    \boldsymbol{T}(r\phi)(u,u)&= C_{\mathrm{in},\phi}^1|u|^{-p}+\mathcal{O}(|u|^{-p-\epsilon_\phi}), \label{eq:timelike:assumptions:intitialdatadecayofphi}\\
                 r\phi(u,u)&= \frac{C_{\mathrm{in},\phi}^1}{p-1}|u|^{-p+1}+\mathcal{O}(|u|^{-p+1-\epsilon_\phi})
\end{align}
for some constants $C_{\mathrm{in},\phi}^1\neq 0$, $p>1$ and $\epsilon_\phi\in(0,1)$.

  Then, if $U_0$ is chosen to be a sufficiently large negative number, there exists a unique triplet of $C^2$-functions $(r,\phi,m)$ on $\mathcal D_{U_0}$ that solves the equations \eqref{eq:puvarpi}--\eqref{eq:pvzeta} \textit{pointwise} throughout $\mathcal D_{U_0}$, restricts to the boundary data $r(u,u)$, $\phi(u,u)$ above, and that satisfies for any $v\in(-\infty,\infty)$:
  \begin{align}
  \lim_{u\to-\infty}\pv(r\phi)(u,v)=0,&&\lim_{u\to-\infty}\pv r(u,v)=1,&& \lim_{u\to-\infty}m(u,u)=M>0.
  \end{align}

This solution moreover has the following properties: The area radius $r$ tends to infinity along each of the ingoing and outgoing null rays, i.e., $\sup_{C_u}r(u,v)=\infty$ for all $u\in(-\infty,U_0]$, and $\sup_{C_v}r(u,v)=\infty$ for all $v\in(-\infty,\infty)$, and we have throughout $\mathcal{D}_{U_0}$ that $
    	\nu<0$ and $
    	\lambda,\kappa>0
$.
Furthermore, the solution satisfies for any $v\in(-\infty,\infty)$:
\begin{equation}
\lim_{u\to-\infty} m(u,v)-M=0=\lim_{u\to-\infty}r\phi(u,v).
\end{equation}

 \end{prop*}
 In the above, uniqueness is understood with respect to the class of solutions with finite Hawking mass (see Remark~\ref{rem:uniqueness}).

The reader can again refer to the Penrose diagram below (Figure~\ref{fig:6}). 
 %%%%%%%%%%%
\begin{figure}[htbp]
\floatbox[{\capbeside\thisfloatsetup{capbesideposition={right,top},capbesidewidth=4cm}}]{figure}[\FBwidth]
{\caption{The Penrose diagram of $\mathcal D_{U_0}$. Note that, with our choice of coordinates (namely, $u=v$ on $\Gamma$), $\Gamma$ should be a vertical line instead of a curved line. We depicted it as a curved line to avoid the possible confusion that $\Gamma$ describes the centre of spacetime.}\label{fig:6}}
{  
  \includegraphics[width = 135pt]{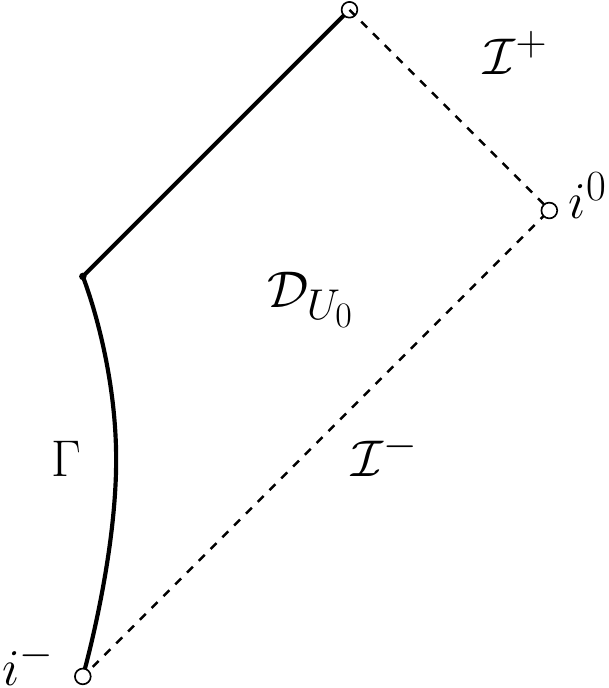}
}
\end{figure}

\subsection{The ambient manifold}\label{sec:ambient}
In this section, we introduce the ambient manifold and coordinate chart that shall provide us with the geometric background on which we shall be working. 

Let $U<0$, and let $\mathcal D_U$ be the manifold with boundary
\begin{equation}
    \mathcal{D}_{U}:=\{(u,v)\in\mathbb{R}^2\, |\, -\infty< u\leq v< \infty \,\,\text{and}\,\, u< U\}.
\end{equation}
We can equip $\mathcal D_U$ with the Lorentzian metric $-\dd u \dd v$.
In the sequel, we shall prescribe on the boundary of $\mathcal D_U$ suitable boundary data $(\hat r, \hat \phi)$ for the Einstein-Scalar field system \eqref{eq:puvarpi}--\eqref{eq:pvzeta}. 
Together with suitable data on an outgoing null ray, these will, at least locally,  lead to solutions $(r, \phi,m)$, as will be demonstrated in the next section. These solutions correspond, according to section~\ref{sec:subsec:system of equations}, to  spherically symmetric spacetimes, whose quotient under the action of $SO(3)$ we will view as subsets of the ambient manifold $\mathcal D_U$. 

\textbf{Throughout this entire section, any causal-geometric concepts such as "null", "timelike" or "future" will refer to the background manifold $(\mathcal D_U,-\dd u\dd v)$.}%, and, given a timelike curve $\Gamma$, we shall work in the coordinates described above.}

\subsection{A local existence result}\label{sec:local existence}
%
%
% \paragraph{Coordinates}
%  As the initial boundary value problem we are considering here is more subtle than it is in the case of initial data on two null cones, we will start by proving a local existence result for initial boundary value problems for the spherically symmetric Einstein-Scalar field equations. It will be necessary to first specify a gauge.
%  Suppose that $\Gamma$ is a past-complete smooth timelike curve and $\bar{C}$ a complete outgoing null ray emanating from $\Gamma$.
% Suppose also that on $\bar{C}$ we have defined a $v$-coordinate (increasing to the future) s.t.\ $v(\Gamma\cap \bar{C})=u_0$. 
%Starting from a point $q$ on $\mathcal{C}_{u_0}$ with coordinate $v\geq u_0$,  we follow the future-directed ingoing null ray emanating from $q$  until we hit $\Gamma$ at a point $r$. We then define $u(r):=v$. In particular, $u(\Gamma\cap \bar{C}):=u_0$.
% In these coordinates,  $\Gamma$  is generated by $\boldsymbol{T}=\pu+\pv$, and points on $\Gamma$ to the future of $\Gamma\cap\mathcal{C}_{u_0}$  are all of the form $(u,u)$. In particular, we can express functions on $\Gamma$ as $f(u)$.
%
% \paragraph{A local existence result}
%We will now prove a local existence result within this gauge. A geometric uniqueness statement to remove the gauge dependence can then be recovered by standard techniques.
%
We will start by proving a local existence result in the gauge specified above:
    \begin{prop}\label{prop:localexistence}
    %    Let $\mathcal M_M$ be as described in section~\ref{sec:ambient}, let $\Gamma$ be a smooth timelike curve in $\mathcal M_M$, and equip $\mathcal M_M$ with $(u,v)$-coordinates as in  section~\ref{sec:ambient}.
 Let $\mathcal D_U$ be as described in section~\ref{sec:ambient}, let $u_0\leq u_1< U$ and $v_1>u_0$, and let $\bar{\mathcal C}=\{u_0\}\times[u_0, v_1]$  be an outgoing null ray intersecting $\bar\Gamma=\{(u,u)\in \mathcal D_U\,|\,u_0\leq u\leq u_1\}$  
    at a point $q=(u_0,u_0)$.
     Specify on $\bar{\mathcal C}$  two $C^2$-functions $\bar{r}(v)$, $\bar{\phi}(v)$,  and specify on $\bar\Gamma$ two $C^2$-functions $\hat{r}(u)$, $\hat\phi(u)$. 
    Moreover, specify a value $\bar m(q)<\bar{r}(u_0)/2$, and define on $\bar{\mathcal C}$ the function $\bar m(v)$ as the unique solution to the ODE 
             \begin{equation}
             \pv \bar m=\frac12\left(1-\frac{2\bar m}{\bar r}\right)\bar r^2\frac{(\pv\bar \phi)^2}{\pv\bar r}\label{eq:timelike:existenceprop:constraint}
         \end{equation}
        with initial condition $\bar{m}(q)$.
    Finally, assume that the following data bounds are satisfied
               \begin{align}
           \max_{\bar{\mathcal C}}
           \{|\log \bar r|, |\log \pv\bar r|,|\bar\phi|, |\pv\bar\phi|,| \log (1-2\bar m(q)/\bar r(u_0))|\}\leq C ,\\
%           \end{equation}
%           and that the following bounds hold on the boundary $\Gamma$
%           \begin{equation}
           \max_{\bar{\Gamma}}
           \{|\log \hat r|, | \boldsymbol T\hat r|,|\hat\phi|, |\boldsymbol T\hat\phi|\}\leq C ,
           \end{align}
 and  assume the usual compatibility conditions at the corner $q$ as well as\footnote{This should be thought of as implying that $\pu r<0$ along $\Gamma$. Note that we cannot specify $\pu r$ as data on $\Gamma$.} 
           \begin{equation}
           \boldsymbol T\hat r(u=v)-\pv \bar r(v)<0. \label{eq:timelike:existenceTrnurlambdar}
           \end{equation}
    Then, for $\epsilon$ sufficiently small and depending only on $C$, there is a region 
        \begin{equation}\label{eq:Deltauepsilon}
        \Delta_{u_0,\epsilon}:=\{(u,v)\in \mathcal D_U\,|\,u_0\leq u\leq v\leq u_0+\epsilon\} 
        \end{equation}
    in which  a unique $C^2$-solution to the spherically symmetric Einstein-Scalar field equations~\eqref{eq:puvarpi}--\eqref{eq:pupvr} that restricts correctly to the initial/boundary data exists. 
    Moreover, higher regularity is propagated, that is to say: If the initial data are in $C^k$ for $k>2$, then the solution will also be in $C^k$.
     \end{prop}

 \begin{proof}
 The proof will follow a classical iteration argument: We will define a contraction map $\Phi$ on a complete metric space such that the fixed-point of this map will be a solution to the system of equations. 

 For an $\epsilon>0$ to be specified later, define 
       \begin{equation}
       Y(\Delta_{u_0,\epsilon}):=\{(r,\phi,m, \mu) \in C^1(\Delta_{u_0,\epsilon})\times C^1(\Delta_{u_0,\epsilon})\times C^0(\Delta_{u_0,\epsilon})\times C^0(\Delta_{u_0,\epsilon})\} 
       \end{equation}
  and the corresponding subspace
       \begin{multline}Y_E(\Delta_{u_0,\epsilon}):=\{(r,\phi,m, \mu) \in Y(\Delta_{u_0,\epsilon})|\max\{|\log r|, |\log \pv r|, |\log (-\pu r)|,\\
       ||\phi||_{C^1(\Delta_{u_0,\epsilon})}, |m|, |\log(1-\mu)|\}\leq E\},
        \end{multline}
 equipped with the metric
    \begin{align*}
        d((r_1,\phi_1, m_1,\mu_1),(r_2,\phi_2, m_2,\mu_2)):=\sup _{\Delta_{u_0,\epsilon}}\{|\log |r_1/r_2||, |\log |\pv r_1/\pv r_2||, \\|\log |\pu r_1/\pu r_2||, 
    ||\phi_1-\phi_2||_{C^1(\Delta_{u_0,\epsilon})}, |m_1-m_2|, |\log|1-\mu_1|/|1-\mu_2||\}.
    \end{align*}

 For any element $(r,\phi, m,\mu)\in Y_E(\Delta_{u_0,\epsilon})$, we now define our candidate for the contraction map~$\Phi$ via $(r',\phi', m',\mu')=\Phi((r,\phi, m,\mu))$, where the primed quantities $r'$, $\phi'$ are, for $(u,v)\in \Delta_{u_0,\epsilon}$, defined via 
     \begin{align}
    	 r'(u,v)=\hat{r}(u)&+\bar{r}(v)-\bar{r}(u)
    			-\int_{u}^v\int_{u_0}^u \frac{2m\nu\lambda}{r^2(1-\mu)} \dd\bar{u}\dd \bar{v}
    			\end{align} and (notice that these double integrals are simply integrals over rectangles)
    			\begin{align}
    	\phi'(u,v)=\hat{\phi}(u)&+\bar{\phi}(v)-\bar{\phi}(u)
    	-\int_{u}^v\int_{u_0}^u \frac 1 r (\pu r \pv \phi+\pv r \pu\phi )\dd\bar{u}\dd \bar{v},
     \end{align}
 and $m'$, $\mu'$  are defined as solutions to the ODE's
     \begin{align}
        \pu\frac{\pv r'}{1-\mu'}&=r'\frac{(\pu\phi')^2}{\pu r'}\frac{\pv r'}{1-\mu'},\\
        \pu m'(u,v)&=\frac{1}{2}(1-\mu')r'^2\frac{(\pu\phi')^2}{\pu r'},
     \end{align}
 with initial conditions  $\mu'(u_0,v)=2\bar m(v)/\bar r(v)$, $m'(u_0, v)=\bar m (v)$, respectively.

 One now checks  that \textbf{a)} $\Phi$  is a map from $Y_E(\Delta_{u_0,\epsilon})$  to itself, and that \textbf{b)} it is a contraction w.r.t.\ the associated metric $d$.
 Both these facts can easily be established by (after also integrating the equations for $\pv r'/(1-\mu')$, $m'$, and solving the ODE for $\bar m$ for some suitably small interval)  bounding in each case the integrand by a continuous function of $E$ and then making the integrals sufficiently small by using the smallness in $\epsilon$, whereas the initial data terms from integrating the equations can be bounded from above and below by continuous functions of $C$ in case \textbf{a)}\footnote{This is the reason for the necessity of condition \eqref{eq:timelike:existenceTrnurlambdar} -- without it, we wouldn't be able to guarantee that the "initial value" term for $\pu r'$ on $\Gamma$ would be negative.}, and they vanish in case \textbf{b)}.\footnote{Note that we introduced $\mu$  and $m$  as a priori independent variables only to deduce more easily that $1-\mu$  remains bounded away from zero.}

 Having established these  two facts, we invoke the Banach fixed point theorem to obtain a unique fixed-point
 \[(r,\phi,m,\mu)\in Y_E(\Delta_{u_0,\epsilon}),\]
 which clearly solves the equations \eqref{eq:puvarpi}, \eqref{eq:pukappa}--\eqref{eq:pupvr} and restricts properly to the initial/boundary data.
 However, it is not yet clear that this fixed point has the desired regularity, that $\mu=2m/r$, or that eq.~\eqref{eq:pvvarpi} is satisfied.
 
 To obtain the desired regularity, observe that the equations that the fixed point obeys immediately tell us  that we  have that $\pu\pv r$, $\pu\pv \phi$, $\pu m$ and $\pu \mu$ are, in fact, continuous. 
 Moreover, we also have that $\mu=2m/r$ everywhere since it holds initially (on $\bar{\mathcal{ C}}$) and we can differentiate in $u$ to propagate equality inwards. (See also the argument below.)
 To infer higher regularity, consider now the equation that $\pu m$ satisfies (an ODE with coefficients that are continuously differentiable in $v$) to see that $\pv m$ is continuous; hence $m$ is continuously differentiable. 
 By a similar argument, we can then also infer that $r$ and $\phi$ are in $C^2$.  This in turn implies that  $m$ is in $C^2$,  which allows us to propagate the constraint equation \eqref{eq:timelike:existenceprop:constraint} inwards by differentiating both sides in $u$: Indeed, we have, by virtue of $m$ being $C^2$:
 \begin{align*}
     \pu\pv m=&\pv \left(\frac12 \left(1-\frac{2m}{r}\right)r^2\frac{(\pu\phi)^2}{\pu r}    \right)\\
             =&\left( \frac{-\pv m}{r}+\frac{m\pv r}{2r^2}   \right)r^2\frac{(\pu\phi)^2}{\pu r}+\left(1-\frac{2m}{r}\right)r\frac{\pv r(\pu\phi)^2}{\pu r}\\
             &+\frac12\left(1-\frac{2m}{r}\right)r^2\left(\frac{2\pu\phi\pv\pu\phi}{\pu r}-\frac{(\pu \phi)^2\pv\pu r}{(\pu r)^2}\right)\\
             =&\left( -r\frac{(\pu\phi)^2}{\pu r}\right)\pv m+\left(1-\frac{2m}{r}\right)r(\pu\phi\pv\phi),
 \end{align*}
 where, in the last step, we used equations \eqref{eq:sys:dudvphi} and \eqref{eq:pupvr}.
 On the other hand, we have:
 \begin{align*}
     &\pu \left(\frac12 \left(1-\frac{2m}{r}\right)r^2\frac{(\pv\phi)^2}{\pv r}    \right)\\
     &=\left( -r\frac{(\pu\phi)^2}{\pu r}\right)\left(\frac12\left(1-\frac{2m}{r}\right) r^2\frac{(\pv\phi)^2}{\pv r}    \right) +\left(1-\frac{2m}{r}\right)r(\pu\phi\pv\phi),
 \end{align*}
 where we used eqns.~\eqref{eq:sys:dudvphi} and \eqref{eq:pukappa} in the last step.
Applying Gr\"onwall's inequality to the two identities above shows that \eqref{eq:pvvarpi} holds everywhere in $\Delta_{u_0, \epsilon}$ and, thus, completes the main part of the proof.

In order to show that higher regularity is propagated if assumed initially, one first shows that if $\bar r$ is in $C^3$, then $\pv^2 r$ is in $C^1$ by considering the equation satisfied by $\pu\pv^2 r$. One then shows a similar statement concerning $\pu^2 r$ by considering the equation for $\pv\pu^2 r$. One thus shows that $r$ is in $C^3$ if the data for $r$ are in $C^3$. A similar argument gives that $\phi$ is in $C^3$, from which it follows directly that $m$ is in $C^3$. One now proceeds inductively.
 \end{proof}
 \begin{rem}
 Comparing to the local existence proof for the characteristic initial value problem, the main difficulty here was that we couldn't treat the equation for $\pu m$ as a constraint equation that is prescribed initially (cf.~\ref{eq:timelike:existenceprop:constraint}), so we had to include it in the contraction map. 
 The value of $m$ along $\Gamma$ is not known initially, but found dynamically via the fixed point theorem (by integrating $\pu m$ from $\bar{\mathcal C}$).
(We recall that, for the characteristic initial value problem, one can conveniently define the contraction map via the three wave equations for $\pu\pv\phi$, $\pu\pv r$ and $\pu\pv m$, and then propagate the equations for $\pu m$ and $\pv m$ "inwards" from initial data.) 
 \end{rem}

\subsection{The finite problem: Data on an outgoing null hypersurface \texorpdfstring{$\mathcal C_{u_0} $}{C-u0}}\label{sec:timelike bounds}

Now that we have established local existence for a small triangular region $\Delta_{u_0,\epsilon}$ as described above, we want to increase the region of existence. % so that we can, in particular, take the limit $\mathcal{C}_{u_0}\to\mathcal{I}^-$. 
For this, we will first need to prove uniform bounds for $r,m,\phi$ (like those obtained in the null case) for initial/boundary data as described in section~\ref{sec:timelikepenrose}. As discussed in the introduction to this section, we will, in the present subsection, assume that the data for $r\phi$ are compactly supported on $\Gamma:=\partial D_U$. We will remove this assumption of compact support in section~\ref{sec:timelike:limit}.

Furthermore, as mentioned in the overview (sec.~\ref{sec:timelike:overview}), we will from now on make an extra assumption in order to simplify the presentation, namely that the exponent $p$ from the bound \eqref{eq:timelike:assumptions:intitialdatadecayofphi} be larger than $3/2$, that is, we will assume $p>3/2$ instead of $p>1$. This assumption will be removed in section~\ref{sec:Refinements}. We will also introduce a lower bound on $R$ \eqref{eq:thm:cc:lowerboundonR} in order to more clearly expose the ideas. In reality, if we only want to show upper bounds, this bound can always be replaced by $R>2M$, see Remark~\ref{rem:betterbound}. We will only need a slightly stronger lower bound on $R$ once we prove lower bounds for the radiation field in Theorem~\ref{thm:timelike:final}.

 \begin{thm}\label{thm:timelike:cc}
 Let $\mathcal D_U$ be as described in section~\ref{sec:ambient}, and
     specify smooth functions $\hat r$, $\hat \phi$ on $\Gamma=\partial D_U=\{(u,u)\in\mathcal D_U\}$, with $\hat\phi$ having compact support.
     Let $\mathcal{C}_{u_0}$ denote the future-complete outgoing null ray emanating from a point $q=(u_0,u_0)$ on $\Gamma$ that lies to the past of the support of $\hat{\phi}$. On $\mathcal{C}_{u_0}$, specify $\bar{m}\equiv M>0$,  $\bar{\phi}\equiv 0$ and an increasing smooth function $\bar{r}$ defined via
   $\bar r(v=u_0)=\hat r(u=u_0)$ and the ODE
     \begin{equation}\label{eq:thmcc:lambdagauge}
     \pv \bar{r}=1-\frac{2M}{\bar{r}}.
     \end{equation}
     Finally, assume that (denoting again the generator of $\Gamma$ by $\boldsymbol{T}=\pu+\pv$) the following bounds hold on $\Gamma$:
     \begin{align}
            |\boldsymbol{T}(\hat r\hat\phi)(u)|&\leq C_{\mathrm{in},\phi}^1|u|^{-p},\\
             |\boldsymbol{T}\hat r (u)|&\leq C_{\mathrm{in},r}|u|^{-s},\\
             \hat{r}& \geq R> 2M,
    \end{align}
    with positive constants  $p>3/2$, $C_{\mathrm{in},\phi}^1$, $C_{\mathrm{in},r}$  and $s>0$. 
   
   Let $\Delta_{u_0,\epsilon}$ denote the region of local existence (cf.~\eqref{eq:Deltauepsilon}) of the solution $(r,\phi,m)$ arising from these initial/boundary data in the sense of Proposition~\ref{prop:localexistence}, with $\epsilon$ only depending on the size of the data--in particular, $\epsilon$ can be chosen indepently of $u_0$.
     Then we have, for sufficiently large negative values of $U_0$ (the choice of $U_0$ depending only on data), and if
     \begin{equation}\label{eq:thm:cc:lowerboundonR}
     R\geq \frac{2M}{1-\e^{-2+\delta(U_0)}}
     \end{equation}
     for some function $\delta(u)\sim 1/|u|^{2p-3}$,
     that the following pointwise bounds hold  throughout  $\Delta_{u_0,\epsilon}\cap\{u\leq U_0\}$:
     \begingroup\allowdisplaybreaks
     \begin{align}
                0<\frac{M}{2}\leq m\leq M,\\
                0<1-\frac{2M}{r}\leq 1-\mu\leq 1,\\
                0<1-\delta(u)=\kappa\leq 1,\\
                0<\left(1-\delta(u)\right)\left(1-\frac{2M}{r}\right)=\lambda\leq 1,\\
                0<d_\nu=:\left(1-\frac{2M}{R}\right) \e^{-\frac{2M}{R-2M}}\leq |\nu|\leq \e^{\frac{2M}{R-2M}}=:C_\nu,\\
                |r\phi|\leq  \frac{ C_{\mathrm{in},\phi}|u|^{-p+1}}{1-\frac{1}{2(1-\delta (U_0))}\log\frac{R}{R-2M}}=:C'|u|^{-p+1}\label{eq:timelikemainbound1},\\
                |\pv(r\phi)|\leq M  C'\frac{|u|^{-p+1}}{r^2},\label{eq:timelike:thm:BS:pvrphi}
      \end{align} \endgroup
    where  $C_{\mathrm{in},\phi}=C_{\mathrm{in},\phi}^1/(p-1)$. In particular, all these bounds are independent of $u_0$.     
    \end{thm}
 
 \begin{proof}
The proof will consist of a nested bootstrap argument. First, we will assume boundedness of the Hawking mass. This will essentially allow us to redo the calculations done in the proof of Prop.~\ref{prop:null geometric quantities boundedness} to show boundedness of the geometric quantities $\lambda, \nu,\mu$. We will then assume $|u|$-decay for the radiation field $r\phi$ and improve this decay by using the previously derived bounds on $\nu,\lambda,\mu$, the assumed bound on $m$, and by integrating the wave equation \eqref{eq:wave} in $u$ and in $v$. We will then use the decay for $r\phi$ to get enough decay for $\zeta$ to also improve the bound on $m$ by integrating $\pu m$. 
Indeed, it will turn out to slightly simplify things if we also introduce a bound on $\zeta$ as a bootstrap assumption.

Let us  start the proof: It is easy to see that the assumptions of the theorem allow us to apply Proposition~\ref{prop:localexistence}; in particular, a solution $(r,\phi,m)$ exists in $\Delta_{u_0,\epsilon}$ for sufficiently small $\epsilon$. Next, notice that, by the monotonicity property of the Hawking mass\footnote{By the above local existence result, we already know that $\nu<0$ and $\lambda>0$ in $\Delta_{u_0, \epsilon}$.}, it is clear that $1-\mu >1-\frac{2M}{R}=:d_\mu$. Having made these  preliminary observations, we now initiate the bootstrap argument. Consider the set 
     \begin{align}\nonumber
     \Delta:=\{(u,v)\in \Delta_{u_0, \epsilon}\,|\,\text{ such that, for all } &(u',v') \in \Delta_{u_0, \epsilon} \,\text{ with }\, u'\leq u,\, v'\leq v :\\     
     |m&(u',v')|\leq \eta\cdot M ,
     \tag{BS(1)} \label{eq:timelike:bootstrap:m}\\
	|\zeta&(u',v')|\leq \tilde{C}|u'|^{-p+1}, 
	\tag{BS(2)}\label{eq:timelike:bootstrap:zeta}\\    
	|r\phi&(u',v')|\leq C'|u'|^{-p+1}\tag{BS(3)} \label{eq:timelike:bootstrap:phi}
    \,\, \}, 
     \end{align}
  where $\eta$, $\tilde{C}$ and $C'$ are positive constants with $\eta>1$, $\tilde C$ sufficiently large, and  $C'=\frac{\eta' C_{\mathrm{in},\phi}}{1-\frac{1}{2(1-\delta (U_0))}\log\frac{R}{R-2M}}>0$ (the positivity of this constant of course being precisely the condition that $R$ be sufficiently large), where $\eta'>1$ is arbitrary and $\delta(u)\sim |u|^{-2p+3}$ can be made arbitrarily small by choosing $U_0$ large enough.
 
  Proposition~\ref{prop:localexistence} guarantees that $\Delta$ is non-empty by continuity; $\{q\}\subsetneq \Delta$ ($q$ is trivially contained in $\Delta$, but $q$ alone wouldn't be enough, we instead need a small triangle in which we can integrate -- this is where we need the assumption of compactness (and thus of compact support) to exploit continuity). 

 Furthermore, $\Delta$ is clearly closed, so it suffices to show that it is also open.
 We will essentially follow the same structure as we did in the null case for this:
 First, note that, in $\Delta$, we again have  the energy estimates (cf.~\eqref{eq:ebinvdirection},~\eqref{eq:ebinudirection}) as a consequence of \eqref{eq:timelike:bootstrap:m}
  \begin{align}
    \int_{v_1}^{v_2}\frac{1}{2}\frac{\theta^2}{\kappa}(u,v)\dd v\leq 2\eta M \label{eq:timelike:ee},\\
    -\int_{u_1}^{u_2}\frac{1}{2}\frac{(1-\mu)}{\nu}\zeta^2(u,v)\dd u\leq  2\eta M.
\end{align}
We then obtain that 
    \[1\geq\kappa\geq \e^{-\frac{2\eta M}{R d_\mu}}=:d'_\kappa\]
 by integrating equation \eqref{eq:pukappa} for $\pu\kappa$ and applying the energy estimate as in the proof of Prop.~\ref{prop:null geometric quantities boundedness}. Later, we will want to show that the lower bound for $\kappa$ can be improved beyond the estimate given by the energy estimate, which is why have also introduced the bootstrap bound on $\zeta$ \eqref{eq:timelike:bootstrap:zeta}. But let us first derive bounds for $\lambda$ and $\nu$:
It is clear that 
    \[d_\lambda:=d'_\kappa d_\mu\leq \lambda\leq 1.\]
For $|\nu|$, we integrate eq.~\eqref{eq:pupvr} from $\Gamma$:
    \[|\nu(u,v)|=|\nu(u,u)|\e^{2\int_u^v\frac{\kappa}{r^2}m\dd v'}.\]
Using that 
    \[\left|\int_u^v\frac{\kappa}{r^2}m\dd v'\right|\leq \frac{\eta M}{d_\mu R},\]
as well as the fact that 
    \[\nu|_\Gamma=\boldsymbol T\hat r-\lambda |_\Gamma\leq-\frac12 \lambda |_\Gamma,\]
where we used that $\boldsymbol T\hat r$ satisfies $|\boldsymbol T\hat{r}|\leq |u|^{-s}$ and can thus be made small by choosing $U_0$ large enough,
we thus get that 
    \[d_\nu\leq |\nu|\leq C_\nu,\] 
with the constants $d_\nu,C_\nu$ only depending on initial/boundary data.

Let us now invoke our second bootstrap assumption \eqref{eq:timelike:bootstrap:zeta}.
% in which $\tilde{C}$ is a constant to be determined later:
%\begin{equation}\tag{BS(2)}
%	|\zeta(u,v)|\leq \tilde{C}|u|^{-p+1}.\label{eq:timelike:bootstrap:zeta}
%\end{equation}
Note that \eqref{eq:timelike:bootstrap:zeta} directly implies that $|m|\leq M$.
Indeed, integrating, as in the null case, the equation for $\pu m$ from $\mathcal{C}_{u_0}$ (see eq.~\eqref{eq:null:varpiintegral}), we obtain
\begin{align}\label{eq:timelike:bootstrapproof:mleqEi}
    \begin{split}
    |m(u,v)|
    &\leq M \e^{\frac{\tilde{C}}{(2p-3)d_\nu R}|u|^{-2p+3}}+\frac12 \frac{\tilde{C}}{(2p-3)d_\nu}|u|^{-2p+3}\e^{2\frac{\tilde{C}}{(2p-3)d_\nu R}|u|^{-2p+3}}.
    \end{split}
\end{align}
The above expression is strictly less than $\eta M$ for $|U_0|$ large enough; this improves the bootstrap assumption \eqref{eq:timelike:bootstrap:m}.
Notice that the second term in the above expression is strictly smaller than the first one for large enough $|u|$. Therefore, by considering again \eqref{eq:null:varpiintegral}, we also get that $m$ is positive, say, $m>\frac{M}{2}$.  
By the monotonicity properties of $m$ (namely, $\pu m\leq 0$), we thus conclude that $\frac{M}{2}\leq m\leq M$ (in fact, we have established that $m-M=\mathcal{O}(|u|^{-2p+3})$).  Therefore, we shall henceforth assume that $|m|\leq M$.

We now use \eqref{eq:timelike:bootstrap:zeta} to improve the lower bound on $\lambda$: Inserting \eqref{eq:timelike:bootstrap:zeta} into the integration of eq.~\eqref{eq:pukappa}, %and turning some of the $u$-decay in $\zeta$ into $r$-decay so that the integrand becomes integrable in $r$,
 we get
\begin{equation}\label{eq:bootstrapproofkappadelta}
    \kappa= 1-\delta(u)>0,
\end{equation}
where $\delta(u)\sim 1-\e^{-\frac{1}{|u|^{2p-3}}}$ tends to $0$ as $|u|$ tends to infinity. In particular, we now get the lower bound:
\begin{equation}
    \lambda= (1-\delta(u))\left(1-\frac{2M}{r}\right). \label{eq:lambdafastdecay}
\end{equation}

We finally invoke our bootstrap assumption \eqref{eq:timelike:bootstrap:phi} for the $|u|$-decay for $r\phi$.
%Take as bootstrap assumption 
%\begin{equation}\tag{BS(3)}
%	|r\phi|\leq C'|u|^{-p+1} \label{eq:timelike:bootstrap:phi}
%\end{equation}
%for $C'=\frac{\eta' C_{\mathrm{in},\phi}}{1-\frac{1}{2(1-\delta (U_0))}\log\frac{R}{R-2M}}>0$ (the positivity of this constant of course being precisely the condition that $R$ be sufficiently large), where $\eta'>1$ is arbitrary.
From the wave equation
    $\pu\pv( r\phi)=2m\nu\kappa\frac{r\phi}{r^3}$,
we get, by integrating in $u$ from $\mathcal{C}_{u_0}$:
    \begin{equation}
        |\pv(r\phi)(u,v)|\leq M C' \frac{|u|^{-p+1}}{r^2}.\label{eq:bootstrapboundfordvrphi}
    \end{equation}
In turn, integrating the above in $v$ from $\Gamma$, we then get, plugging in the lower bound \eqref{eq:lambdafastdecay} to substitute $v$-integration with $r$-integration,
    \begin{nalign}\label{integralthathasbeencomputedbefore}
        |r\phi(u,v)|&\leq \frac{C_{\mathrm{in},\phi}}{|u|^{p-1}}+\frac{MC'}{|u|^{p-1}}\int_{r(u,u)}^{r(u,v)}\frac {1}{r(r-2M)(1-\delta(U_0))}\dd r\\
        &\leq \frac{1}{|u|^{p-1}}\left(C_{\mathrm{in},\phi}+\frac{C'}{2(1-\delta(U_0))}\log\left(\frac{r(u,u)}{r(u,v)}\cdot\frac{r(u,v)-2M}{r(u,u)-2M}\right)\right)\\
          &\leq \frac{1}{|u|^{p-1}}\left(C_{\mathrm{in},\phi}+\frac{C'}{2(1-\delta(U_0))}\log\left(\frac{R}{R-2M}\right)\right)<\frac{C'}{|u|^{p-1}}.
    \end{nalign}
The condition that this last term be less than $\frac{C'}{|u|^{p-1}}$ leads to the following lower bound on $R$:
\begin{equation}
    1-\frac{1}{2(1-\delta(U_0))}\log\frac{R}{R-2M}>0\implies R>\frac{2M}{1-\e^{-2(1-\delta(U_0))}}.
\end{equation}
This closes the bootstrap assumption \eqref{eq:timelike:bootstrap:phi} for $r\phi$. 
Since  $\eta'>1$ was arbitrary, we can take the limit $\eta'\to 1$.

Finally, we use this decay in the scalar field to improve the bootstrap bound on $\zeta$. This will essentially come from the wave equation: First, note that, on $\Gamma$, 
    \[|\pu(r\phi)|(u,u)\leq |\pv(r\phi)|(u,u)+|\boldsymbol{T}(\hat r\hat\phi)|(u)\leq |u|^{-p+1}\(\frac{C_{\mathrm{in},\phi}^1}{|u|}+\frac{ M  C'}{R^2}\),\]
where the second inequality comes from the bound \eqref{eq:bootstrapboundfordvrphi}.
We can now integrate the wave equation from $\Gamma$ to obtain 
    \[|\pu(r\phi)(u,v)|\leq\(\(\frac{C_{\mathrm{in},\phi}^1}{|u|}+\frac{ M  C'}{R^2}\)+\frac{MC_\nu C'}{d_\mu R^2}\)|u|^{-p+1}.\]
We thus get, for some constant $C''$ independent of $\tilde C$,
    \[|\zeta(u,v)|\leq |\pu(r\phi)(u,v)|+|(\nu\phi)(u,v)|\leq C''|u|^{-p+1},\]
improving \eqref{eq:timelike:bootstrap:zeta} for large enough $\tilde{C}$, hence closing the bootstrap argument for $\zeta$. 

This shows that the set $\Delta$ is open and, thus, concludes the proof.
\end{proof}
\begin{rem}\label{rem:pgreater32}
The reason why the above proof only works for $p>3/2$ is that, with the method presented, we cannot show sharp decay for $\pu(r\phi)$ (just integrating the wave equation in $v$ will always pick up the bad boundary term on $\Gamma$, and the decay shown for $\pv(r\phi)$ is sharp).
 This in turn means that we can only close the bootstrap assumption for $m$ for $p>3/2$, see the bound \eqref{eq:timelike:bootstrapproof:mleqEi}. 
 We will explain how to deal  with this issue later in section~\ref{sec:Refinements}, where all the above bounds are made sharp. 
 The reader interested in this may wish to skip to section~\ref{sec:Refinements} directly.
\end{rem}
\begin{rem}\label{rem:betterbound}
The lower bound on $R$ is, in fact, wasteful, as one can already see from the fact that we did not explicitly use the decay for $r\phi$ inside the relevant integrals in \eqref{integralthathasbeencomputedbefore}. Alternatively, one can do a Gr\"onwall argument as follows: We have
    \begin{nalign}\label{groenwallinsteadofBS}
        |r\phi|(u,v)&\leq \frac{C_{\mathrm{in},\phi}}{|u|^{p-1}}+\int_{u}^v\int_{u_0}^u \left|2m\nu\kappa\frac{r\phi}{r^3}\right|(u',v')\dd u'\dd v'\\
        &\leq \frac{C_{\mathrm{in},\phi}}{|u|^{p-1}}+\int_{u_0}^u M\frac{\sup_{v'\in[u,v]}|r\phi|(u',v')}{(1-\mathcal O(|u|^{-\epsilon}))r(r-2M)(u',u)}\dd u',
    \end{nalign}
    where, in the second line, we applied Tonelli and then used that $\nu=1-2M/r+\mathcal O(|u|^{-\epsilon})$ for some $\epsilon>0$, which we will prove in section~\ref{sec:refine:tr} (see~\eqref{thm.refinements.boundontr}). Taking the supremum $\sup_{v'\in[u,v]}$ on the LHS of \eqref{groenwallinsteadofBS} and then applying Gr\"onwall's inequality to this yields
    \begin{equation}\label{eq:betterbound}
   \sup_{v'\in[u,v]}  |r\phi|(u,v')|\leq \frac{C_{\mathrm{in},\phi}}{|u|^{p-1}}\left( \sqrt{\frac{R}{R-2M}}+\mathcal O(|u|^{-\epsilon})\right).
    \end{equation}
 \textbf{   In other words, we can replace the bootstrap argument for \eqref{eq:timelike:bootstrap:phi} by a direct Gr\"onwall argument, and this  only requires the lower bound $R>2M$. } In particular, one can obtain an \textit{a priori estimate} for $r\phi$ (without needing to assume compactness) provided that $m$ is bounded. For now, however, we will continue working with the lower bound $ R\geq \frac{2M}{1-\e^{-2+\delta(U_0)}}$.
\end{rem}

 The above theorem (when also taking into account the bound on $\pu(r\phi)$ shown in the proof) in particular allows us to increase the region of existence of the solution to 
\begin{equation}
        \mathcal{D}_{u_0,U_0}:=\{(u,v)\in \mathcal D_U\,|\,u\in[u_0,U_0),\,v\in[u,\infty)\},
\end{equation}
i.e.\ a region extending towards a part of $\mathcal{I}^+$. In fact, we have:%by continuously applying local existence theory (now also for the characteristic initial value problem).
\begin{thm}\label{thm:globalinv}
Under the same assumptions as in Theorem~\ref{thm:timelike:cc}, the resulting solution exists (and satisfies the bounds of Theorem~\ref{thm:timelike:cc}) in all of $ \mathcal{D}_{U}\cap\{u_0\leq u\leq U_0\}$, and can be smoothly extended to all $u\leq u_0$ by the vacuum solution with mass $M$ that satisfies $\pv r=1-\frac{2M}{r}$ on $\Gamma\cap\{u\leq u_0\}$.
\end{thm}
\begin{proof}
The existence of the solution in the region $\mathcal D_{u_0, U_0}$ follows by continuously applying local existence theory (now also for the characteristic initial value problem) combined with the uniform bounds from Theorem~\ref{thm:timelike:cc}.
Moreover, in view of Birkhoff's theorem, we can smoothly extend to $u\leq u_0$ with the mass--$M$-Schwarzschild solution that satisfies $\pv r=1-\frac{2M}{r}$ on $\Gamma\cap\{u\leq u_0\}$ (see Figure~\ref{fig:7}).\footnote{This may seem a bit confusing at first: We are specifying the full data (both tangential \textit{and} normal derivatives) on $\Gamma\cap\{u\leq u_0\}$, which is a timelike hypersurface! The reason that this works is that, in spherical symmetry, i.e.\ in $1+1$ dimensions, time and space are on the same footing. Of course, one does not have to exploit this and could, instead, "anchor" $\pv r $ on $\mathcal I^-$.}
To be more concrete, one can compute $\pu\left(\pv r-(1-\frac{2M}{r})\right)=\frac{2M\nu}{r^2(1-\mu)}\left(\pv r-(1-\frac{2M}{r})\right)$ and apply Gr\"onwall's inequality combined with \eqref{eq:thmcc:lambdagauge} to see that $\pv r=1-\frac{2M}{r}$ for all $v\geq u_0\geq u$. 
Similarly, if one imposes that $\pv r=1-\frac{2M}{r}$ on $\Gamma\cap\{u\leq u_0\}$, one can apply Gr\"onwall's inequality to $\pv\left(\pu r-\left(\boldsymbol T \hat r-(1-\frac{2M}{r})\right)\right)$ to show that $\pu r-\left(\boldsymbol T \hat r-(1-\frac{2M}{r})\right)=0$ for all $u\leq u_0$. Combining these two facts uniquely determines the vacuum solution, which we shall henceforth refer to as $(r_0, 0, M)$. In particular, this solution is independent of $u_0$ in the sense that, for any two different values of $u_0$, say $u_{0,a}<u_{0,b}$, the two arising solutions $(r_0,0,M)$ are identical for $u\leq u_{0,a}$. 
%Moreover, in view of Birkhoff's theorem\footnote{Strictly speaking, it suffices to appeal to uniqueness of the corresponding initial/boundary value problem, which, on the other hand, of course constitutes the main ingredient to the proof of Birkhoff's theorem in this setting. This viewpoint may be simpler when considering the corresponding problem with a Maxwell field included.}, since the solution generated by Theorem~\ref{thm:timelike:cc} is \textit{vacuum} ($\phi=0$) in a neighbourhood of $\mathcal{C}_{u_0}$, we can extend to $u=-\infty$ with the Schwarzschild solution of mass $M$ (see Figure~\ref{fig:7}).
\end{proof}

%%%%%%%%%%%%%%
\begin{figure}[htbp]
\floatbox[{\capbeside\thisfloatsetup{capbesideposition={right,top},capbesidewidth=5cm}}]{figure}[\FBwidth]
{\caption{Extension of the finite solution towards $\mathcal I^-$ via the vacuum Schwarzschild solution. The finite solution arises  from trivial data on $\mathcal C_{u_0}$ and boundary data on $\Gamma$ that are compactly supported towards the future of $\mathcal C_{u_0}$ as depicted. We shaded the region that is uniquely determined by the specification of $\pv r=1-\frac{2M}{r}$ on $\Gamma\cap\{u\leq u_0\}$.}\label{fig:7}}
{ \includegraphics[width = 135pt]{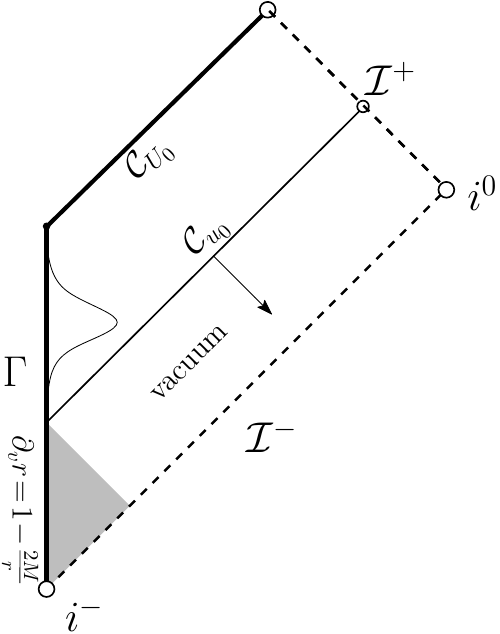}}
\end{figure}
%%%%%%%%%%%%%
\begin{rem}
The resulting Schwarzschild solution $(r_0,0,M)$ for $u\leq u_0$ can be related to the Schwarzschild solution $\mathcal M_M$ in double null Eddington--Finkelstein coordinates $(\tilde u,v)$ of section~\ref{sec:Schwarzschild} by redefining the $u$-coordinate such that $u(\tilde u)=v$ on a past-complete timelike curve $\Gamma\subset \mathcal M_M$ on which the area radius function coincides with $\hat r$.
\end{rem}
 We have thus generated a family of semi-global\todo{or global?} solutions coming from compactly supported and decaying boundary data on a timelike curve $\Gamma$, equipped with bounds independent of the support of the data. We will remove the assumption of compact support in the next section by sending $u_0$ to $-\infty$.
\subsection{The limiting problem: Sending \texorpdfstring{$\mathcal{C}_{u_0}$}{Cu0} to \texorpdfstring{$\mathcal{I}^-$}{I-}}\label{sec:timelike:limit}
In the previous proof, we crucially needed to exploit compactness to show the non-emptiness parts of the bootstrap arguments. We now remove this restriction to compact regions.
%Note that we could, in principle, use a density argument to extend the continuous (in a weighted supremum norm that already captures the $|u|$-decay of $r\phi$) operator sending compactly supported boundary data to solutions  to an operator acting on non-compactly supported boundary data. However, if we define the spaces on which the operator acts via weighted supremum norms that incorporate the full decay of $r\phi$, compactly supported data would no longer form a dense subset in these. 
%We would thus lose decay if we were to follow this approach.

We shall follow a direct limiting argument:  Given \textit{non-compactly supported} boundary data, we construct a sequence of solutions with \textit{compactly supported boundary data} that will tend to a unique limiting solution such that this limiting solution will restrict correctly to the initially prescribed non-compactly supported boundary data. 
The advantage of this approach over, say, a density argument, is that it will also provide us with a uniqueness statement as we will explain in Remark~\ref{rem:uniqueness}.

We will first state the "data assumptions" that the limiting solution is required to satisfy and, after that, write down the sequence of finite solutions. The remainder of the subsection will then contain the detailed analysis of the convergence of this sequence.
%\subsubsection*{The background manifold}
%In order to be able to compare different solutions in the limiting argument, we use the topology of the background manifold $\mathcal{M}_M$ from section~\ref{sec:ambient}, whose contents we briefly summarise here:
%We first specify $M>0$ and take the Schwarzschild manifold $\mathcal M_M$ as our background manifold. On $(\mathcal M_M,g_M)$, we fix an Eddington--Finkelstein-normalised $v$-coordinate (with $\pv r=1-\frac{2M}{r}$) and a smooth past-complete timelike curve $\Gamma$ along which the Schwarzschild area radius function $r_0$ satisfies, for $v\leq V_0$ and $V_0$ a sufficiently large negative value,
%   \begin{align}
%             |\boldsymbol{T'} r_0|_{\Gamma}|&\leq \tilde C_{\mathrm{in},r}|t|^{-s},\\
%             r_0|_\Gamma& \geq R\geq \frac{2M}{1-\e^{-2+\delta(V_0)}}.
%    \end{align}
%Here, $\boldsymbol{T'}$ denotes the normalised future-directed generator of $\Gamma$ and $t$ its affine parameter, $\boldsymbol{T'}t=1$,  $\tilde C_{\mathrm{in},r}$, $p$ and $s$ are positive constants, and  $\delta(V_0)\sim |V_0|^{-2p+3}$ appeared before in Thm.~\ref{thm:timelike:cc}.
%We then define a $u$-coordinate such that $\Gamma$ coincides with the set of points where $u=v$.  We shall henceforth work within this coordinate system. The (non-normalised) generator of $\Gamma$ is hence given by $\mathbf T=\pu+\pv$.
\subsubsection*{The "final" boundary data}
Let $M>0$. We restrict to sufficiently large negative values of $u\leq U_0<0$ and specify boundary data $(\hat r, \hat \phi)$ on $\Gamma=\partial \mathcal D_U\cap\{u\leq U_0\}$ as follows: 
The datum $\hat r\in C^2(\Gamma)$ is to satisfy:\footnote{We remind the reader that this lower bound  on $R$ can be replaced by $R>2M$ in view of Remark~\ref{rem:betterbound}.}% chosen to coincide with $r_0|_{\Gamma}$, from which one readily verifies\footnote{One can, for instance, follow the proof of Theorem~\ref{thm:timelike:cc} with $\phi=0$.} that  
 \begin{align}\label{111}
             |\boldsymbol{T}\hat r (u)|&\leq C_{\mathrm{in},r}|u|^{-s},\\
             \hat{r}& \geq R{\geq \frac{2M}{1-\e^{-2+\delta(U_0)}}}\label{222}
 \end{align}
for some positive constant $C_{\mathrm{in},r}$, {where $\delta(U_0)\sim |U_0|^{-2p+3}$ appeared before in Thm.~\ref{thm:timelike:cc}}. On the other hand, $\hat\phi\in C^2(\Gamma)$ is chosen to obey $\lim_{u\to-\infty}\hat r\hat\phi(u)=0$ and
 \begin{align}
            |\boldsymbol{T}(\hat r\hat\phi)(u)|&\leq C_{\mathrm{in},\phi}^1|u|^{-p}
    \end{align}
  with $p>3/2$ and some constant $C_{\mathrm{in},\phi}^1>0$.
  
% \begin{rem}
%Alternatively to the approach above, one may think of $\hat r$ as being the function that \textit{defines} the choice of $\Gamma$. 
%For instance, one could set up the background manifold by solving the vacuum problem ($\phi=0$) with boundary data $\hat r$ satisfying \eqref{111} and \eqref{222} on a timelike curve $u=v$ embedded into two-dimensional Minkowski space $(\mathbb R^2, -\dd u\dd v)$.
%This will, by Birkhoff's theorem, give rise to a Schwarzschild solution in the coordinates specified above. (The explicit construction of this solution could be done by for instance following the methods of section~\ref{sec:timelike} with $\phi=0$.)
% \end{rem}

\subsubsection*{The sequence of finite solutions $(r_k,\phi_k,m_k)$}
We finally prescribe a sequence of initial/boundary data as follows:
In a slight abuse of notation, let, for $k\in\mathbb N$, $\mathcal{C}_k=\{u=-k, v\geq u\}$ denote an outgoing null ray emanating from $\Gamma$, and let $(\chi_k)_{k\in \mathbb N}$ denote a sequence of smooth cut-off functions on $\Gamma$ (that are translates of each other), 
    \begin{equation}\label{eq:chi_k}
        \chi_k=\begin{cases}
        1,& u\geq- k+1,\\
        0,&u\leq -k.
        \end{cases}
    \end{equation}
On $\mathcal{C}_k$, we denote by $\bar{r}_k$ the solution to the ODE 
    \begin{equation}\pv(\bar{r}_k)=1-\frac{2M}{\bar{r}_k}\label{eq:pvr=1-2mr}\end{equation}
with initial condition given by $\hat{r}(\Gamma \cap \mathcal{C}_k)$. 
%In other words, $\bar r_k$ coincides with $r_0$ on $\mathcal{C}_k$, where $r_0$ is the area radius of the Schwarzschild solution of Theorem~\ref{thm:globalinv}.
Then our sequence of initial data is given by:
    \begin{equation}
    (I.D.)_k=\begin{cases}
    \hat{r}_k=\hat{r}, \,\hat{\phi}_k=\chi_k\hat{\phi} &\text{ on }\, \,\,\Gamma, \\
    \bar{r}_k,\,\bar{\phi}_k=0, \,m=M & \text{ on }\, \,\,\mathcal{C}_k.
    \end{cases}
    \end{equation}
This sequence of data leads, by Theorems~\ref{thm:timelike:cc} and~\ref{thm:globalinv} (for sufficiently large values of $U_0$), to a sequence of solutions $(r_k, \phi_k,m_k)$, which we can smoothly extend (by Thm.~\ref{thm:globalinv}) with the background mass--$M$-Schwarzschild solution $(r_0, 0, M)$ that satisfies $\pv r_0=1-\frac{2M}{r_0}$ on $\Gamma$ for $u\leq-k$.  In particular, these solutions all obey the same bounds (uniformly in $k$) from Theorem~\ref{thm:timelike:cc} in the region where they are non-trivial -- most notably, they obey the upper bound \eqref{eq:timelikemainbound1}.

 In the sequel, we will always mean this extended solution when referring to $(r_k,\phi_k,m_k)$.
We will now show that this sequence tends to a limiting solution, and that this limiting solution still obeys the bounds from Thm.~\ref{thm:timelike:cc} and, moreover, restricts correctly to the non-compactly supported "final" boundary data $\hat{r},\hat{\phi}$ and vanishes on $\mathcal I^-$.

Before we state the theorem, let us recall the notation
\begin{equation}
       \DU:=\{(u,v)\in \mathbb{R}^2\,|\,u\in(-\infty,U_0],\,v\in[u,\infty)\} .
\end{equation}

\newcommand{\DUU}{{\mathcal D}_{U_0}}
\begin{thm}\label{thm:timelike:final}
    Let $(r_k, \phi_k,m_k)$ denote the sequence of solutions constructed above.
    Let $p\geq2$, and let $U_0<0$ be sufficiently large.  In the case $p=2$, assume in addition that $C_{\mathrm{in},\phi}$ is sufficiently small.\footnote{The precise smallness conditions depends on both the ratios $M/R$ and $C_{\mathrm{in},\phi}/R$. 
    It is of no importance, however, since we will remove this assumption, as well as the assumption that $p\geq2$, in Thm.~\ref{thm:timelike:final!!!}. (These are the restrictive assumptions mentioned at the beginning of this section.)}
       
       Then, as $k\to\infty$, the sequence $(r_k, \phi_k,m_k)$ uniformly converges to a limit $(r,\phi,m)$, 
            \begin{equation}
                ||r_k\phi_k-r\phi||_{C^1(\DUU)}+||r_k-r||_{C^1(\DUU)}+||m_k-m||_{C^1(\DUU)}\to 0.
            \end{equation}
       This limit is also a solution to the spherically symmetric Einstein-Scalar field equations.    Moreover, $(r,\phi,m)$ restricts correctly to the boundary data $(\hat{r},\hat{\phi})$ and satisfies, for all $v$,
        \begin{equation}
            \lim_{u\to-\infty}r\phi(u,v)=  \lim_{u\to-\infty}\pv m(u,v)=  \lim_{u\to-\infty}\pv (r\phi)(u,v)=0
        \end{equation}
        as well as
        \begin{equation}\label{eq:thm:final:mlimit}
              \lim_{u\to-\infty}m(u,v)=\lim_{u\to-\infty}m(u,u)=M.
        \end{equation}
        Furthermore,  $(r,\phi, m)$, as well as the quantities $\lambda, \nu,\kappa$, satisfy the same bounds as those derived in Theorem~\ref{thm:timelike:cc}, i.e., we have throughout all of $\DU$:\footnote{Recall the definition of $\delta(u)\sim|u|^{-2p+3}$ from Theorem~\ref{thm:timelike:cc}.}
        %i.e.\ we have for all $(u,v)$ such that $-\infty<u\leq U_0$, $u\leq v\leq \infty$:
        \begingroup
\allowdisplaybreaks
            \begin{align}
                0<\frac{M}{2}\leq m\leq M,\\
                0<1-\frac{2M}{r}= 1-\mu\leq 1,\\
                0<1-\delta(u)\leq\kappa\leq 1,\\
                0<\left(1-\delta(u)\right)\left(1-\frac{2M}{r}\right)=\lambda\leq 1,\\
               0< d_\nu=(1-2M/R) \e^{-\frac{2M}{R-2M}}\leq |\nu|\leq \e^{\frac{2M}{R-2M}}=C_\nu,\\
                |r\phi|\leq  \frac{1}{p-1}\frac{ C_{\mathrm{in},\phi}^1|u|^{-p+1}}{1-\frac{1}{2(1-\delta (U_0))}\log\frac{R}{R-2M}}=C'|u|^{-p+1} \label{eq:timelike:limit:theorem:upperboundforrphi},\\
                |\pv(r\phi)|\leq M  C'\frac{|u|^{-p+1}}{r^2}\label{eq:timelike:limit:theorem:upperboundfordvrphi}.
            \end{align} \endgroup
       If we moreover assume that there exists a positive constant $d_{\mathrm{in},\phi}^1\leq C_{\mathrm{in},\phi}^1$ such that 
       \begin{equation}
           |\boldsymbol{T}(\hat r\hat\phi)(u)|\geq d^1_{\mathrm{in},\phi}|u|^{-p},
       \end{equation}
       then, depending on the value of $C_{\mathrm{in},\phi}^1-d^1_{\mathrm{in},\phi}$, and if $R$ is large enough, there exists a positive constant $d'$ depending only on data such that, throughout $\mathcal D_{U_0}$,
            \begin{equation}\label{eq:timelike:limit:theorem:lowerboundforphi}
                |r\phi|\geq d'|u|^{-p+1},
            \end{equation}
        and furthermore, $C'$ can be replaced by $C_{\mathrm{in},\phi}^1/(p-1)$ in the  estimates \eqref{eq:timelike:limit:theorem:upperboundforrphi}, \eqref{eq:timelike:limit:theorem:upperboundfordvrphi}.
        In the special case where $d^1_{\mathrm{in},\phi}=C_{\mathrm{in},\phi}^1$, the condition that $R$ be large enough is given by
        \begin{equation}
            R>\frac{2M}{1-\e^{-1+\delta(U_0)}}.
        \end{equation}
    \end{thm}

Note that, from the proof of this theorem, one can, \textit{a fortiori}, derive Proposition~\ref{prop:exis4} from the previous section.
\begin{proof}[Proof of Theorem~\ref{thm:timelike:final}]
We focus on the case $p=2$; the cases $p>2$ will follow \textit{a fortiori}.
We will show that the sequence $(r_k, r_k\phi_k,m_k)$ is a Cauchy sequence in the $C^1(\DUU)\times C^1(\DUU)\times C^0(\DUU)$-norm:
Let $\epsilon>0$. We want to show that there exists an $N(\epsilon)$ (to be specified later) such that 
    \begin{equation}
        ||(r\phi)_n-(r\phi)_k||_{C^1(\DUU)}+||r_n-r_k||_{C^1(\DUU)}+||m_n-m_k||_{C^0(\DUU)}<\epsilon
    \end{equation} for all  $n>k>N(\epsilon)$.
%%%%%%%%%%%%%%
\begin{figure}[htbp]
\floatbox[{\capbeside\thisfloatsetup{capbesideposition={right,top},capbesidewidth=4.4cm}}]{figure}[\FBwidth]
{\caption{Depiction of the three regions in which we subdivide. In Region 1, both solutions are vacuum. In Region 2, both solutions can be shown to have small matter content for sufficiently large $k$. Estimating the $r$-difference in Region 2, and estimating all differences in Region 3, however, requires more work. }\label{fig:8}}
{\includegraphics[width = 135pt]{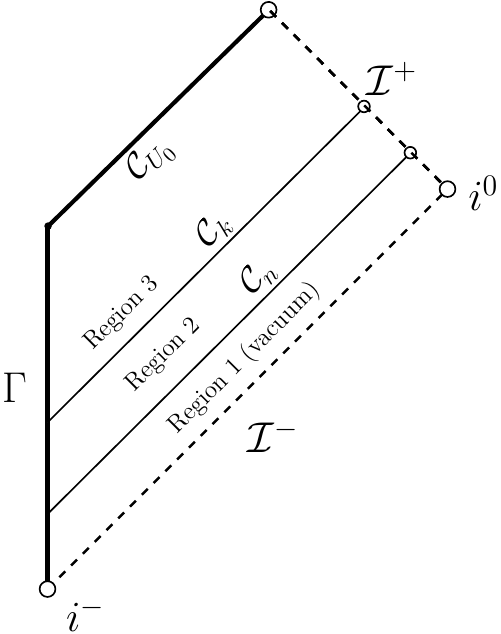}
}
\end{figure}
%%%%%%%%%%%%%

We will show this by splitting up into three regions (see Figure~\ref{fig:8} above). 

Note that, from now on, we will replace most uniform constants simply by $C$ and adopt the usual algebra of constants ($C+D=CE=C=\dots $).
\paragraph{Region 1:}
For $u\leq -n$, the solutions $(r_n, \phi_n,m_n)$, $(r_k,\phi_k,m_k)$ both agree with the vacuum solution $(r_0,0,M)$, so the difference vanishes. See also the argument of Thm.~\ref{thm:globalinv}.
%To be concrete, one can compute $\pu\left(\pv r-(1-\frac{2M}{r})\right)=\frac{2M\nu}{r^2(1-\mu)}\left(\pv r-(1-\frac{2M}{r})\right)$ and apply Gr\"onwall's inequality combined with \eqref{eq:pvr=1-2mr}.
\paragraph{Region 2:}
For $-n\leq u\leq -k$ (let's call this region $\mathcal{D}_{n,k}$), we obtain, in view of the decay of $|r\phi|\leq \frac{C}{|u|}$ and the related bounds for $\pu(r\phi)$ and $\pv(r\phi)$:
    \begin{equation}\label{eqq1}
         ||(r\phi)_n-(r\phi)_k||_{C^1(\mathcal{D}_{n,k})}\leq \frac{C}{k}.
    \end{equation}

For the $m$-difference, recall that $|\zeta|\leq \frac{C}{|u|}$, so we can just integrate $\pu m\sim-\zeta^2$ from $\mathcal{C}_n$, where the difference $m_n-m_k$ is 0:
    \begin{equation}\label{eqq2}
        ||m_n-m_k||_{C^0(\mathcal{D}_{n,k})}\leq\frac{C}{k}.
    \end{equation}

Controlling the $r$-difference, on the other hand, turns out to be more tricky: First, consider the $\kappa$-difference by integrating  $\pu\kappa\sim-\frac{\zeta^2}{r}\kappa$ from $\mathcal{C}_n$ (where the difference $\kappa_n-\kappa_k$ is again 0):
\begin{equation}\label{eqq222}
   ||\kappa_n-\kappa_k||_{C^0(\mathcal{D}_{n,k})}\leq \frac{C}{k}.
\end{equation}
For the other terms appearing in the $r$-difference, we will want to appeal to a Gr\"onwall-type argument.
Let $(u,v)\in\mathcal{D}_{n,k}$. Then we have, by the fundamental theorem of calculus:\footnote{Strictly speaking, we already control the $\lambda$-difference, so we don't need to include it here.}
    \begin{align}
    \begin{split}\label{eq:timelike:limit:proof:rdifferences}
       & |r_n(u,v)-r_k(u,v)|+|\lambda_n(u,v)-\lambda_k(u,v)|+|\nu_n(u,v)-\nu_k(u,v)|\\
        \leq& \left|\int_u^v\int_{-n}^u \pu\pv(r_n-r_k)(u',v')\dd u'\dd v'\right|+ \left|\int_u^v \pu\pv(r_n-r_k)(u,v')\dd v'\right|\\
        +&\left|\int_{-n}^u \pu\pv(r_n-r_k)(u',v)\dd u' \right|+\left|\int_{-n}^u \pu\pv(r_n-r_k)(u',u)\dd u'\right|,
          \end{split}
    \end{align}
where the fourth integral on the RHS comes from the difference $\nu_n-\nu_k$ on $\Gamma$, which we estimate by using that
 \[|\nu_n(u,u)-\nu_m(u,u)|=|\lambda_n(u,u)-\lambda_m(u,u)|\]
 and then integrating $\pu \lambda$ in $u$ from $\mathcal{C}_n$ to estimate the right-hand side $\lambda_n(u,u)-\lambda_m(u,u)$.
 Otherwise, there are no boundary terms since $r_n$ and $r_k$ coincide on $\Gamma$ and $\mathcal{C}_n$.
 
Our strategy shall now be to estimate the RHS of \eqref{eq:timelike:limit:proof:rdifferences} against integrals over products of integrable functions  and $r$-differences so that we can apply a Gr\"onwall argument to it.  
Let's first focus on the double integral on the RHS of \eqref{eq:timelike:limit:proof:rdifferences} (which controls $|r_n-r_k|$). By eq.~\eqref{eq:pupvr}, we have
\begin{align}\label{eq:prooflim1}
    \begin{split}
        \frac12 \pu\pv(r_n-r_k)=&(m_n-m_k)\left(\frac{\nu\kappa}{r^2}\right)_n+m_k\nu_k\frac{\kappa_n-\kappa_k}{r_n^2}\\
        +&m_k(\nu_n-\nu_k)\left(\frac{\kappa}{r^2}\right)_n+m_k\nu_k\frac{\kappa_k(r_n+r_k)}{r_n^2 r_k^2}(r_k-r_n).
    \end{split}
\end{align}
Notice already that, while the first two terms in the equation above are integrable, the $(\nu_n-\nu_k)$-term in \eqref{eq:prooflim1} comes with a non-integrable factor of $1/r^2$. We can, however, write\footnote{Here, and in what follows, the constant hidden inside $\lesssim$ depends only on $M, C_{\mathrm{in},\phi}^1$ and $R$.}
\begin{nalign}
%    \begin{split}
    \label{eq:timelike:limit:proof3}
        &\left|\int_u^v\int_{-n}^u \pu\pv(r_n-r_k)(u',v')\dd u'\dd v'\right|\\
        \lesssim& 
                    \left|\int_u^v\int_{-n}^u \frac{|m_n-m_k|}{r_n^2}+\frac{|\kappa_n-\kappa_k|}{r_n^2}\dd u'\dd v'\right|
                    + \left|\int_u^v\int_{-n}^u \frac{r_n+r_k}{r_n^2 r_k^2}|r_k-r_n|\dd u'\dd v'\right|\\
                    +& \left|\int_u^v\int_{-n}^u     \pu\left(   \frac{r_n-r_k}{r_n^2}\kappa_n m_k    \right)+\(\frac{\pu r_n}{2r_n^3}\kappa_n  m_k-\frac{\pu(\kappa_n m_k )}{r_n^2}\)(r_n-r_k)    \dd u'\dd v'\right|,
 %   \end{split}
\end{nalign}
where we dealt with the bad $(\nu_n-\nu_k)$-term by using the Leibniz rule.

For the first integral in this expression, we use the fact that\footnote{We shall write differences $f_n-f_k$ as $\Delta f$.} $\Delta\kappa+\Delta m\leq \frac{C}{|u|}$.
 In the region where $r\lesssim|u|$, i.e.\ away from $\mathcal{I}^+$, we can convert some of this $|u|$-decay, say $|u|^{-\eta}$ for some $\eta>0$, into $r$-decay so that the integral becomes integrable in $r$. Similarly, near null infinity, we can convert the $r$-decay into $|u|$-decay.
Thus, the first integral can be bounded by $C k^{-1+\eta}$. It will play the role of the function multiplying the exponential in the Gr\"onwall argument.

The second integral of \eqref{eq:timelike:limit:proof3} consists of an integrable 
 function $\sim r^{-3}$ multiplied by $\Delta r$ and, thus, is of the form we want (this term will appear inside the exponential term in the Gr\"onwall lemma, along with the other remaining terms).
 
For the third integral of \eqref{eq:timelike:limit:proof3}, the first term in it can be integrated using the fundamental theorem of calculus and gives, recalling that $\Delta r(u=-n)=0$,
\begin{equation}
    \left|\int_u^v       \left( \frac{r_n-r_k}{r_n^2}\kappa_n m_k \right)(u,v')\dd v'   \right|\lesssim  \int_u^v       \left( \frac{|r_n-r_k|}{r_n^2} \right)(u,v')\dd v' ,  
\end{equation}
so this integrand can also be written as a product of $\Delta r$ and an integrable function.
On the other hand, the second term in the third integral of \eqref{eq:timelike:limit:proof3} is, again, a product of $\Delta r$ and an integrable function that goes like $r^{-3}+r^{-2}|u|^{-2}$ and, thus, can be dealt with in the same way as the second integral.

To summarise, we can estimate \eqref{eq:timelike:limit:proof3} (we also exchange the order of integration using Tonelli) via
\begin{nalign}
%    \begin{split}
    \label{eq:timelike:limit:proof:rdifference1}
      &|r_n-r_k|(u,v) \leq \left|\int_u^v\int_{-n}^u \pu\pv(r_n-r_k)(u',v')\dd u'\dd v'\right|\\
       \leq&
        \frac{C}{k^{1-\eta}}+C  \int_u^v       \left( \frac{|r_n-r_k|}{r_n^2} \right)(u,v')\dd v'     
         + C\int_{-n}^u\int_{u}^v f(u',v')|r_k-r_n|(u',v')\dd v'\dd u'\\
        \leq&         \frac{C}{k^{1-\eta}}+ C  \int_{-n}^v r_n^{-2}|r_k-r_n|(u,v')    \dd v'+C\int_{-n}^u  \sup_{v'\in[u,v]}|r_k-r_n|(u',v')  \int_{u}^v f(u',v') \dd v'\dd u', 
    %\end{split}
\end{nalign}
where $f\geq 0$ is a positive, (doubly) integrable\footnote{Indeed, we have  $\int\int \frac{1}{r_n r_k^2}\dd v\dd u\lesssim \int\frac{1}{r_nr_k}\dd u\leq \left(\int \frac{1}{r_n^2}\dd u\right)^{\frac{1}{2}}\left(\int \frac{1}{r_k^2}\dd u\right)^{\frac{1}{2}} $ by Cauchy--Schwarz.}
 function obeying 
$f(u,v)\leq \frac{C}{r_n^3}+\frac{C}{r_nr_k^2}+\frac{C}{r_k r_n^2}+\frac{C}{r_n^2|u|^2}.$ 
Rewriting, we thus have
\begin{nalign}\label{eq:gr1}
|\Delta r|(u,v)\leq \frac{C}{k^{1-\eta}}+C \int_u^v      r_n^{-2}|\Delta r|(u,v')\dd v' +C\int_{-n}^u F(u')\sup_{v'\in[u,v]}|\Delta r|(u',v')\dd u'
\end{nalign}
for some positive $F(u')$ which obeys, for all $u'\leq u$,
$$F(u')\leq  \frac{C}{r_n^2(u',u)}+\frac{C}{r_n(u',u)r_k(u',u)}+\frac{C}{r_n(u',u)|u'|^2}.$$
Now, if we consider \eqref{eq:gr1} for fixed $u$, and note that $\sup_{v'\in[u,v]}|r_k-r_n|(u',v')$ is non-decreasing in $v$, we obtain, by an application of Gr\"onwall's inequality, that, for each~$u$, 
\begin{nalign}
|\Delta r|(u,v)\leq \left(\frac{C}{k^{1-\eta}}+C\int_{-n}^u F(u')\sup_{v'\in[u,v]}|\Delta r|(u',v')\dd u'\right)\cdot \e^{C \int_u^v      r_n^{-2}(u,v')\dd v'} .
\end{nalign}
The integral in the exponential can be bounded uniformly in $u$ against $\e^{C/R}$, and we thus obtain
\begin{nalign}
|\Delta r|(u,v)\leq\frac{C}{k^{1-\eta}}+C\int_{-n}^u F(u')\sup_{v'\in[u,v]}|\Delta r|(u',v')\dd u'.
\end{nalign}
Taking now the supremum in $v'\in[u,v]$ on the LHS, and applying another Gr\"onwall estimate, we then obtain
\begin{equation}
\sup_{v'\in[u,v]}|\Delta r|(u,v')\leq \frac{C}{k^{1-\eta}}.
\end{equation}

Notice that it was crucial that we took out the supremum in $v$ in \eqref{eq:timelike:limit:proof:rdifference1} -- taking the supremum in $u$ wouldn't work because the rectangle over which we integrate, namely $-n\leq u'\leq u\leq v'\leq v$, increases in $v$-length as $u$ approaches $-n$.

Let's now move to the other three terms in \eqref{eq:timelike:limit:proof:rdifferences}. All of them can be dealt with in a simpler way than the first term. 
Consider, for example, the second term of \eqref{eq:timelike:limit:proof:rdifferences}:
\begin{align}
    \begin{split}\label{eq:timelike:limit:proof:rdifference2}
      & \left|\int_u^v \pu\pv(r_n-r_k)(u,v')\dd v'\right|  \\  
        \leq& C\int_u^v  \frac{|m_n-m_k|+|\kappa_n-\kappa_k|}{r_n^2}(u,v') \dd v'
        \\+& C\int_u^v \left(\frac{1}{r_n^2r_k}+\frac{1}{r_kr_n^2}\right)(|r_n-r_k|+|\nu_n-\nu_k|)(u,v')\dd v' \\
        \leq& \frac{C}{k}+ C\int_u^v \left(\frac{1}{r_n^2r_k}+\frac{1}{r_kr_n^2}\right)(|r_n-r_k|+|\nu_n-\nu_k|)(u,v')\dd v' .
    \end{split}
\end{align}

The last two terms in \eqref{eq:timelike:limit:proof:rdifferences} can be estimated similarly (after also taking the supremum  $\sup_{v'\in[u,v]}$). One thus obtains the inequality \eqref{eq:gr1}, but with $|\Delta r|$ replaced by $|\Delta r|+|\Delta \nu|+|\Delta \lambda|$. From this, we conclude that
\begin{equation}\label{eqq3}
      ||r_n-r_k||_{C^1(\mathcal{D}_{n,k})}\leq \frac{C}{k^{1-\eta}}\leq\frac\epsilon2
\end{equation}
for some $\eta>0$ that can be chosen arbitrarily small. The last inequality holds if $N(\epsilon)$ is chosen accordingly. Similarly, we can make each of the RHS's of estimates \eqref{eqq1} and \eqref{eqq2} smaller than $\epsilon/4$ by choosing $N(\epsilon)$ sufficiently large. We thus conclude that
\begin{equation}\label{eeeeeeeeeqqqqqqqqqqq}
 ||(r\phi)_n-(r\phi)_k||_{C^1(\mathcal{D}_{n,k})}+ ||r_n-r_k||_{C^1(\mathcal{D}_{n,k})}+        ||m_n-m_k||_{C^0(\mathcal{D}_{n,k})}\leq \epsilon.
\end{equation}

\paragraph{Region 3:}
Finally, we consider the region $-k\leq-u\leq -U_0$ (which we shall call $\mathcal{D}_k$). We will again want to perform a Gr\"onwall argument. This time however, we will need to include the $m$- and the $r\phi$-differences in the Gr\"onwall estimate as well, since we can no longer appeal to smallness in $\frac1k$ to estimate them directly.

First, we write down estimates for the differences on the boundary $\Gamma$ and on $\mathcal{C}_k$.
On the boundary $\Gamma$, we have, schematically:
\begin{align}
    \begin{split}
   % |m_n-m_k|_{C^0(\Gamma)}&\leq |m_n-m_k|_{C^0(\mathcal{C}_k)}+\int \pu (\Delta m) \dd u,\\
    |\pu r_n-\pu r_k|_{C^0(\Gamma)}&\leq |\pv r_n-\pv r_k|_{C^0(\mathcal{C}_k)}+\int \pu\pv (\Delta r) \dd u,\\
%    \end{split}\end{align}
%    \begin{align}\begin{split}
    |\pu (r\phi)_n-\pu (r\phi)_k|_{C^0(\Gamma)}&\leq |\mathbf{T} (r\phi)_n-\mathbf{T} (r\phi)_k|_{C^0(\Gamma)}\\
    &+ |\pv (r\phi)_n-\pv (r\phi)_k|_{C^0(\mathcal{C}_k)}+\int \pu\pv(\Delta (r\phi)) \dd u,\\
    | (r\phi)_n- (r\phi)_k|_{C^1(\Gamma)}&\leq \frac{C}{k},
    \end{split}
\end{align}
whereas, on $\mathcal{C}_k$, we have, by the previous result in region 2, that
\begin{equation}
     ||(r\phi)_n-(r\phi)_k||_{C^1(\mathcal{C}_k)}+||r_n-r_k||_{C^1(\mathcal{C}_k)}+||m_n-m_k||_{C^0(\mathcal{C}_k)}\leq \frac{C}{k^{1-\eta}}.
\end{equation}

Moving now on to the Gr\"onwall estimate, we write, using the above estimates for the initial/boundary differences,
\begin{align}
    \begin{split}\label{eq:timelike:limit:proof:gronwallfullthing}
      &  (|(r\phi)_n-(r\phi)_k|+|\pu(r\phi)_n-\pu(r\phi)_k|+|\pv(r\phi)_n-\pv(r\phi)_k|\\
        &+|r_n-r_k|+|\pv r_n-\pv r_k|+|\pu r_n-\pu r_k|+\mathbf{D}\cdot|m_n-m_k|)(u,v)\\
        \leq& \frac{C}{k^{1-\eta}}
       +\left|\int_{-m}^u \pu\pv((r\phi)_n-(r\phi)_k)(u',u)\dd u'\right|+\left|\int_{-m}^u \pu\pv(r_n-r_k)(u',u)\dd u'\right|\\
    &    +\underbracket{\left|\int_u^v\int_{-k}^u \pu\pv((r\phi)_n-(r\phi)_k)(u',v')\dd u'\dd v'\right|}_{(1)}\\
    &    + \left|\int_u^v \pu\pv((r\phi)_n-(r\phi)_k)(u,v')\dd v'\right|
        +\left|\int_{-k}^u \pu\pv((r\phi)_n-(r\phi)_k)(u',v)\dd u' \right|\\
      &  +\underbracket{\left|\int_u^v\int_{-k}^u \pu\pv(r_n-r_k)(u',v')\dd u'\dd v'\right|}_{(2)}\\
       & + \left|\int_u^v \pu\pv(r_n-r_k)(u,v')\dd v'\right|
        +\left|\int_{-k}^u \pu\pv(r_n-r_k)(u',v)\dd u' \right|\\
        &+\underbracket{\mathbf{D}\left|\int_{-k}^u \pu (m_n-m_k)(u',v)\dd u'\right|}_{(3)}.
    \end{split}
\end{align}
Here, we introduced a positive constant $\mathbf{D}>0$. Its relevance will become clear later.
Note that, once we control the integrals~(1),~(2), we can, a fortiori, also control all the other integrals except for~(3).

In order to estimate the underlined term~(1), observe that, schematically (omitting the indices $n,k$),
\begin{align}
    \frac12\Delta(\pu\pv(r\phi))=\Delta \(m \frac{\nu\lambda}{1-\mu}\)\frac{r\phi}{r^3}+m \frac{\nu\lambda}{1-\mu}\frac{\Delta(r\phi)}{r^3}+m \frac{\nu\lambda}{1-\mu}\frac{r\phi}{r^4}\cdot3\Delta r,
\end{align}
so (1) does not pose a problem (as all of the terms multiplying the $\Delta$-differences are integrable).

Next, we look at the underlined term~(2). Again, we have that, schematically, 
\begin{align}\label{eq:limit:problem}
    \begin{split}
        \frac12\Delta(\pu\pv r)&=\Delta \nu \frac{m\lambda}{1-\mu}\frac{1}{r^2}+\Delta \lambda \frac{m\nu}{1-\mu}\frac{1}{r^2}\\
        &+\frac{2\Delta r}{r^3} \frac{\nu m\lambda}{1-\mu}+\frac{\nu m \lambda}{(1-\mu)^2}\frac{1}{r^2}\(\frac{2\Delta m}{r}+\frac{2m\Delta r}{r^2}\)+\frac{\Delta(\mathbf{D} m)}{\mathbf{D}} \frac{\nu\lambda}{1-\mu}\frac{1}{r^2}.
    \end{split}
\end{align}
The first two terms on the RHS can be dealt with using an integration by parts as in~\eqref{eq:timelike:limit:proof3} (after also interchanging the order of integration, using Fubini, for the $\Delta\lambda$-term). The second two terms are, again, products of $\Delta$-differences and integrable functions; they also pose no problem.

The last term, however, is borderline non-integrable and will lead to a $\log$-divergence in the exponential of the Gr\"onwall lemma.
Indeed, we have (cf.~Lemma~\ref{prop: null comparibilty of r, v-u})
\begin{equation}\nonumber
\frac{1}{\mathbf{D}}\int_{-k}^u\int_{u}^v\frac{1}{r^2(u',v')}\dd v'\dd u'\leq \frac{C}{\mathbf{D}} \int_{-k}^u \frac{1}{r(u',u)}\dd u'\leq \frac{C}{\mathbf{D}}\log r(-k,u)\leq \frac{C}{\mathbf{D}}\log k.
\end{equation}
However, we can deal with this divergence! The $\frac{C}{\mathbf{D}}\log k$-term in the exponential of the Gr\"onwall lemma will lead to a power $k^{C/\mathbf{D}}$.
This is the reason for the inclusion of the constant $\mathbf{D}$: By making $\mathbf{D}$ large enough, we can ensure that the problematic term in \eqref{eq:limit:problem} will only lead to a divergence of, say, $k^{\frac{1}{2}-\eta}$. 
In other words, we can ensure that it grows slower than the initial data difference ($\sim k^{-1+\eta}$) decays. 
This trick will come at a price, though, as we will need to absorb the largeness of $\mathbf{D}$ into the smallness of $C_{\mathrm{in},\phi}^1$.\footnote{A more careful investigation reveals that we can also absorb it into the smallness of $M/R$.}

To see how this can be done, we now move on to the underlined term (3) in \eqref{eq:timelike:limit:proof:gronwallfullthing}:
We write
\begin{equation}
    \pu m=\frac12 \frac{1-\mu}{\nu}\((\pu(r\phi))^2-2\pu(r\phi)\frac{\nu r\phi}{r}+\(\frac{\nu r\phi}{r}\)^2\).
\end{equation}
Then, when considering the difference $\Delta(\pu m)$, there will again be precisely one borderline non-integrable term, with all other terms being easily controlled (recall that $\pu(r\phi)$, $ r\phi$ are bounded by a term proportional to $\frac{C_{\mathrm{in},\phi}^1}{|u|}$):
\begin{equation}\label{eq:timelike:limit:proof:borderlineterm}
   \mathbf{D}\cdot \Delta (\pu m) =\mathbf{D}\cdot \frac12 \frac{(1-\mu)\pu(r\phi)}{\nu}\Delta(\pu(r\phi))+\dots \leq C\mathbf{D}\frac{C_{\mathrm{in},\phi}^1}{|u|}\Delta(\pu(r\phi))+\dots 
\end{equation}
for some constant $C$, where the $\dots$-terms denote products of integrable terms and $\Delta$-differences.
Hence, \eqref{eq:timelike:limit:proof:borderlineterm} will again contribute a logarithmic term to the exponential and, thus, a factor of $k^{ C\mathbf{D}C_{\mathrm{in},\phi}^1}$  -- but this factor can now be made small by choosing $C_{\mathrm{in},\phi}^1$ sufficiently small.

In summary, we obtain a similar inequality to \eqref{eq:gr1}, with $|\Delta r|$ replaced by $|\Delta r|+|\Delta \nu|+|\Delta \lambda|+|\Delta r\phi|+|\Delta \pu(r\phi)|+|\Delta \pv(r\phi)|+|\Delta m|$, and can apply a similar Gr\"onwall argument to finally obtain that
\begin{equation}\label{eqq4}
     ||(r\phi)_n-(r\phi)_k||_{C^1({\mathcal{D}}_k)}+||r_n-r_k||_{C^1({\mathcal{D}}_k)}+||m_n-m_k||_{C^0({\mathcal{D}}_k)}\leq \frac{C}{\sqrt{k}}\leq\epsilon,
\end{equation}
where the last inequality holds for $N(\epsilon)\geq\frac{C^2}{\epsilon^2}$. 

Combining now the estimates \eqref{eeeeeeeeeqqqqqqqqqqq} and \eqref{eqq4} concludes the proof that the sequence $(r_n,\phi_n,m_n)$ is a Cauchy sequence. In particular, it converges to a limit $(r,\phi,m)$ in $\DUU$. Moreover, it  not only converges in  the $C^1\times C^1\times C^0$-norm: By looking at e.g.\  the equation for $\pu\pv r$, we obtain that 
\[||\pu\pv r_n-\pu\pv r_k||_{C^0(\DUU)}\to 0\]
as well. Similarly, we get convergence in other higher differentiability norms -- in particular, the limit is also a solution  to the system of equations \eqref{eq:puvarpi}--\eqref{eq:pvzeta}.

The uniform convergence then immediately tells us that all the pointwise bounds derived in Thm.~\ref{thm:timelike:cc} carry over to the limit.

It remains to show the lower bound \eqref{eq:timelike:limit:theorem:lowerboundforphi}.
In order to derive it, we simply integrate $\pv(r\phi)$ from $\Gamma$ and use the upper bound \eqref{eq:timelike:limit:theorem:upperboundfordvrphi} for $\pv(r\phi)$:
\begin{equation}
   \left |\int_u^v \pv(r\phi)\dd v'\right|\leq \frac{C_{\mathrm{in},\phi}^1}{|u|} \frac{\log\frac{R}{R-2M}}{2(1-\delta(U_0))}\frac{1}{1-\frac{1}{2(1-\delta(U_0))}\log\frac{R}{R-2M}}.\label{lowerbound}
\end{equation}
The condition that this be strictly smaller than $\frac{d^1_{\mathrm{in},\phi}}{|u|} $ can always be satisfied for large enough $R$. In the case where $d^1_{\mathrm{in},\phi}=C_{\mathrm{in},\phi}^1$, we obtain:\footnote{Note that if we inserted the better bound \eqref{eq:betterbound} here, we would obtain the marginally better lower bound $R>2.95 M$ instead. Again, this can be improved, but we do not present this here.}
\[R>\frac{2M}{1-\e^{-1+\delta(U_0)}}. \]

Finally, if $r\phi$ has a sign, as implied by the lower bound above, then $\pv(r\phi)$ has the opposite sign, and, thus, $C'$ can be replaced by $C_{\mathrm{in},\phi}^1$ in the estimates \eqref{eq:timelike:limit:theorem:upperboundfordvrphi}, \eqref{eq:timelike:limit:theorem:upperboundforrphi}.

As the main difficulty was dealing with the  bound for $\pu(r\phi)\leq |u|^{-p+1}$, which is borderline non-integrable for $p=2$, the cases where $p>2$ are strictly easier to deal with. This concludes the proof.
\end{proof}

\begin{rem}\label{rem:shortcoming of limiting proof}
The proof above has an obvious shortcoming, as manifested in the estimate \eqref{eq:timelike:limit:proof:borderlineterm}: Not only does it only work for $p\geq 2$ (as otherwise, the bound for $\pu(r\phi)$ wouldn't be integrable), but we even need to resort to some smallness condition for $C_{\mathrm{in},\phi}^1/R$.
Both of these difficulties can be traced back to the fact that we thus far have not shown sharp decay for $\pu(r\phi)$ (or for $\zeta$). Conversely, if we knew that 
\begin{equation}\label{eq:supposeddecayforpurphi}|\pu(r\phi)|\lesssim |u|^{-p}+r^{-2}|u|^{-p+1},\end{equation}
we could not only resolve the aforementioned difficulties, but the above proof would also work for all $p>1$ (as this would also allow to show better bounds for \eqref{eqq2}, \eqref{eqq222} and \eqref{eq:timelike:limit:proof:borderlineterm}).
We will show this improved decay for $\pu(r\phi)$ in section~\ref{sec:Refinements}, allowing us to remove both the smallness assumption \emph{and} the restriction on $p$.
\end{rem}

\begin{rem}\label{rem:shortcoming of limiting proof 2}
Even if one is happy with the smallness assumption (note that the smallness assumption is, in fact, not necessary in the uncoupled problem) and the restriction on $p$ -- after all, the case $p=2$ is the one we are most interested in anyway -- there is still one problem:
Even though we already know at this point that $|u|r\phi(u,v)$ remains bounded from above and away from zero, we cannot yet show that it takes a limit on $\mathcal I^-$.\footnote{Unless, of course, $r|_\Gamma\to\infty$ as $t\to \infty$, in which case we can simply integrate $|u|\pv(r\phi)$ from $\Gamma$ and use that $|u|r\phi(u,u)$ tends to a limit (which would have to be provided on data).} 
This means that, while we can show that $|\pv(r\phi)|\sim \log r/r^3$ near $\mathcal{I}^+$ by going through the same steps as in the proof of Thm.~\ref{thm:null:asymptotics of dvrphi}, we cannot write down precise asymptotics for $\pv(r\phi)$ yet, that is: We cannot yet say that $\pv(r\phi)(u,v)=C\log r/r^3 +\mathcal O(r^{-3})$ for some constant $C$. This problem will also be resolved in section~\ref{sec:Refinements}.
\end{rem}

\begin{rem}\label{rem:uniqueness}
Despite the shortcomings mentioned in Remarks~\ref{rem:shortcoming of limiting proof},~\ref{rem:shortcoming of limiting proof 2} (which we will fix anyway), the explicitness of the proof above allows one to directly extract a uniqueness statement (as claimed in Theorem \ref{thm.intro:timelikecase}) from it: 
Namely, one can extract that there exists a unique solution restricting correctly to the prescribed data on $\Gamma$ and $\mathcal I^-$. 
The precise class with respect to which this uniqueness holds can be read off from the proof. 
Let us here only give a brief sketch of how this works:

Assume that there are two smooth solutions $(r,\phi,m)_i$, $i=1,2$, defined on $\mathcal D_{U}$ which satisfy, for $u\leq U_0$, and $U_0$ a sufficiently large negative number, the following: 
The corresponding geometric quantities $\lambda_i, \nu_i, 1-\mu_i$ and $m_i$ are uniformly comparable to 1, and the scalar fields satisfy $|\pu(r\phi_i)|\leq C_{\text{small}} |u|^{-1}+Cr^{-1}|u|^{-\varepsilon}$ for some $\varepsilon>0$, some constant $C$, and a suitably small constant $C_{\text{small}}$ (cf.\ \eqref{eq:supposeddecayforpurphi}). 
Moreover, both solutions restrict correctly to $\hat\phi$, $\boldsymbol T\hat\phi$, $\hat r$ and $\boldsymbol T \hat r$ on $\Gamma$, and they satisfy $\lim_{\mathcal I^-}\pv r_i=1$, as well as $\lim_{\mathcal I^-}m_i=M$ and $\lim_{\mathcal I^-} r\phi_i=\lim_{\mathcal I^-}\pv(r\phi_i)=0$ on $\mathcal I^-$. Finally, $r_1-r_2$ is a bounded quantity that tends to 0 as $u\to-\infty$.

With these assumptions, we then let $k\in \mathbb N$ be arbitrary and split ${\mathcal D}_{U_0}$ into the subsets ${\mathcal D}_{U_0}\cap \{u\leq -k\}$ and ${\mathcal D}_{U_0}\cap \{u> -k\}$. 
In the former set, we can treat the difference $\Delta(r,r\phi,m)$ of the two solutions  as in the region $\mathcal D_{n,k}$ of the proof, leading to the estimate \eqref{eeeeeeeeeqqqqqqqqqqq} (with $\epsilon$ replaced by, say, $Ck^{-\varepsilon/2}$). 
In the latter set, we can then treat the difference as in the region ${\mathcal D}_k$ of the proof, leading to the estimate \eqref{eqq4} (with $\epsilon$ replaced by, say, $Ck^{-\varepsilon/4}$). Taking $k\to \infty$ then shows that the two solutions agree. 

In fact, if we also take into account the \textit{a priori estimate} of Remark~\ref{rem:betterbound} (which, combined with the energy estimate \eqref{eq:timelike:ee}, gives sharp decay for $r\phi$) as well as the corresponding \textit{a priori} arguments that can be extracted from section~\ref{sec:Refinements} (cf.\ Remark~\ref{rem:betterbound2}), then all of the above global assumptions can be recovered from the assumptions that the $m_i$ remain uniformly bounded. % and that $\nu_i<0<\lambda_i$, provided that one also assumes that $\lim_{u\to-\infty}\boldsymbol Tr_i=\lim_{u\to-\infty}\pv^2 r_i=\lim_{u\to-\infty}\pv^2(r\phi)_i=0$ for $i=1,2$. 
Thus, the solution constructed in Theorem~\ref{thm:timelike:final} is the unique solution restricting correctly to the data on $\Gamma$ and $\mathcal I^-$ that has a uniformly bounded Hawking mass.% as well as $\nu<0$, $\lambda>0$.
\end{rem}

\subsection{Refinements}\label{sec:Refinements}
\renewcommand{\T}{\boldsymbol{T}}
In this section, we will refine the above results and remove the unnecessary assumptions on $p$ and on the smallness of $C_{\mathrm{in},\phi}^1/R$ made thus far (see the Remarks~\ref{rem:shortcoming of limiting proof},~\ref{rem:shortcoming of limiting proof 2} in the previous section).
More precisely, we will show sharp decay for $\pu(r\phi)$ and, thus, of $\zeta$, hence allowing  us to take $p>1$ in the proofs of Thms.~\ref{thm:timelike:cc},~\ref{thm:timelike:final} and to remove the smallness assumption on $C_{\mathrm{in},\phi}/R$ in the latter.
Finally, in order to compute the asymptotics of $\pv(r\phi)$, we will also show that the limit $\lim_{u\to-\infty}|u|^{p-1}r\phi(u,v)$ exists. These refinements  will ultimately allow us to prove Theorems~\ref{thm:timelike:final!!!} and~\ref{thm:timelike:logs} in section~\ref{sec:refine:limitingargumentover}.

Let us briefly sketch the ideas going into the following proofs. For simplicity, assume for the moment that $\Gamma$ is a curve of constant $r=R$. 
The crucial observation is that, trivially, $\pu(r\phi)=\boldsymbol{T}(r\phi)-\pv(r\phi)$.
 Since we already know the sharp decay for $\pv(r\phi)$, it is thus left to show that the bound that $\boldsymbol{T}(r\phi)$ satisfies on $\Gamma$ is propagated outwards. 
 In the case of the linear wave equation on a fixed Schwarzschild background, this would be straight-forward; in fact, in that case, $\boldsymbol{T}$ commutes with the wave equation, and, hence, the bounds would propagate immediately by Thm.~\ref{thm:timelike:cc}.

This approach cannot directly be used in the coupled problem. 
However, similarly to how we proved decay for $r\phi$ in Theorem~\ref{thm:timelike:cc} (using a bootstrap argument), we can hope to prove decay for $\boldsymbol{T}r$ since $\boldsymbol{T}r$ and $r\phi$ satisfy similar wave equations. 
Modulo technical difficulties arising from all the error terms coming from commuting with $\boldsymbol{T}$, this decay for $\boldsymbol{T}r$ then allows us to prove better decay for $\boldsymbol{T}(r\phi)$ and, thus, for $\pu(r\phi)$.

Notice that this argument needs to be slightly modified in the case where $r\to\infty$ along $\Gamma$. There, we will additionally use that we can convert some $r$-decay along $\Gamma$ into $|u|$-decay as we did when integrating the wave equation from $\mathcal{C}_{\mathrm{in}}$ (see the proof of Thm.~\ref{thm:null boundedness of phi}).

Finally, in order to show that $|u|^{p-1}r\phi $ attains a limit, we consider its derivative
\begin{align*}
    \pu(|u|^{p-1}r\phi)&=|u|^{p-2}(-(p-1)r\phi+|u|\pu(r\phi))\\
        &=|u|^{p-2}(-(p-1)r\phi+|u|\boldsymbol{T}(r\phi))-|u|^{p-1}\pv(r\phi).
\end{align*}
The goal is to show that the above is integrable, knowing already that the $\pv(r\phi)$-term decays fast enough.
To achieve this, we will compute the wave equation satisfied by the difference $-(p-1)r\phi+|u|\boldsymbol{T}(r\phi)$ and show that if this difference decays like $|u|^{-p+1-\epsilon}$ on $\Gamma$, then we can perform a similar bootstrap argument for it as we did for $r\phi$ in order to propagate this decay outwards.

The remainder of section~\ref{sec:timelike} is structured as follows:
In section~\ref{sec:refine:tr}, we will assume (as a bootstrap assumption) sharp decay on $\T(r\phi)$ and then prove decay of $\T r$ and $\T m$ as a consequence. 
In section~\ref{sec:refine:purphi1}, we will then recover the decay assumption on $\T (r\phi)$, using the decay of $\T r$ and $\T m$ proved in the preceding section.
We will show decay for $-(p-1)r\phi+|u|\boldsymbol{T}(r\phi)$ in section~\ref{sec:refine:limit}. As all these proofs are based on bootstrap methods, we always need to work with compactly supported data. 
We will remove this assumption of compact support in section~\ref{sec:refine:limitingargument} by again performing a limiting argument as in the previous section. We finally derive the asymptotics of $\pv(r\phi)$ near $\mathcal I^+$ in section~\ref{sec:timelike:asymptotics}.

\subsubsection{Proving decay for \texorpdfstring{$\boldsymbol{T}r$}{Tr} and \texorpdfstring{$\boldsymbol{T} m$}{Tm}}\label{sec:refine:tr}
\renewcommand{\c}{C_{\mathrm{in},r}}
\renewcommand{\d}{d_{\mathrm{in},r}}
\newcommand{\C}{C'_{r}}
Ultimately, we will want to make the following natural extra assumptions (in addition to $\hat r>2M$) on the boundary data on  $\Gamma=\partial \mathcal D_u$ in each of the two following cases:\footnote{In principle, we can also deal with cases where $\hat{r}$ oscillates, but, in order to simplify the presentation, and because these cases are also not physically interesting, we avoid discussing them here.}
\paragraph{Case 1: $\hat{r}\to R <\infty$:}
In this case, we assume that there exists a constant $\c>0$ s.t.
\begin{align}\label{eq:refine:Case1}
    |\T \hat{r}|\leq \frac{\c}{|u|^{s}},&& s=1+\epsilon_r >1.
\end{align}
In fact, the upper bound could be replaced by any integrable function, but we here only write polynomial upper bounds in order to simplify notation. We do require the upper bound to be integrable, however.
\paragraph{Case 2: $ \hat{r}\to\infty$:}
In this case, we assume that
\begin{align}\label{eq:refine:Case2}
    -\T \hat{r}\sim \frac{1}{|u|^{s}},&& 1\geq s> 0,
\end{align}
i.e., we assume both upper \textit{and} lower bounds for $\T \hat{r}$. This means that either, for $s=1$, $\hat{r}\sim \log |u|$, or, for $s<1$, that $\hat{r}\sim |u|^{\epsilon}$ for $\epsilon=1-s$. The reason why we need the lower bound is that, in the case where $r|_\Gamma$ tends to infinity, we also need to convert some of the $r$-weights on $\Gamma$ into $|u|$-weights.
Again, the above bounds can, in principle, be replaced by any non-integrable function.
\begin{rem}
Note that if $s\leq\frac12$, then the results of the following theorems follow trivially. Moreover, if $\hat r\sim |u|$, then, modulo the local existence part, we can apply the methods of section~\ref{sec:null}.
\end{rem}

We will now prove decay for $\T r$ in each of these two cases. We will find that its decay is dictated by its initial decay on $\Gamma$ and the decay of the scalar field.  Thus, in order to capture the sharp decay of $\boldsymbol T r$, we will, in the theorem below, \textit{assume} the sharp decay for $\T(r\phi)$. This assumption will be recovered in Thm.~\ref{thm:refinements:TPhi}.

%As alluded to above, we want to perform a bootstrap argument similar to the one in the proof of Thm.~\ref{thm:timelike:cc}, but for $\T r$ instead of $r\phi$. This, again, comes with the technicality that we need to consider a compact region, and extend to the past by the vacuum solution. We will remove this assumption in section~\ref{sec:refine:limitingargument}. 

 \begin{thm}\label{thm:refinements:TR}
 
   %  Let $\mathcal M_M$ be as described in section~\ref{sec:ambient}, let $\Gamma$ be a smooth timelike curve in $\mathcal M_M$, and equip $\mathcal M_M$ with $(u,v)$-coordinates as in  section~\ref{sec:ambient}.
    Let $\mathcal D_U$ be as described in section~\ref{sec:ambient}, and
     specify smooth functions $\hat r$, $\hat \phi$ on $\Gamma=\partial D_U=\{(u,u)\in\mathcal D_U\}$, with $\hat\phi$ having compact support.
     Let $\mathcal{C}_{u_0}$ denote the future-complete outgoing null ray emanating from a point $q=(u_0,u_0)$ on $\Gamma$ that lies to the past of the support of $\hat{\phi}$.
    %Specify smooth functions $\hat r$, $\hat \phi$ on $\Gamma$, with $\hat\phi$  having compact support.
%     Let $\mathcal{C}_{u_0}$ denote a future-complete outgoing null ray emanating from a point $q=(u_0,u_0)$ on $\Gamma$ that lies in the past of the support of $\hat{\phi}$. 
     On $\mathcal{C}_{u_0}$, specify $\bar{m}\equiv M>0$,  $\bar{\phi}\equiv 0$, and an increasing smooth function $\bar{r}$ defined via
   $\bar r(v=u_0)=\hat r(u=u_0)$ and the ODE
     $$\pv \bar{r}=1-\frac{2M}{\bar{r}}.$$
     Finally, assume that (denoting, again, the generator of $\Gamma$ by $\boldsymbol{T}=\pu+\pv$) the following bounds hold on $\Gamma$:
     \begin{align}
            |\boldsymbol{T}(\hat r\hat\phi)(u)|&\leq C_{\mathrm{in},\phi}^1|u|^{-p},\\
             |\boldsymbol{T}\hat r (u)|&\leq C_{\mathrm{in},r}|u|^{-s},
    \end{align}
    with positive constants $p>1$,\footnote{Note that this was $p>3/2$ in Thm.~\ref{thm:timelike:cc}.} $C_{\mathrm{in},\phi}^1$, $C_{\mathrm{in},r}$, and $s>0$; and assume that $\hat{r}$ tends to either an infinite (in the case $s\leq 1$) or a finite (in the case $s>1$) limit $R\geq 4M$.     
    
Let $\Delta_{u_0,\epsilon}$ denote the region of local existence from Proposition~\ref{prop:localexistence}. If the assumption 
\begin{equation}\tag{BS(4)}\label{eq:refine:bootstrapfortrphi1}
   | \T (r\phi)|\leq \frac{E}{|u|^p}
\end{equation}
  is satisfied throughout $\Delta_{u_0,\epsilon}$ for some positive constant $E$ ,
     then we have, for sufficiently large negative values of $U_0$ (the choice of $U_0$ depending only on data),
     that, throughout  $\Delta_{u_0,\epsilon}\cap\{u\leq U_0\}$ the estimates of Theorem~\ref{thm:timelike:cc} hold for the arising solution $(r,\phi,m)$.
      Moreover, we have the additional bounds:\footnote{Compare these bounds to the ones for $r\phi$ and $\pv(r\phi)$ in Thm.~\ref{thm:timelike:cc}.}
  \begingroup\allowdisplaybreaks
\begin{align}
    |\T r|\leq \C |u|^{-\min(s,2p-1)},  \label{thm.refinements.boundontr}\\
    |\pv\T r| \leq D' M \C \frac{|u|^{-\min(s,2p-1)}}{r^2},\label{thm.refinementes.boundonpvtr}\\
    |\T m|\leq C_m\left( \frac{1}{|u|^{2p-1}r}+\frac{1}{|u|^{2p}}\right).\label{thm.refinements.Tm}
\end{align}\endgroup
Here, $\C$ and $D'$ are constants which, for $2p-1>s$, only depend on the value of $\c$ and the ratio $M/R$ and, for $2p-1\leq s$, also depend on $E$, whereas $C_m$ always depends also on $E$.  In particular, none of these constants depend on $u_0$.
\end{thm}

\begin{rem}\label{rem:betterbound2}
Again, the lower bound for $R$ stated here is not necessary -- one can prove the same result under the weaker assumption $R>2M$ if one replaces the bootstrap arguments (\eqref{eq:refine:bootstrapTr}, \eqref{eq:refine:bootstrappvTr}) below with Gr\"onwall arguments as described in Remark~\ref{rem:betterbound}. See footnote~\ref{fn:betterbound}. The same applies to the subsequent results (Theorem~\ref{thm:refinements:TPhi} and Lemma~\ref{lemma:trick}).
\end{rem}
\begin{rem}
In order to recover bootstrap assumption \eqref{eq:refine:bootstrapfortrphi1} in the next section, it is helpful to observe that $\C$ and $D'$ are independent of $E$ for $2p-1>s$. The distinction between the cases $2p-1>s$ and $2p-1\leq s$ is best understood by looking at the estimate \eqref{thm.refinements.boundontr}: If $\phi$ decays sufficiently slowly, then its decay dominates the decay of $\T r$.
\end{rem}

\begin{proof}
Let's first draw our attention to the claim that $p$ can be taken to be $p>1$ instead of $p>3/2$.
The reason for this is that we assume $\T(r\phi)\lesssim |u|^{-p}$. Recall that, in the proof of Thm.~\ref{thm:timelike:cc}, the obstruction to taking $p>1$ was that we were only able to show that $\zeta$ decayed as fast as $r\phi$, see Remark~\ref{rem:pgreater32}.
However, if we assume ineq.~\eqref{eq:refine:bootstrapfortrphi1}, then we immediately get that
$$\pu(r\phi)=\T(r\phi)+\pv(r\phi)\lesssim \frac{1}{|u|^p}+\frac{1}{|u|^{p-1}r^2},$$
which improves the bound on $\zeta$ to 
$$|\zeta|\lesssim \frac{1}{|u|^{p-1}r}+\frac{1}{|u|^{p}} ,$$
so we can also prove Theorem~\ref{thm:timelike:cc} for $p>1$.
Therefore, we now assume the results of Thm.~\ref{thm:timelike:cc} to hold for $p>1$. In fact, because of this improved bound for $\zeta$, we can now take $\delta(u)\sim|u|^{2-2p}$ instead of $\delta(u)\sim|u|^{3-2p}$ (cf.~\eqref{eq:bootstrapproofkappadelta}).

Furthermore, observe that, along $\mathcal{C}_{u_0}$, we have
\begin{equation}\label{eq:proof:tinygroenwall}
    \pv\T r=(\pv r+\pu r)\frac{2M}{r^2}=\frac{2M}{r^2}\T r
\end{equation}
in view of $\pv\bar{r}=1-2M/\bar{r} $ and the wave equation for $r$ \eqref{eq:pupvr}.
\paragraph*{\underline{Outline of the main ideas:}}
With these preliminary observations understood, we now give an outline of the heart of the proof:
If we consider the set
 \begin{nalign}
     \Delta:=\{(u,v)\in \Delta_{u_0, \epsilon}\,|\,|\T r(u',v')|\leq  \frac{\C}{ |u'|^{\min(s,2p-1)} }
     \,\forall (u',v') \in \Delta_{u_0, \epsilon} \,\,\text{s.t. }\, u'\leq u, v'\leq v\} ,
\end{nalign}
for some suitable $\C>\c$ (which will be specified later)  and $\epsilon>0$, then
we immediately conclude that this is non-empty ($\{q\}\subsetneq\Delta$) by continuity and compactness.
As in the proof of Thm.~\ref{thm:timelike:cc}, we want to show that $\Delta$ is open (it is clearly closed), that is, we want to improve the bound 
\begin{equation}\tag{BS(5)}
|\T r|\leq \C |u|^{-\min(s,2p-1)} \label{eq:refine:bootstrapTr}
\end{equation}
inside $\Delta$.

It turns out that we need to include another bootstrap assumption inside $\Delta$ in order for this to work, namely
\begin{equation}\tag{BS(6)}
|\pv\T r|\leq D \frac{|u|^{-\min(s,2p-1)}}{r^2} \label{eq:refine:bootstrappvTr}
\end{equation}
for some suitable constant $D$.
To see why, commute the wave equation for $r$, eq.~\eqref{eq:pupvr}, with the vector field $\T$: 
\begin{align}\label{eq:pupvTr}
\begin{split}
    \pu\pv(\T r) &= \T m \left( \frac{2\nu\kappa}{r^2}    \right)+\T\mu \left(\frac{\kappa}{1-\mu}\frac{2m\nu}{r^2}\right)\\
                  & +\T \lambda\left(\frac{2m\nu}{r^2(1-\mu)}\right) + \T \nu \left(\frac{2m\lambda}{r^2(1-\mu)}\right)
%                  \\
%                  &
                   -\T r \left( \frac{4m\nu\kappa}{r^3}\right).
\end{split}  
\end{align}
Looking at the $\T m$-term, we will show below that the bootstrap assumptions \eqref{eq:refine:bootstrapfortrphi1} and \eqref{eq:refine:bootstrapTr} together imply that
\begin{equation}\label{eq:temp1}
\T m\lesssim \frac{1}{|u|^{2p-1}r}+\frac{1}{|u|^{2p}},
\end{equation}
which means that this term is in some sense independent of the bootstrap argument.
 In fact, it is this term which is responsible for the $\min(s,2p-1)$ appearing in the bootstrap assumptions. Note already that ineq.~\eqref{eq:temp1} is an improvement over plugging in the naive bounds for $\zeta^2, \theta^2$ into $\T m$.

On the other hand, the $\T r$-term in \eqref{eq:pupvTr} comes with a factor $1/r^3$, so we expect to be able to treat it exactly like we treated $r\phi$ in Thm.~\ref{thm:timelike:cc}. The $\T\mu$-term can be written as a sum of faster decaying $\T m$ and $\T r$-terms.
As for the $\T\lambda$-term, we expect to be able to bound it in the same way as we bounded $\pv(r\phi)$ in Thm.~\ref{thm:timelike:cc} (this explains the bootstrap assumption \eqref{eq:refine:bootstrappvTr}), so the non-integrable factor of $1/r^2$ multiplying $\T\lambda$ poses no problem.

However, for the $\T \nu$-term, it seems like we're in trouble, as we cannot expect to show a similar bound for $\T\nu$ as for $\T\lambda$, for the same reason for which we weren't able to show a better bound for $\pu(r\phi)$ in Thm.~\ref{thm:timelike:cc}.
We deal with this problem by making the $\T \nu$-term into a boundary term (using that $\T\nu=\pu\T r$), that is, we will write, inserting also the equations for $\pu\kappa$, $\pu\lambda$ and $\pu m$:
\begin{align}\label{eq:refine:pupvtr}
\begin{split}
    \pu\pv(\T r) &= \T m \left( \frac{2\nu\kappa}{r^2}+\frac{4m\nu\kappa}{r^3(1-\mu)}    \right)\\
                  &+\T \lambda\left(\frac{2m\nu}{r^2(1-\mu)}\right) + \pu\left(2m\kappa\frac{\T r}{r^2}\right)\\
                  & +\T r \left( \frac{4m^2\nu\kappa}{r^4(1-\mu)}-\frac{2\kappa}{r^2}\frac{(1-\mu)\zeta^2}{2\nu}-\frac{2m}{r^2}\frac{\kappa\zeta^2}{r\nu} \right).
\end{split}  
\end{align}

\paragraph*{\underline{The details:}}
Having given a rough outline of how we will deal with each term in the above equation, we now give the details. First, we derive the estimate \eqref{eq:temp1}: Plugging in eqns.~\eqref{eq:puvarpi} and \eqref{eq:pvvarpi}, we get
\begin{align}\label{eq:temp2}
    \begin{split}
        \T m&= \pu m+\pv m=\frac12(1-\mu)\left(\frac{\zeta^2}{\nu}+\frac{\theta^2}{\lambda}\right)\\
            &=\frac12(1-\mu)\left( \frac{(\pv(r\phi)-\lambda\phi)^2}{\lambda}+\frac{(\T(r\phi)-\pv(r\phi)-\lambda\phi \frac{\nu}{\lambda})^2}{\lambda}\frac{\lambda}{\nu}         \right).
    \end{split}
\end{align}
Now, by the bootstrap assumption \eqref{eq:refine:bootstrapTr}, we have that
$$\frac{\nu}{\lambda}=\frac{\T r -\lambda}{\lambda}=-1+\mathcal{O}(|u|^{-s}).$$
Upon inserting this back into equation \eqref{eq:temp2}, we get\footnote{Again, constants hidden inside $\lesssim$ only depend on initial/boundary data.}
\begin{align}\label{5.106}
    \begin{split}
        \T m&=\frac{1-\mu}{2\lambda}\left( -(\T(r\phi))^2+2\T(r\phi)(\pv(r\phi)-\lambda \phi +\mathcal{O}(|u|^{-s}))+\zeta^2\mathcal{O}(|u|^{-s})   \right)\\
            &\lesssim \frac{E^2}{|u|^{2p}}+\frac{E}{|u|^{2p-1}r}+\frac{\C E}{|u|^{2p-1}r}\frac{1}{|u|^{s-1}r}.
    \end{split}
\end{align}
Now, if $s>1$, then the third term on the RHS of \eqref{5.106} has more $|u|$-decay than the second term and  can hence be ignored. 
If $s=1$, the third term only has more $r$-decay, but in this case $r|_\Gamma\sim \log|u|$ by the lower bound in \eqref{eq:refine:Case2}, so it can again be ignored.
If $s<1$, then, again by the lower bound in \eqref{eq:refine:Case2},
$$\frac{1}{|u|^{s-1}r}\leq 1,$$
so the third term decays just as fast as the second one, and it cannot be ignored. 
We conclude that 
\begin{equation}\label{eq:refine:boundforTM}
    |\T m|\lesssim_{\C,E} \frac{1}{|u|^{2p-1}r}+\frac{1}{|u|^{2p}}.
\end{equation}
It is important that the implicit constant in $\lesssim_{\C,E}$ only depends on $\C$ if $r|_\Gamma\to\infty$. We shall return to this point later.

We now have all the tools to close the bootstrap argument, i.e.\ to improve assumptions \eqref{eq:refine:bootstrapTr} and \eqref{eq:refine:bootstrappvTr}.
The idea is to integrate the wave equation \eqref{eq:refine:pupvtr} for $\T r$ twice along its characteristics. 
We consider the cases $2p-1>s$, $2p-1=s$ and $2p-1<s$ separately.

\paragraph{Case i): $2p-1> s$:}
 Let us also assume, for simpler presentation, that $\T\hat{r}$ is compactly supported such that $\T\hat{r}(q)=0$. This will just mean that we won't pick up a boundary term when integrating $\pu\pv\T r$ from $\mathcal{C}_{u_0}$, in view of \eqref{eq:proof:tinygroenwall}. We will remove this assumption below.
Integrating equation \eqref{eq:refine:pupvtr} from data, we thus obtain 
\begin{align}
    |\pv\T r(u,v)| = \frac{2m\kappa\T r}{r^2}+\int_{u_0}^u \T m\left(\dots \right)+\T\lambda\left(\dots \right)+\T r\left(\dots \right)\dd u'.
\end{align}

For the $\T m$-term, we plug in the bound from \eqref{eq:refine:boundforTM}, resulting in 
$$\int_{u_0}^u \T m\left( \frac{2\nu\kappa}{r^2}+\frac{4m\nu\kappa}{r^3(1-\mu)}    \right)\dd u'\lesssim_E \frac{1}{r^2|u|^{2p-1}}. $$

For the other terms, we first use $\nu$ to turn the $|u|$-integration into $r$-integration. 
We then plug in the bootstrap assumptions \eqref{eq:refine:bootstrapTr}, \eqref{eq:refine:bootstrappvTr}, as well as all the bounds from Thm.~\ref{thm:timelike:cc}; in particular, we use that $\kappa\leq1$, $m\leq M$ and\footnote{The reader can instead just take the lower bound $ 1-\frac{2M}{R}\leq \mu$. This will simplify the integrals below but lead to a slightly worse lower bound on $R$. }  $\mu-(1-2M/r)=\mathcal{O}(|u|^{-2p+1})$.\footnote{\label{fn:betterbound}Similarly to the calculation in Remark~\ref{rem:betterbound}, one can, instead of using bootstrap arguments and going through the calculations below, first apply a Gr\"onwall estimate to $\pv\T r$ in the $u$-direction. From this, one can then obtain an inequality for $\T r$ similar to \eqref{eq:gr1} and apply a similar double Gr\"onwall estimate to it. This only requires the lower bound that $R>2M$.}
This yields
\begin{align}\label{eq:refine:integralofpupvtr}
    \begin{split}
        &|\pv\T r(u,v)| \\
        \leq& \frac{2m\kappa \C}{r^2|u|^s}+\int_{r(u_0,v)}^{r(u,v)} \left|\T r\frac{4m^2\kappa}{r^3(r-2M)}\right|+\left|\T \lambda\frac{2m}{r(r-2M)}\right|\dd r+\mathcal{O}\left(\frac{1}{r^{2}|u|^{2p-1}}\right)\\
                        \leq&  \frac{2M\C}{r^2|u|^s}+\frac{2MD+4M^2 \C}{|u|^{s}}\left(\frac{-\log(1-\frac{2M}{r})}{(2M)^3}-\frac{r+M}{r^2(2M)^2}\right)+\mathcal{O}\left(\frac{1}{r^{2}|u|^{2p-1}}\right).
    \end{split}
\end{align}
\renewcommand{\O}{\mathcal{O}}In order to close the bootstrap argument for $\pv\T r$ \eqref{eq:refine:bootstrappvTr}, we require the RHS to be strictly smaller than $\frac{D}{r^2|u|^s}$.
This leads to the following condition on $D$ and $\C$:
\begin{align}\label{eq:refine:DvsC}
    2M\C\frac{(1+\frac{1}{x^2}\log(\frac{1}{1-x})-\frac{1}{x}-\frac12)}{(1-\frac{1}{x^2}\log(\frac{1}{1-x})+\frac{1}{x}+\frac12)}+\mathcal{O}(|u|^{s+1-2p})< D,
\end{align}
where we wrote $x=\frac{2M}{r}$. The LHS is maximised when $x$ is, so it is maximised for $x=\frac{2M}{r(u,u)}$.
In the case where $r(u,u)\to\infty$ as $u\to -\infty$, \eqref{eq:refine:DvsC} is trivially satisfied for large enough values of $|u|$. 

On the other hand, if $r(u,u)\to R <\infty$, we have $r(u,u)\geq R-\O(|u|^{1-s})$. We can thus insert $x=\frac{2M}{R}$ into \eqref{eq:refine:DvsC}, resulting in another $o(1)$-term. 
The estimate \eqref{eq:refine:DvsC} then holds provided that $|U_0|$ is large enough, that $2M/R\lesssim 0.86$, and for
\begin{equation} \label{eq:temp3}
   D=\eta 2M \C\cdot  \frac{(1+\frac{1}{x^2}\log(\frac{1}{1-x})-\frac{1}{x}-\frac12)}{(1-\frac{1}{x^2}\log(\frac{1}{1-x})+\frac{1}{x}+\frac12)}=:\eta 2M \C\cdot A(x),
  %  D=\eta 2M \C \frac{(1+\frac{1}{x^2}\log(\frac{1}{1-x})-\frac{1}{x}-\frac12)}{(1-\frac{1}{x^2}\log(\frac{1}{1-x})+\frac{1}{x}+\frac12)}=:\eta 2M \C \cdot A(x),
\end{equation}
where $x=\frac{2M}{R}$ and $\eta>1$ (where $\eta-1\to 0$ as $U_0\to-\infty$).

This improves the bootstrap assumption \eqref{eq:refine:bootstrappvTr}.

Next, we integrate the estimate \eqref{eq:refine:integralofpupvtr} in $v$  from $\Gamma$. In order to convert the $v$-integration into $r$-integration, we use the estimate $\lambda-(1-2M/r)=\mathcal{O}(|u|^{-2p+1})$. One obtains
\begin{align}\label{eq:refine:integralofpvtr}
    \begin{split}
        |\T r(u,v)|\leq  &\frac{\c}{|u|^s} +\frac{1}{|u|^s}\int_{r(u,u)}^{r(u,v)}\frac{2M\C}{ r(r-2M)}\dd r\\
                                    +& \frac{\frac{D}{2M}+\C}{|u|^s}\int_{r(u,u)}^{r(u,v)}\frac{-\log(1-\frac{2M}{r})}{2M(1-\frac{2M}{r})}-\frac{r+M}{r(r-2M)}\dd r+\mathcal{O}(|u|^{1-2p}).
    \end{split}
\end{align}
The first integral has been computed before (cf.~\eqref{integralthathasbeencomputedbefore}). For the second integral, we substitute $x=2M/r$, which brings it into the following form
\begin{align}
    \begin{split}
        \int_{2M/r(u,u)}^{2M/r(u,v)}\frac{\log(1-x)}{(1-x)x^2}+\frac{1}{x}\frac{1+\frac12 x}{1-x}\dd x.
    \end{split}
\end{align}
This can now be computed using the dilogarithm  $\mathrm{Li}_2(x)$ (cf.~\eqref{dilog}). Reinserting this back into \eqref{eq:refine:integralofpvtr}, using also that $\mathrm{Li}_2(1)=\pi^2/6$, we get the estimate:
\begingroup\allowdisplaybreaks
\begin{nalign}\label{eq:refine:integralofpvtr2}
   % \begin{split}
     & |u|^s  |\T r(u,v)|\leq  \c +\C\log\(\frac{1}{1-x}\)\\
     & +\left(\frac{D}{2M}+\C\right)\underbrace{\left(  1+\frac{\pi^2}{6}+\frac12 \log(1-x)\left(\frac{2}{x}+\log\(\frac{1-x}{x^2}\)\right)-   \mathrm{Li}_2(1-x) \right)}_{:=B(x)}+\mathcal{O}(|u|^{1-2p})     
   % \end{split}
\end{nalign}\endgroup
 where we wrote $x=2M/r(u,u)$:
As before, we want this to be strictly smaller than $\C|u|^{-s}$ in order to close the bootstrap assumption \eqref{eq:refine:bootstrapTr}. This can trivially be achieved for large enough $|U_0|$ in the case where $r(u,u)\to\infty$. 
In the case where $r(u,u)\to R$, we numerically find, plugging in \eqref{eq:temp3}, that  the lower bound $2M/R\lesssim 0.516$ needs to hold.
The constant $\C$ can then be chosen to be
\begin{align}
    \begin{split}
        \C=\frac{\eta \c }{1+\log(1-2M/R)-(1+\eta A(2M/R))B(2M/R)},
    \end{split}
\end{align}
where $B(x)$ was defined in the estimate above and $A(x)$ was defined in \eqref{eq:temp3}. This bound is, in particular, independent of $E$.

This closes the bootstrap argument for $2p-1>s$ in the case of $\T\hat{r}$ being compactly supported. 

If $\T\hat{r}$ is not compactly supported, the only difference is that we pick up a boundary term $\pv(\T r)(u_0,v)$ from integrating $\pu\pv\T r$. In view of equation \eqref{eq:proof:tinygroenwall}, this boundary term can be bounded directly in terms of initial data (after applying a Gr\"onwall inequality to bound $\T r$ on $\mathcal{C}_{u_0}$). This boundary term will slightly change the definitions of $D$ and $\C$, but will not affect the lower bound on $R$ in any way, precisely because it is bounded by data!
\paragraph{Case ii): $2p-1=s$:}
In this case, it seems like there is an additional difficulty since the $\O(r^{-2}|u|^{-2p+1})$-term we treated as negligible in estimate \eqref{eq:refine:integralofpupvtr} is now of the same order as the other terms: we therefore need to add a term $\lesssim_{\C,E}\frac{1}{r^2|u|^s}$ to the estimate \eqref{eq:refine:integralofpupvtr}.
The $E$-dependence of the implicit constant in $\lesssim_{\C,E}$ then means that $\C$ and $D$ will depend on $E$. The $\C$-dependence in $\lesssim_{\C,E}$, on the other hand, seems like it would add a further restriction on the lower bound of $R$. This is where it is important that the implicit constant in $\lesssim_{\C,E}$ only depends on $\C$ in the case where $r(u,u)\to\infty$ (see the remark below estimate \eqref{eq:refine:boundforTM}). Therefore, no new bound on $R$ is introduced, and the bootstrap argument works in the same way, with $D$ and $\C$ now depending also on $E$.

\paragraph{Case iii): $2p-1<s$:}
In this case, the $\O(r^{-2}|u|^{-2p+1})$-term we treated as negligible in estimate \eqref{eq:refine:integralofpupvtr} now dominates all other terms and depends on $E$ as well as on $\C$ if $r(u,u)\to\infty$. By the same reasoning as above, the bootstrap argument then closes trivially, with $\C$ and $D$ again depending on $E$.

This concludes the proof.
\end{proof}

We will now recover the bootstrap assumption \eqref{eq:refine:bootstrapfortrphi1}:

\subsubsection{Sharp decay for \texorpdfstring{$\T(r\phi)$}{T(r phi)} and \texorpdfstring{$\pu(r\phi)$}{d/du(r phi)}}\label{sec:refine:purphi1}
In this section, we prove the next refinement to Thm.~\ref{thm:timelike:cc}. This refinement shows sharp decay for $\T(r\phi)$ and, thus, for $\pu(r\phi)$.

\begin{thm}\label{thm:refinements:TPhi}
%     Let $\mathcal M_M$ be as described in section~\ref{sec:ambient}, let $\Gamma$ be a smooth timelike curve in $\mathcal M_M$, and equip $\mathcal M_M$ with $(u,v)$-coordinates as in  section~\ref{sec:ambient}.
   %  Specify smooth functions $\hat r$, $\hat \phi$ on $\Gamma$, with $\hat\phi$ having compact support.
    % Let $\mathcal{C}_{u_0}$ denote a future-complete outgoing null ray emanating from a point $q=(u_0,u_0)$ on $\Gamma$ that lies in the past of the support of $\hat{\phi}$.
     Let $\mathcal D_U$ be as described in section~\ref{sec:ambient}, and
     specify smooth functions $\hat r$, $\hat \phi$ on $\Gamma=\partial D_U=\{(u,u)\in\mathcal D_U\}$, with $\hat\phi$ having compact support.
     Let $\mathcal{C}_{u_0}$ denote the future-complete outgoing null ray emanating from a point $q=(u_0,u_0)$ on $\Gamma$ that lies to the past of the support of $\hat{\phi}$.
     On $\mathcal{C}_{u_0}$, specify $\bar{m}\equiv M>0$,  $\bar{\phi}\equiv 0$, and an increasing smooth function $\bar{r}$ defined via
   $\bar r(v=u_0)=\hat r(u=u_0)$ and the ODE
     $$\pv \bar{r}=1-\frac{2M}{\bar{r}}.$$
     Finally, assume that the following bounds hold on $\Gamma$:
     \begin{align}
            |\boldsymbol{T}(\hat r\hat\phi)(u)|&\leq C_{\mathrm{in},\phi}^1|u|^{-p},\\
             |\boldsymbol{T}\hat r (u)|&\leq C_{\mathrm{in},r}|u|^{-s},
    \end{align}
    with positive constants $p>1$,  $C_{\mathrm{in},\phi}^1$, $C_{\mathrm{in},r}$  and $s>0$, and assume that $\hat{r}$ tends to either an infinite (in the case $s\leq 1$) or a finite (in the case $s>1$) limit $R\geq 4M$.     
    If $s\leq1$, we moreover assume that there exists a positive constant $d_{\mathrm{in},r}<C_{\mathrm{in},r}$ such that
    \begin{equation}
         -\T \hat r (u)\geq d_{\mathrm{in},r}|u|^{-s}.
    \end{equation}
      
Then we have, for sufficiently large negative values of of $U_0$ (the choice of $U_0$ depending only on data), that, throughout  $\Delta_{u_0,\epsilon}\cap\{u\leq U_0\}$,  the estimates of Theorem~\ref{thm:timelike:cc} hold.
Moreover, we have the following additional bounds\footnote{Compare these bounds to the ones for $r\phi$ and $\pv(r\phi)$ in Thm.~\ref{thm:timelike:cc}.}:
\begingroup\allowdisplaybreaks
\begin{align}
    |\T (r\phi)|&\leq C'_{\boldsymbol{T}} |u|^{-p} , \label{thm.refinements.boundontphi}\\
    |\pv\T (r\phi)| &\leq D_{\boldsymbol{T}} M C'_{\boldsymbol{T}} \frac{|u|^{-p}}{r^2}.\label{thm.refinementes.boundonpvtphi}
\end{align}\endgroup
Here, $C'_{\boldsymbol{T}}$ and $D_{\boldsymbol{T}}$ are constants which depend on the value of $C_{\mathrm{in},\phi}^1$ and the ratio $M/R$; in particular, they do not depend on $u_0$.

Finally, in view of \eqref{thm.refinements.boundontphi}, the estimates from  Theorem~\ref{thm:refinements:TR} hold as well, with the constant $E$ given by $E=C'_{\boldsymbol{T}}$.
\end{thm}

\begin{proof}
As in the previous proof, we will bootstrap the decay of $\T(r\phi)$, that is, we will assume
\begin{equation}\tag{BS(4)}\label{eq:refine:bootstrapTPHI2}
    |\T(r\phi)(u,v)|\leq E |u|^{-p}
\end{equation}
for some suitable constant $E$, and we will subsequently improve this assumption. Note that, by the bootstrap assumption, the results of Thm.~\ref{thm:refinements:TR} hold. 

We will distinguish between the cases $s>1$, $s=1$ and $s<1$, i.e.\ between the cases where $r|_\Gamma$ tends to a finite or infinite limit.

\paragraph{Case i): $s>1$:}
We start by commuting the wave equation for $r\phi$ with $\T$. As in the previous proof, we deal with the bad $\T\nu$-term by converting it into a boundary term (cf.\ eq.~\eqref{eq:refine:pupvtr}):
\begin{nalign}\label{eq:refine:pupvtPHI}
%\begin{split}
    \pu\pv(\T (r\phi)) &= \T m \left( \frac{2\nu\kappa}{r^2}+\frac{4m\nu\kappa}{r^3(1-\mu)}    \right)\frac{r\phi}{r}\\
                  &+\T \lambda\left(\frac{2m\nu}{r^2(1-\mu)}\right) \frac{r\phi}{r}+ \pu\left(2m\kappa\frac{\T r}{r^2}\frac{r\phi}{r}\right)\\
                  & +\T r \left( \frac{4m^2\nu\kappa}{r^4(1-\mu)}-\frac{2\kappa}{r^2}\frac{(1-\mu)\zeta^2}{2\nu}-\frac{2m}{r^2}\frac{\kappa\zeta^2}{r\nu} \right)\frac{r\phi}{r}
                  +\T r\left(  -\frac{2m\kappa}{r^2} \frac{\pu(r\phi)}{r}\right)   \\
                  &+\pu\pv r \cdot \frac{\T(r\phi)}{r}   .
%\end{split}  
\end{nalign}
The last term is exactly the same term that appears in $\pu\pv(r\phi)$, but with $r\phi$ replaced by $\T (r\phi)$. We will show that all other terms decay faster in $|u|$. More precisely, we will show that all other terms can be bounded by $\frac{1}{|u|^{p+\epsilon}r^{3}}$ for some $\epsilon>0$. 

For the $\T m$-term, plugging in the bound \eqref{thm.refinements.Tm} as well as $r\phi\lesssim |u|^{-p+1}$, we find that it is bounded by
$$ \T m\left(\dots \right)\frac{r\phi}{r}\lesssim \frac{1}{|u|^{3p-2}r^4}+\frac{1}{|u|^{3p-1}r^3}.$$
The RHS  can be bounded by $\frac{1}{|u|^{p+\epsilon}r^{3}}\left(\frac{1}{|u|}+\frac{1}{r}\right)$, where $\epsilon$ is given by $\epsilon=2p-2>0$.

The $\T\lambda$- and the $\T r$-terms can be dealt with similarly in view of the bounds \eqref{thm.refinements.boundontr} and \eqref{thm.refinementes.boundonpvtr} and since we assumed that $s>1$ (and since $p>1$ implies $2p-1>1$).

For the boundary term in the second line, we find that, after integrating first in $u$ and then in $v$, it can be bounded against $R^{-2}|u|^{-p+1-\min(s,2p-1)}$. 

In conclusion, we find that 
\begin{multline}
    %\begin{split}
     |\T(r\phi)|(u,v)\leq |\T(r\phi)|(u,u)+\left| \int_{u}^v dv'\int_{u_0}^u du'\, \pu\pv(\T r\phi) \right| \\
                    \leq C_{\mathrm{in},\phi}^1 |u|^{-p} + \left| \int_{u}^v dv'\int_{u_0}^u du' \,\pu\pv r \frac{\T(r\phi)}{r^3} \right|+\mathcal{O}(|u|^{-p -\min(2p-2,s-1)}),
    %\end{split}   
\end{multline}
so the bootstrap argument can be closed in the same manner as in the proof of Thm.~\ref{thm:timelike:cc} for $r\phi$. (Alternatively, one can perform a Gr\"onwall argument as in  Remark~\ref{rem:betterbound}).

This concludes the proof in the case where $r(u,u)$ tends to a finite limit.
\paragraph{Case ii): s=1:}
In this case, the $\T r$- and $\T\lambda$-terms in eq.~\eqref{eq:refine:pupvtPHI} are no longer subleading compared to the $\T(r\phi)$-term (the $\T m$-term remains unchanged). Nevertheless, since in this case $r(u,u)\sim \log|u|$ diverges, the bootstrap argument still closes.
\paragraph{Case iii): $s\leq1$:}
 Let us finally deal with the case where $r(u,u)\sim |u|^{1-s}$. Here, the $\T r$- and $\T\lambda$-terms in equation \eqref{eq:refine:pupvtPHI} exhibit less decay in $u$ than the other terms, however, we can convert some of the extra $r$-decay present in these terms into $u$-decay according to
 \begin{align}\label{eq:temp5}
     r^{-1}\lesssim |u|^{s-1}.
 \end{align}
 Using this, we have e.g.\  for the $\T\lambda$-term in \eqref{eq:refine:pupvtPHI}
 \begin{equation}
     \left|\T \lambda\left(\frac{2m\nu}{r^2(1-\mu)}\right) \frac{r\phi}{r}\right|\lesssim \frac{1}{r^5|u|^s}\frac{1}{|u|^{p-1}}\lesssim\frac{1}{r^3}\frac{1}{|u|^{p+1-s}}.
 \end{equation}
 
 The boundary term in the second line, as well as the third line of \eqref{eq:refine:pupvtPHI}, can be dealt with in a similar fashion.
 For the term in the fourth line, 
we use the bootstrap assumption \eqref{eq:refine:bootstrapTPHI2} as well as the bound for $\pv(r\phi)$ from \eqref{eq:timelike:thm:BS:pvrphi} to conclude that
$$|\pu(r\phi)|\lesssim \frac{1}{|u|^p}+\frac{1}{|u|^{p-1}r^2}.$$
Plugging this bound back into the above, we see that
\begin{equation*}
     \T r\left(  -\frac{2m\kappa}{r^2} \frac{\pu(r\phi)}{r}\right)\lesssim \frac{1}{r^3|u|^s}\left(\frac{1}{|u|^p}+\frac{1}{|u|^{p-1}r^2}\right)\lesssim\frac{1}{r^3|u|^{p+s}}+\frac{1}{r^3|u|^{p+1-s}}, 
 \end{equation*}
where we again used \eqref{eq:temp5}.
 We conclude that the $\pu\pv r\frac{\T (r\phi)}{r}$-term again dominates and that we can repeat the bootstrap argument as before.
 
 This finishes the proof.
\end{proof}

\subsubsection{Convergence of  \texorpdfstring{$\lim_{u\to-\infty}|u|^{p-1}r\phi$}{u r phi}}\label{sec:refine:limit}
We now have all the tools at hand to reprove Theorem~\ref{thm:timelike:final} without the smallness assumption on $C_{\mathrm{in},\phi}$ and without the restriction on $p\geq2$, see Remark~\ref{rem:shortcoming of limiting proof}.
However, as explained in  Remark~\ref{rem:shortcoming of limiting proof 2}, we still would not be able to conclude that $|u|^{p-1}r\phi(u,v)$ tends to a limit as $u$ tends to $-\infty$.
We now prove a lemma that allows us to do precisely this:
\begin{lemma}\label{lemma:trick}
Under the same assumptions as in Theorem~\ref{thm:refinements:TPhi}, assuming moreover that, on $\Gamma$,
\begin{equation}\label{eq:lemmatrick:ass}
  \left|\hat{r}\hat{\phi}-\frac{|u|}{p-1}\T(\hat{r}\hat{\phi})\right|\leq F |u|^{-p+1-\epsilon_\phi}
\end{equation}
for some $1>\epsilon_\phi>0$ and a constant $F>0$, we have that, for $s\neq 1$, for large enough values of $|U_0|$ and for $\hat{r}\geq R\geq 4M$:
\begin{equation}\label{eq:refine:trickforlimit}
  \left  |r\phi-\frac{|u|}{p-1}\T(r\phi)\right|\leq F'|u|^{-p+1-\epsilon'}
\end{equation}
for a constant $F'$ depending only on initial data and not on the value of $u_0$.
Here, $\epsilon'$ is given by
\begin{equation}\label{eq:epsilonprime}
     \epsilon'=\min(\epsilon_\phi,2p-2,s,|s-1|)=
\begin{cases}
\min(\epsilon_\phi,2p-2,s-1),& s>1, \\
\min(\epsilon_\phi, 2p-2, s,1-s), & s<1.\\
\end{cases}
\end{equation}

If $s=1$, then we instead have
\begin{equation}\label{eq:refine:trickforlimit2}
   \left |r\phi-\frac{|u|}{p-1}\T(r\phi)\right|\leq F'|u|^{-p+1}\frac{1}{\log^2|u|}.
\end{equation}

\end{lemma}
\begin{proof}
The proof will follow the same ideas as the previous proofs, and we will only sketch it. First, consider the case $s>1$.
We compute
\begin{align}
    \begin{split}
        \pu\pv\left(r\phi-\frac{|u|}{p-1}\T(r\phi)\right)=&\frac{\pu\pv r}{r}\left(r\phi-\frac{|u|}{p-1}\T(r\phi)\right)+\frac{1}{p-1}\pv\T(r\phi)\\
        &-\frac{|u|}{p-1}\left(\pu\pv\T(r\phi)-\pu\pv r\frac{\T (r\phi)}{r}\right)    .
    \end{split}
\end{align}
We will again assume \eqref{eq:refine:trickforlimit} as a bootstrap assumption and improve it. 
The first term in the equation above is the usual one, and we can deal with it. It is left to show that the others decay faster:

For the second term, plugging in the bound \eqref{thm.refinementes.boundonpvtphi} and converting some of the $u$-decay in it into $r$-decay and integrating in $u$ and $v$ does the job (this is where we need $\epsilon_\phi<1$).

For the final term, we proceed exactly as in the proof of Theorem~\ref{thm:refinements:TPhi}.

We proceed similarly in the cases $s=1$ and $s\leq1$.
\end{proof}
\begin{rem}\label{rem:lemma:trick}
Notice that one can replace the RHS of the assumption \eqref{eq:lemmatrick:ass} with any function that is non-increasing in $|u|$ and recover the correspondingly adapted \eqref{eq:refine:trickforlimit}. In particular, we can add a constant to the RHS of \eqref{eq:lemmatrick:ass}. This will play a role later on because of the cut-off functions introduced below, see Remark~\ref{rem:cutofferror}
\end{rem}

\subsection{Proof of Thm.~\ref{thm.intro:timelikecase}}\label{sec:refine:limitingargumentover}
\subsubsection{Sending \texorpdfstring{$\mathcal{C}_{u_0}$}{Cu0} to \texorpdfstring{$\mathcal{I}^-$}{I-} (revisited)}\label{sec:refine:limitingargument}
We can now prove the refined version of Thm.~\ref{thm:timelike:final}. The setup will be the same as in section~\ref{sec:timelike:limit}, with some minor modifications that we here point out:
%\subsubsection*{The background manifold}
%The background manifold is constructed as in section~\ref{sec:timelike:limit}, with the conditions on $r_0|_{\Gamma}$ now being that either  
%    \begin{align}
%    \boldsymbol{T'} r_0|_{\Gamma} \sim \frac{1}{|t|^{s}},\,\,\, 1\geq s> 0,\,  \,\,\text{and}\,\,\,      \lim_{t\to-\infty}r_0|_{\Gamma}=\infty,
%    \end{align}
%    or\footnote{Recall that the lower bound on $R$ is not optimised and can be pushed lower.} 
%    \begin{align}
%        |\boldsymbol{T'} r_0|_{\Gamma}|\lesssim  \frac{1}{|t|^{s}},\,\,\, s>1,\,  \,\,\text{and}\,\,\,      \lim_{t\to-\infty} r_0|_{\Gamma}=R\geq4M.
%    \end{align}
\subsubsection*{The "final" boundary data}
As in section~\ref{sec:timelike:limit}, we let $M>0$, and we restrict to sufficiently large negative values of $u\leq U_0<0$ and specify boundary data $(\hat r, \hat \phi)$ on $\Gamma$ as follows: The datum $\hat r\in C^2(\Gamma)$ is to satisfy $\hat r>2M$ and either 
    \begin{align}
    \T \hat{r} \sim \frac{1}{|u|^{s}},\,\,\, 1\geq s> 0,\,  \,\,\text{and}\,\,\,      \lim_{u\to-\infty}\hat{r}=\infty,
    \end{align}
    or\footnote{We now use the better lower bound on $R$, cf.\ Remarks~\ref{rem:betterbound} and~\ref{rem:betterbound2}.}
    \begin{align}
        |\T \hat{r}|\lesssim  \frac{1}{|u|^{s}},\,\,\, s>1,\,  \,\,\text{and}\,\,\,      \lim_{u\to-\infty}\hat{r}=\textcolor{black}{R>2M.}
    \end{align}
 On the other hand, $\hat\phi\in C^2(\Gamma)$ is chosen to obey $\lim_{u\to-\infty}\hat r\hat\phi(u)=0$ and
\begin{align}\label{eq:refine:assumptionP}
             \T(\hat{r}\hat{\phi}) &= C_{\mathrm{in},\phi}^1|u|^{-p}+\mathcal{O}(|u|^{-p-\epsilon_\phi})\end{align}
%as well as\begin{align}
%%             \hat{r}\hat{\phi}&=\frac{C_{\mathrm{in},\phi}^1}{p-1}|u|^{-p+1}+\mathcal{O}(|u|^{-p+1-\epsilon_\phi})
%    \end{align}
for $p>1$, $1>\epsilon_\phi>0$ and some constant $C_{\mathrm{in},\phi}^1>0$.

\subsubsection*{The sequence of finite solutions $(r_k,\phi_k,m_k)$}
We finally prescribe a sequence of initial/boundary data as in section~\ref{sec:timelike:limit}:
We recall the notation that, for $k\in\mathbb N$, $\mathcal{C}_k=\{u=-k, v\geq u\}$, 
and we also recall the sequence of smooth cut-off functions $(\chi_k)_{k\in \mathbb N}$ on $\Gamma$ from \eqref{eq:chi_k}, which equal 1 for $u\geq -k+1$, and which equal 0 for $u\leq -k$.
 
%    \begin{equation}
%        \chi_k=\begin{cases}
%        1,& u\geq- k+1,\\
%        0,&u\leq -k.
%        \end{cases}
%    \end{equation}
Our sequence of initial data shall then be given\footnote{Notice that, for technical reasons, we here cut off $\T(\hat r \hat\phi)$ rather than $\hat r \hat\phi$ itself.} by:
    \begin{equation}
    (I.D.)_k=\begin{cases}
    \hat{r}_k=\hat{r}, \,\hat r \hat{\phi}_k(u)=\int_{-\infty}^u \mathbf\chi_k T(\hat r\hat{\phi})\dd u' & \text{ on }\, \,\,\Gamma ,\\
    \bar{r}_k=r_0,\,\bar{\phi}_k=0, \,m=M& \text{ on }\, \,\,\mathcal{C}_k.
    \end{cases}
    \end{equation}
These lead to a sequence of solutions $(r_k, \phi_k,m_k)$, which we extend with the vacuum solution $(r_0,0,M)$ for  $u\leq-k$ (cf.~Thm.~\ref{thm:globalinv}) and which obey, uniformly in $k$, the bounds from Theorem~\ref{thm:timelike:cc} and also the refined bounds from Theorems~\ref{thm:refinements:TR} and~\ref{thm:refinements:TPhi} and Lemma~\ref{lemma:trick}. 
\begin{rem}\label{rem:cutofferror}There is one small technical subtlety here: The difference $\hat r_k\hat \phi_k-\frac{|u|}{p-1} \mathbf T(\hat r_k\hat\phi_k)$ has an error term coming from the cut-off function $\chi_k$:
\[|\hat r_k\hat \phi_k-\frac{|u|}{p-1} \mathbf T(\hat r_k\hat\phi_k)|\leq \frac{F}{|u|^{p-1+\epsilon_\phi}}+\frac{C}{k^{p-1}}\]
for some positive constants $F$ and $C$. As explained in Remark~\ref{rem:lemma:trick}, Lemma~\ref{lemma:trick} still applies to this. The error contribution $Ck^{1-p}$ arising from the cut-off function then vanishes as $k\to\infty$.
\end{rem}

\begin{thm}\label{thm:timelike:final!!!}
 Let $p>1$ and $U_0<0$ be sufficiently large. Then, as $k\to\infty$, the sequence $(r_k, \phi_k,m_k)$ uniformly converges to a limit $(r,\phi,m)$, 
            \begin{equation}
                ||r_k\phi_k-r\phi||_{C^1(\DUU)}+||r_k-r||_{C^1(\DUU)}+||m_k-m||_{C^1(\DUU)}\to 0.
            \end{equation}
       This limit is also a solution to the spherically symmetric Einstein-Scalar field equations.   Moreover, $(r,\phi,m)$ restricts correctly to the boundary data $(\hat{r},\hat{\phi})$ and satisfies, for all $v$,             
          %   The sequence $(r_k, \phi_k,m_k)$ as defined above converges, as $k\to\infty$, to a limit $(r,\phi,m)$,
%            \begin{equation}
%                ||r_k\phi_k-r\phi||_{C^1(\DU)}+||r_k-r||_{C^1(\DU)}+||m_k-m||_{C^1(\DU)}\to 0,
%            \end{equation}
%        which is also a solution to the spherically symmetric Einstein-Scalar field equations. Moreover, $(r,\phi,m)$ restricts correctly to the boundary data $(\hat{r},\hat{\phi})$ and satisfies for all $v$
        \begin{equation}
            \lim_{u\to-\infty}r\phi(u,v)=  \lim_{u\to-\infty}\pv m(u,v)=  \lim_{u\to-\infty}\pv (r\phi)(u,v)=0
        \end{equation}
        as well as
        \begin{equation}
              \lim_{u\to-\infty}m(u,v)=\lim_{u\to-\infty}m(u,u)=M.
        \end{equation}
       
        Assume now that also $R>2.95 M$. Then the following \underline{sharp} bounds hold throughout  $\DU$, for sufficiently large negative values of $U_0$:
        \begingroup
\allowdisplaybreaks
        \begin{align}
           m-M   &=\mathcal{O}\left(\frac{1}{|u|^{2p-2}r}+\frac{1}{|u|^{2p-1}}\right)\label{master:m},\\
            \T m    &=\mathcal{O}\left(\frac{1}{|u|^{2p-1}r}+\frac{1}{|u|^{2p}}\right)\label{master:Tm},\\
      \kappa-1    &=\mathcal{O}\left(\frac{1}{|u|^{2p-2}r^2}\right)\label{master:kappa},\\
        \lambda-\left(1-\frac{2M}{r}\right)&=\mathcal{O}\left(\frac{1}{|u|^{2p-2}r}\right)\label{master:lambda},\\
        \nu+\lambda=\T r&=\mathcal{O}(|u|^{-\min(s,2p-1)})\label{master:Tr},\\
        |r\phi|     &\leq C_{\mathrm{in},\phi}|u|^{-p+1}\label{master:rphi},\\
        |\pv(r\phi)|&\leq M C_{\mathrm{in},\phi}\frac{|u|^{-p+1}}{r^2}\label{master:pvrphi},\\
        |\pu(r\phi)+\pv(r\phi)|=|\T(r\phi)|&\leq C_{\mathrm{in},\phi}^1|u|^{-p}\label{master:trphi},\\
        |\pv\T(r\phi)|&\leq\eta M C_{\mathrm{in},\phi}^1\frac{|u|^{-p}}{r^2}\label{master:pvtrphi},
        \end{align}\endgroup
      where $C_{\mathrm{in},\phi}=C_{\mathrm{in},\phi}^1/(p-1)$, $\eta>1$ can be chosen arbitrarily close to 1 as $U_0\to-\infty$ and all the constants implicit in $\mathcal{O}$ only depend on initial data. 
      
      Finally, the following limit exists and is non-zero:
      \begin{equation}
             \lim_{u\to -\infty}|u|^{p-1}r\phi(u,v)=:\Phi^-\neq 0\label{master:Phi-}.
      \end{equation}
   More precisely, we have, for $s\neq1$, that
   \begin{equation}
       r\phi(u,v)=\frac{\Phi^-}{|u|^{p-1}}+\mathcal{O}\left( \frac{1}{|u|^{p-1+\min(\epsilon, 2p-2, s, |s-1|)}}+\frac{1}{r|u|^{p-1}}\right)\label{master:Phi- sneq1},
   \end{equation}
   and, for $s=1$,
   \begin{equation}\label{master:Phi-s=1}
       r\phi(u,v)=\frac{\Phi^-}{|u|^{p-1}}+\mathcal{O}\left( \frac{1}{|u|^{p-1}\log^2|u|}+\frac{1}{r|u|^{p-1}}\right).
   \end{equation}
   If $s\leq 1$, then the above limit \eqref{master:Phi-} is given by $\Phi^-=C_{\mathrm{in},\phi}/(p-1)$.
\end{thm}
\begin{rem}
The lower bound $R>2.95M$ is only necessary for the bounds \eqref{master:rphi}--\eqref{master:pvtrphi}, which otherwise still hold with slightly worse constants, and for the statement that $\Phi^-\neq 0$. In other words, it is only necessary for the proof of \textit{lower} bounds, not of upper bounds. We expect that it can be improved.
\end{rem}
\begin{proof}
First, notice that the reason that one can take the lower bound on $R$ to be just $R>2M$ is explained in  Remarks~\ref{rem:betterbound} and~\ref{rem:betterbound2}.
The first part of the theorem, namely that $(r_k,\phi_k,m_k)$ converges to a solution $(r,\phi,m)$ which restricts correctly to the boundary data, is then shown as in the proof of Thm.~\ref{thm:timelike:final}, now using the improved decay on $\T(r\phi)_k$ from Thm.~\ref{thm:refinements:TPhi}. As discussed in Remark~\ref{rem:shortcoming of limiting proof}, the fact that we  now have sharp decay for $\pu(r\phi)$ at our disposal removes the necessity to assume $p\geq2$ as well as the smallness assumption on $C_{\mathrm{in},\phi}$. Moreover, one can perform a similar argument to show convergence in higher derivative norms as well.

We thus obtain a limiting solution which, as before, satisfies all the bounds from Theorems~\ref{thm:timelike:cc},~\ref{thm:refinements:TR} and~\ref{thm:refinements:TPhi}. 

Furthermore, the improvements in the bounds \eqref{master:m}, \eqref{master:kappa} and \eqref{master:lambda} can be obtained from redoing the proof of Theorem~\ref{thm:timelike:cc} with the improved bound on $\T(r\phi)$ from Thm.~\ref{thm:refinements:TPhi}. 

The bounds  \eqref{master:Tm}, \eqref{master:Tr}, \eqref{master:rphi} and \eqref{master:pvrphi} come directly from Theorems~\ref{thm:timelike:cc} and~\ref{thm:refinements:TR}, where, for the latter two bounds, we used that $r\phi$ also satisfies a lower bound provided that $R>2.95$ (as shown in Thm.~\ref{thm:timelike:final}, estimate \eqref{lowerbound}), which allows us to improve $C'$ to $C_{\mathrm{in},\phi}$.

With similar reasoning as for $r\phi$ and $\pv(r\phi)$, one derives the improvements in the estimates \eqref{master:trphi} and \eqref{master:pvtrphi} from Thm.~\ref{thm:refinements:TPhi} by showing that $\T(r\phi)$ also satisfies a lower bound and, in particular, has a sign. (This is done in the same way as for $r\phi$ in the proof of Thm.~\ref{thm:timelike:final}.)

To prove the final part of the theorem, we note that, if $s\leq1$, it is trivial to show that $|u|^{p-1}r\phi$ attains a limit by looking at
$$|u|^{p-1}r\phi(u,v)=|u|^{p-1}r\phi(u,u)+|u|^{p-1}\int_{u}^v\pv(r\phi)\dd v'.$$

On the other hand, if $s\neq 1$, then we can show that $\pu(|u|^{p-1}r\phi)$ is integrable using the results of Lemma~\ref{lemma:trick}:
\begin{align*}
    -\pu(|u|^{p-1}r\phi)&=(p-1)|u|^{p-2}\left(r\phi-\frac{|u|}{p-1}\pu(r\phi)\right)\\
                    &=(p-1)|u|^{p-2}\left(r\phi-\frac{|u|}{p-1}\T(r\phi)\right)+|u|^{p-1}\pv(r\phi).
\end{align*}
The second term in the second line is bounded by $r^{-2}$ and, thus, is integrable. The first term in the second line, on the other hand, has been dealt with in Lemma~\ref{lemma:trick}, see \eqref{eq:refine:trickforlimit} and Remark~\ref{rem:cutofferror}, and is also integrable.

This concludes the proof.
\end{proof}

\subsubsection{Asymptotics of \texorpdfstring{$\pv(r\phi)$}{d/dv(r phi)} near \texorpdfstring{$\mathcal{I}^+$}{I+}, \texorpdfstring{$i^0$}{i0} and \texorpdfstring{$\mathcal{I}^-$}{I-}}\label{sec:timelike:asymptotics}
We now state the asymptotics for the limiting solution $(r,\phi,m)$ in a neighbourhood of spatial infinity.
By the above theorem, we have completely reduced the problem to the null case.
We can therefore reproduce the proofs of section~\ref{sec:nul:asymptotics} to conclude the following:
\begin{thm}\label{thm:timelike:logs}
Consider the solution $(r,\phi,m)$ constructed in Theorem~\ref{thm:timelike:final!!!}, and let $p=2$ in equation \eqref{eq:refine:assumptionP}.
Then, throughout $\DU\cap\{v>1\}$, for sufficiently large negative values of $U_0$, we get the following asymptotic behaviour for $\pv(r\phi)$:
\begin{equation}
     |\pv(r\phi)|\sim  
\begin{cases}
\frac{\log r-\log|u|}{r^3}, & u=\con,\,\, v \to \infty, \\
\frac{1}{r^3}, & v=\con,\,\, u \to -\infty,\\
\frac{1}{r^3}, & v+u=\con,\,\, v\to \infty.
\end{cases}
\end{equation}
More precisely, for fixed $u$, we have the following asymptotic expansion  as $\mathcal{I}^+$ is approached:
\begin{equation}
\left|\pv(r\phi)(u,v)+2M \Phi^- r^{-3} \left(\log r-\log(|u|)-\frac32\right)\right|  =\mathcal{O}(r^{-3}\log^{-2}(|u|)+r^{-4}|u|).
\end{equation}
The $\log^{-2}|u|$-term above can be replaced by $|u|^{-\epsilon'}$ for $\epsilon'$ as in \eqref{eq:epsilonprime} if $s\neq1$.
\end{thm}

Similarly, we can deal with higher-order asymptotics, that is with the cases $p=3,4,\dots $.
See also Thm.~\ref{thm:null:asymptotics of dvrphi,p=3}.

Finally, in view of Remark~\ref{rem:uniqueness}, Theorems \ref{thm:timelike:final!!!} and \ref{thm:timelike:logs} combined prove Theorem~\ref{thm.intro:timelikecase} from the introduction.
\newpage
\section{An application: The scattering problem}\label{sec:scattering}
In the previous sections~\ref{sec:null} and~\ref{sec:timelike}, our motivation for the choice of initial (/boundary) data mainly came from Christodoulou's argument; in particular, the data were chosen so as to lead to solutions that satisfy the no incoming radiation condition and that agree with the prediction of the  quadrupole approximation, that is, we chose initial data such that we would obtain the rate
\begin{equation}
    \pu m(u,\infty)\sim -\frac{1}{|u|^4}
\end{equation}
at future null infinity.

Alternatively, we could have motivated our choice of initial data by the observation that our data can be chosen to be conformally smooth near $\mathcal{I}^-$ for integer $p$ and, nevertheless, lead to solutions that are not conformally smooth near $\mathcal{I}^+$.

In this section, we give yet another extremely natural motivation for our initial data of section~\ref{sec:null}. 
More precisely, we shall show in section~\ref{sec61} that the case $p=3$ appears \textit{generically} in evolutions of compactly supported scattering data on\footnote{Since we always restrict to a region sufficiently close $\mathcal{I}^-$, that is to sufficiently large negative values of $u$, we may without loss of generality assume vanishing data on $\mathcal{H}^-$.} $\mathcal{H}^-$ and $\mathcal{I}^-$. Our main theorem is Thm.~\ref{null:thm:scattering}, which contains Theorem~\ref{thm.intro:scattering} from the introduction.

We shall make further comments on \emph{linear} scattering in sections~\ref{sec62} and~\ref{sec63}, where we will, in particular, prove that the corresponding solutions are \textit{never conformally smooth}  (unless they vanish identically).

\subsection{Non-linear scattering with a Schwarzschildean or  Minkowskian \texorpdfstring{$i^-$}{i-} (Proof of Thm.~\ref{thm.intro:scattering})}\label{sec61}
\paragraph{The Maxwell field}
 As in the timelike case (section~\ref{sec:timelike}), we will ignore the Maxwell field, that is, we set $e^2=0$. However, all results of the present section can be recovered for $e^2\neq 0$ as well.
 \paragraph{The setup}
Let $M>0$, $U<-2M$, and define the rectangle
 \begin{equation}
     \mathcal{E}_{U}:=(-\infty,U]\times (-\infty,\infty)\subset{\mathbb{R}^2}.
 \end{equation}
We refer to the set $(-\infty,U]\times\{-\infty\}$ as $\mathcal{H}^-$ (to be thought of as the past event horizon of Schwarzschild), to the point $\{-\infty\}\times \{-\infty\}$ as $i^-$ or past timelike infinity, and we otherwise keep the conventions from section~\ref{sec:null:assumptions}.
 
% On this rectangle $\mathcal{E}_{U}$, we assume that a strictly positive $C^3$-function $r(u,v)$, a non-negative $C^2$-function $ m(u,v)$ and a $C^2$-function $\phi(u,v)$ are defined and obey the following properties:
%
%The function  $r$ is such that, along each of the ingoing and outgoing null rays, it tends to infinity, i.e., $\lim_{v\to\infty}r(u,v)=\infty$, and $\lim_{u\to-\infty}r(u,v)=\infty$.
% 
%  For the radiation field, we assume that $r\phi$ vanishes along $\mathcal{H}^-$ and that
% \begin{equation}
% \lim_{u\to -\infty}r\phi(u,v)=G(v)
% \end{equation}
% is a smooth compactly supported function, i.e., we assume that there exist $v_2>v_1$ such that $\mathrm{supp}(G)\subset[v_1,v_2]$.
% 
% Concerning $m$, we assume that $m(v,-\infty)=M\geq0$ for all $v\leq v_1$.
% 
% Finally, we assume that, throughout $\mathcal{E}_{U}$,
% \begin{align}
%    \pu r=	\nu\leq0\label{eq:nu-scatter},\\
%    \pv r=	\lambda>0\label{eq:lambda+scatter},
%\end{align}
% whereas we assume that $\nu<0$ throughout $\mathcal{E}_{U}\setminus \mathcal{H}^-$.
% Lastly, assume that $\kappa>0$, 
% $$\lim_{u\to-\infty}\kappa(u,v)=1,$$
%  that
% $$\nu(u,v_2)=-1,$$
% and that $r(U,v_2)=-U$.
%
%To relate to the Einstein-Scalar Field system, we assume that, throughout $\mathcal{E}_{U}$, the equations \eqref{eq:puvarpi}--\eqref{eq:pvzeta} hold \textit{pointwise}.
%
%
%The reader familiar with Penrose diagrams may refer to the Penrose diagram below (Figure~\ref{fig:9}) where the geometric content of this information is summarised.

%%%%%%%%%%%%%%
\begin{figure}[htbp]
\floatbox[{\capbeside\thisfloatsetup{capbesideposition={right,top},capbesidewidth=4cm}}]{figure}[\FBwidth]
{\caption{The Penrose diagram of $\mathcal E_{U}$. We pose compactly supported scattering data on $\mathcal I^-$ and $\mathcal H^-$. Since we are only interested in a region close to $\mathcal I^-$, we can, without loss of generality, set the data on $\mathcal H^-$ to be vanishing. }\label{fig:9}}
{ \includegraphics[width = 150pt]{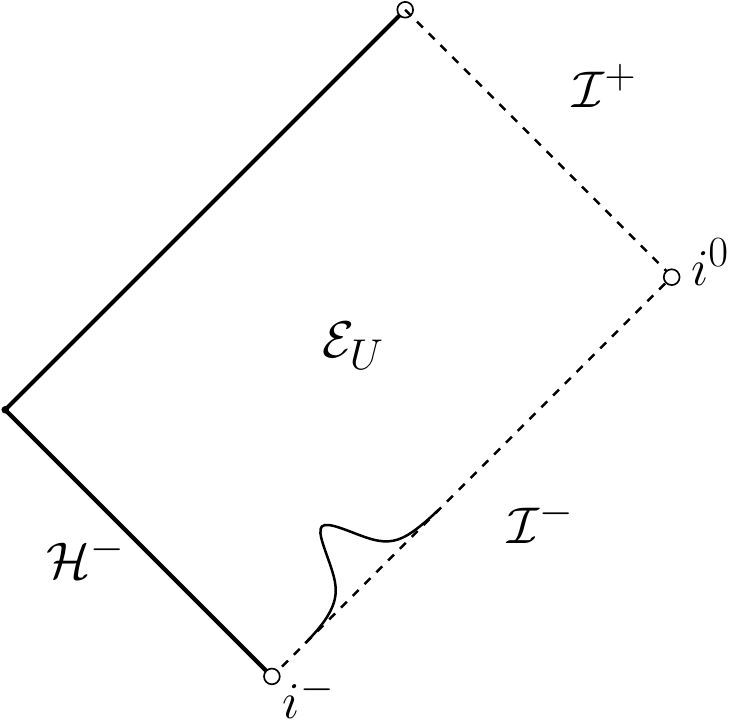}
}
\end{figure}
%%%%%%%%%%%%%%
We will now show that if we pose compactly supported scattering data for $r\phi$ on $\mathcal I^-$ (see Figure~\ref{fig:9}) and vanishing data on $\mathcal H^-$, then $\mathcal E_U$ generically  contains as a subset a set $\mathcal{D}_{U_0}$ as defined in \eqref{eq:DutoDU0}, in which the corresponding scattering solution satisfies the assumptions of section~\ref{sec:null:assumptions} with $p=3$, and hence, according to Theorem~\ref{thm:null:asymptotics of dvrphi,p=3},  has logarithmic terms at second highest order in the expansion of $\pv(r\phi)$ near $\mathcal{I}^+$.

\begin{thm}\label{null:thm:scattering}
Let $G(v)$ be a smooth compactly supported function, $\mathrm{supp} (G)\subset(v_1,v_2)$. 
Then there exists a solution $(r,\phi,m)$ to the spherically symmetric Einstein-Scalar field system  on $\mathcal E_U$ which satisfies $r|_{\mathcal H^-}=2m|_{\mathcal H^-}$ on $\mathcal H^-$, and which satisfies $m(u,v)=M$ and  $\phi(u,v)=0$ for all $v\leq v_1$, and which finally satisfies $\lim_{u\to-\infty}r(u,v)=\infty$, $\pv r|_{\mathcal I^-}(v)=1$, and $r\phi|_{\mathcal I^-}(v)=G(v)$ for all $v\in \mathbb R$.
If we moreover fix $\pu r(u,v_2)=-1$ and $r(U,v_2)=-U$, then this solution is unique in the sense of Remark~\ref{rem:uniqueness}. We will call this solution \emph{the scattering solution}.

Furthermore, for sufficiently large negative values of $U_0$, this scattering solution $(r,\phi,m)$ satisfies the following bounds throughout $\mathcal{E}_{U}\cap \{v\geq v_2\}\cap\{u\leq U_0\}$:
\begin{equation}\label{6.3}
    \left|r\phi(u,v)+\frac{I_0[G]}{u^2}\right|=\mathcal{O}(|u|^{-3}  )  ,
\end{equation}
where $I_0[G]$ is a constant given by
\begin{equation}\label{eq:defofI0}
    I_0[G]:=\int_{v_1}^{v_2}\left(M+\frac{1}{2}\int_{v_1}^v \left(\frac{\dd G}{\dd v}\right)^2(v')\dd v'\right)G(v)\dd v.
\end{equation}
In particular, by the results of section~\ref{sec:null} (see Thm.~\ref{thm:null:asymptotics of dvrphi,p=3}), we have, for fixed values of $u$, the following asymptotic expression near $\mathcal{I}^+$ for $\pv(r\phi)$:
\begin{equation}\label{eq:scatteringtheorem:log}
\left|\pv(r\phi)(u,v)-\frac{F(u)}{r^3}+6\widetilde M I_0[G]\frac{\log(r)-\log|u|}{r^4}\right|  =\mathcal{O}(r^{-4}),
\end{equation}
where $F(u)$ is given by 
\begin{equation}\label{6.6}
    F(u)=\int_{-\infty}^u \lim_{v\to\infty}(2m\nu r\phi)(u',v)\dd u'=\frac{-2\widetilde{M} I_0[G]}{u}+\mathcal{O}(u^{-2}),
\end{equation}
and where $\widetilde{M}$, the final value of the past Bondi mass, is given by
\begin{equation}
  \widetilde{M} =\lim_{v\to\infty} m(-\infty,v)=M+\int_{v_1}^{v_2} \frac{1}{2}\left(\frac{\dd G}{\dd v}\right)^2(v')\dd v'>M.
\end{equation}

Finally, it is clear from its definition~\eqref{eq:defofI0} that the constant $I_0[G]$ is generically non-zero (in an obvious sense).
\end{thm}
Combined with Remark~\ref{rem:minkowskiscattering} below and the specialisation to the linear case described in section~\ref{sec:linear}, this theorem proves Theorem~\ref{thm.intro:scattering} from the introduction.
\begin{proof}
We first restrict to $v< v_1$. 
There, by the domain of dependence property, the scattering solution exists and is identically Schwarzschild. 
The existence and uniqueness of the scattering solution for $v\geq v_1$ can then be obtained by combining the estimates of the present proof with the methods of section~\ref{sec:timelike:limit}. (It is convenient to treat the regions $v\in[v_1,v_2]$ and $v>v_2$ separately.)

Let us now assume that we have already established the existence of the scattering solution.
Then, we first note that the Hawking mass $m$ on $\mathcal{I}^-$ is given by 
\begin{equation}
    m(-\infty,v)=M+\int_{v_1}^v \frac{1}{2}\left(\frac{\dd G}{\dd v}\right)^2(v')\dd v',
\end{equation}
which can be seen by integrating eq.~\eqref{eq:pvvarpi} from $i^-$ (and by standard limiting considerations, see the arguments below).

In the rest of the proof, we restrict to the region $v\in[v_1,v_2]$.
We can then, using the monotonicity of the Hawking mass, redo the proofs of Propositions~\ref{prop:null mass boundedness} and~\ref{prop:null geometric quantities boundedness} to show that $m$, $\nu$, $\lambda$ and $\kappa$ remain bounded from above, and away from zero, for $v\in[v_1,v_2]$.
Moreover, we can apply the energy estimate as in the proof of Thm.~\ref{thm:null boundedness of phi} to show that $\sqrt{r}|\phi|$ is bounded from above as well, cf.~\eqref{usageofenergyestimates}.

In order to improve this bound on $\phi$, we integrate the wave equation \eqref{eq:wave} from the ingoing null ray $v'=v_1$ (where $\phi$ vanishes), for $v\in[v_1,v_2]$:
\begin{equation}\label{eq6.9}
    |\pu(r\phi)(u,v)|\leq C\int_{v_1}^v r^{-\frac52}\dd v'\leq  \frac{C}{|u|^{\frac52}}
\end{equation}
for some positive constant $C$ that depends only on initial data (in particular, $C$ depends on $v_2-v_1$) but which is allowed to change from line to line.

In turn, integrating estimate \eqref{eq6.9} from $\mathcal{I}^-$ implies that
\[|r\phi(u,v)-G(v)|\leq \frac{C}{|u|^{\frac{3}{2}}}.\]
Plugging this improved bound back into the wave equation and repeating the argument \eqref{eq6.9}, we find that
\[ |\pu(r\phi)(u,v)|\leq \frac{C}{|u|^3} \]
and, thus, by again integrating from $\mathcal I^-$,
\[|r\phi(u,v)-G(v)|\leq \frac{C}{u^{2}}.\]

With these decay rates for $r\phi$ and $\pu(r\phi)$, we can prove the analogue of Corollary~\ref{cor:nullcase:asymptotics}; in particular, we can show that, for $v\in[v_1,v_2]$,
\[|\kappa(u,v)-1|+|\nu(u,v)+1|+|m(u,v)-m(-\infty,v)|=\mathcal{O}(u^{-2}).\]
To now obtain the asymptotic behaviour of $\pu(r\phi)$ along $v=v_2$,  we calculate the $v$-derivative of $r^3\pu(r\phi)$: Using the above bounds, we find that
\begin{multline}\pv(r^3\pu(r\phi))\\=3r^2\lambda \pu(r\phi)+2m\nu\kappa r\phi=-2\left(M+\int_{v_1}^v \frac{1}{2}\left(\frac{\dd G}{\dd v}\right)^2(v')\dd v'\right)G(v)+\mathcal{O}(|u|^{-1}).
\end{multline}
Integrating the estimate above from $v_1$ to $v_2$ yields
\begin{align}\label{eq6.11}
    \left|r^3\pu(r\phi)(u,v_2)+\int_{v_1}^{v_2}2\left(M+\int_{v_1}^v \frac{1}{2}\left(\frac{\dd G}{\dd v}\right)^2(v')\dd v'\right)G(v)\dd v\right|=\mathcal{O}(|u|^{-1}).
\end{align}
Since $\nu=-1$ on $v=v_2$, this puts us in precisely the setting of section~\ref{sec:null:assumptions} with $p=3$: 
Indeed, integrating the equation \eqref{eq6.11} from $\mathcal I^-$ along $v=v_2$, we find
%and using Cor.~\ref{cor:nullcase:asymptotics}, we find for all $v>v_2$ and for all $u\leq U_0$, and for sufficiently large negative values of $U_0$, that
\begin{equation}\label{zztr}
    \left|r\phi(u,v_2)+\frac{1}{u^2} \int_{v_1}^{v_2}\left(M+\int_{v_1}^v \frac{1}{2}\left(\frac{\dd G}{\dd v}\right)^2(v')\dd v'\right)G(v)\dd v\right|=\mathcal{O}(|u|^{-3}).
\end{equation}
In fact, by Corollary~\ref{cor:nullcase:asymptotics}, the same holds for any $v\geq v_2$.
In particular, Thm.~\ref{thm:null:asymptotics of dvrphi,p=3} applies with  $p=3$, with $\Phi^-$ given by
\begin{equation}
\Phi^-=-\int_{v_1}^{v_2}\left(M+\int_{v_1}^v \frac{1}{2}\left(\frac{\dd G}{\dd v}\right)^2(v')\dd v'\right)G(v)\dd v,
\end{equation}
and with $M$ in \eqref{eq:null:thm:B*constant} replaced by $\widetilde{M}=m(-\infty,v_2)$. This concludes the proof.
\end{proof}

\begin{rem}[Non-linear scattering for perturbations of Minkowski]\label{rem:minkowskiscattering}

In contrast to the setting in the previous section, we can now also have $M=0$ and still see the logarithmic term. 
This is because the scattering data on $\mathcal{I}^-$ will always generate mass such that there will ultimately be a mass term near $i^0$. 
In particular, the results of Theorem~\ref{null:thm:scattering} not only apply to scattering solutions with a Schwarzschildean $i^-$ and compactly supported scattering data, but also to scattering solutions with a  Minkowskian $i^-$ (see the Penrose diagram below). This is because if one puts vanishing data on the center $r=0$ and compactly supported data on $\mathcal{I}^-$, there will be a backwards null cone which is emanating from the center and on which $r\phi=0$ by the domain of dependence property.

Moreover, we recall from Remark~\ref{rem:BV} that if the initial data on $\mathcal I^-$ are sufficiently small, then, according to~\cite{ChristodoulouBV}, the arising solution is causally geodesically complete and globally regular, it has a complete null infinity, and its Penrose diagram can  be extended to a Minkowskian Penrose diagram as in Figure~\ref{fig:10}.
%%%%%%%%%%%%%%
\begin{figure}[htbp]
\floatbox[{\capbeside\thisfloatsetup{capbesideposition={right,top},capbesidewidth=4cm}}]{figure}[\FBwidth]
{\caption{Scattering solution arising from compactly supported scattering data on $\mathcal I^-$ and a Minkowskian $i^-$. The solution fails to be conformally smooth near $\mathcal I^+$ by Theorem~\ref{null:thm:scattering}. 
Moreover, if the scattering data are suitably small, then the Penrose diagram can be extended to a Minkowskian Penrose diagram by the results of~\cite{ChristodoulouBV}.}\label{fig:10}}
{\includegraphics[width = 135pt]{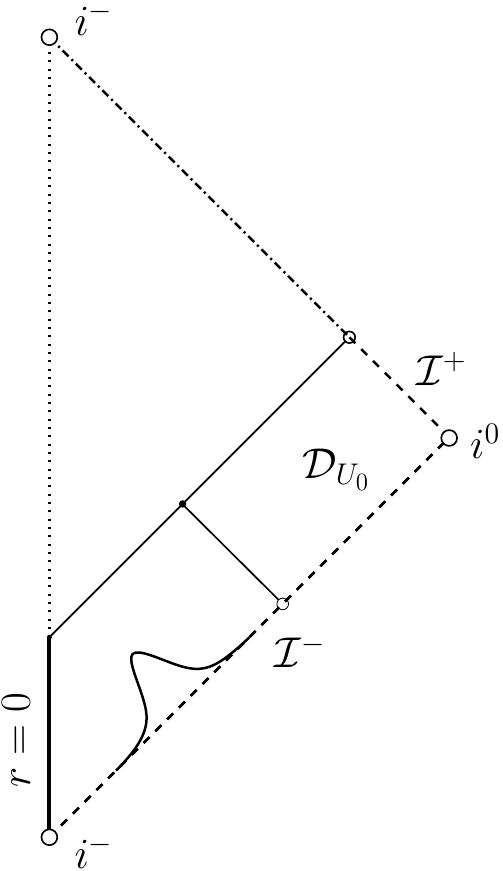}
}
\end{figure}
%%%%%%%%%%%%%
%\begin{figure}[htbp]
%  \centering
%  \includegraphics[width = 135pt]{PenroseDiagramScatteringMinkowskiwithDnotcutoff.pdf}
%\end{figure}
\end{rem}

\subsection{Linear scattering on Schwarzschild}\label{sec62}
By the remarks in section~\ref{sec:linear}, Theorem~\ref{null:thm:scattering} also applies in the case of the linear wave equation on a fixed Schwarzschild background  with mass $M>0$ (see Figure~\ref{fig:11} below).
 In the Eddington--Finkelstein double null coordinates\footnote{The gauge of the $u$-coordinate $\pu r=1-\frac{2M}{r}$ differs from our choice in the non-linear setting, where we set $\nu=-1$ on $v=v_2$. In these coordinates, the only difference to our results is then that the $\mathcal O(|u|^{-3})$-term in \eqref{6.3} is replaced by an $\mathcal O(|u|^{-3}\log |u|)$-term, and similarly for \eqref{6.6}. This is completely inconsequential to any of our other results, however.} of section~\ref{sec:Schwarzschild} (recall that $\pv r=-\pu r=1-\frac{2M}{r}$), the linear wave equation reads
\begin{equation}\label{eq:scatteringwaveequation}
    \pu\pv(r\phi)=-2M\(1-\frac{2M}{r}\)\frac{r\phi}{r^3}.
\end{equation}
The only difference in the linear case is that $I_0[G]$ is now given by
\begin{equation}
    I_0[G]:=\int_{v_1}^{v_2}M G(v)\dd v,
\end{equation}
since the scalar field no longer generates mass along past null infinity.
%%%%%%%%%%%
\begin{figure}[htbp]
\floatbox[{\capbeside\thisfloatsetup{capbesideposition={right,top},capbesidewidth=4cm}}]{figure}[\FBwidth]
{\caption{Smooth compactly supported scattering data for $\phi$ on a fixed Schwarzschild background. The solution generically contains a region $\mathcal D_{U_0}$ with $p=3$ as in section~\ref{sec:null}. Moreover, the scalar field is \textit{never} conformally smooth unless the scattering data vanish, see Theorem~\ref{scatteringthm2}.}\label{fig:11}}
{\includegraphics[width = 195pt]{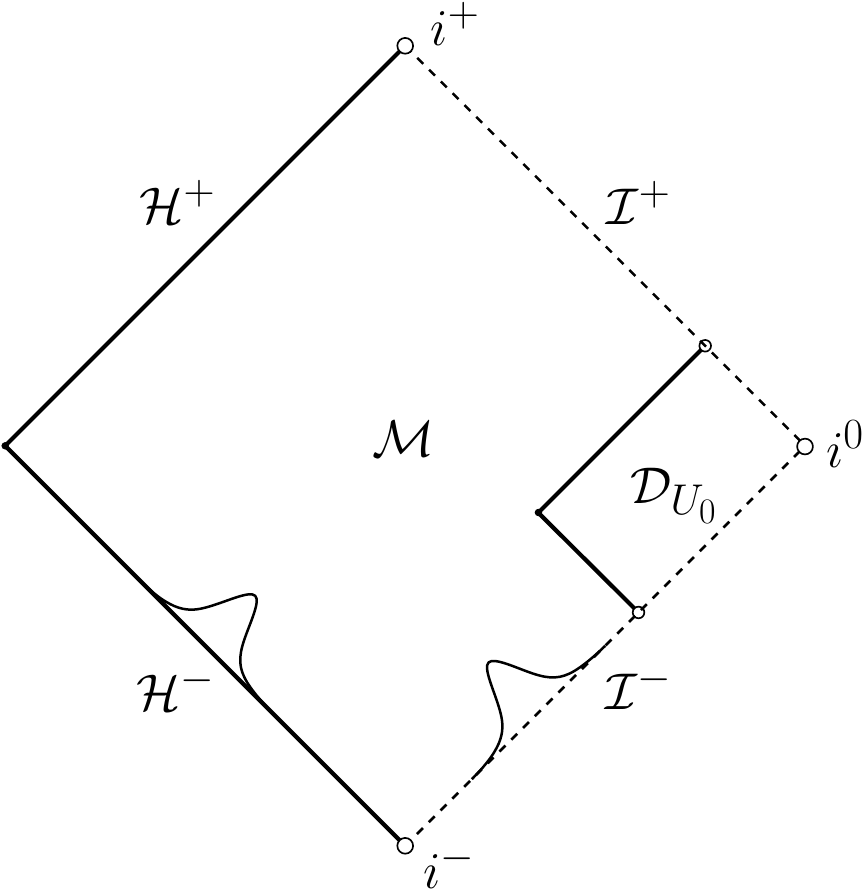}}
\end{figure}
%%%%%%%%%%%
%\begin{figure}[htbp]
%  \centering
%  \includegraphics[width = 195pt]{SchwarzschildDoubleSupportDU00.pdf}
%\end{figure}

We now want to  classify all spherically symmetric solutions to the linear wave equation arising from compactly supported scattering data in terms of their conformal smoothness near $\mathcal{I}^+$.
We have already established that if 
\begin{equation*}
    I_0[G]:=\int_{v_1}^{v_2}M G(v)\dd v\neq0,
\end{equation*}
then there will be a logarithmic term in the expansion of $\pv(r\phi)$ at order $\frac{\log r}{r^4}$.
Let us now discuss the case when $I_0[G]=0$. We prove the following theorem:
\begin{thm}\label{scatteringthm2}
Prescribe compactly supported scattering data $G(v)$ on $\mathcal I^-$ and compactly supported scattering data on $\mathcal H^-$ for the spherically symmetric linear wave equation \eqref{eq:scatteringwaveequation} on a fixed Schwarzschild background with mass $M>0$. Then, by the results of~\cite{DRSR}, there exists a unique smooth scattering solution $\phi$ attaining these data, with the uniqueness being understood in the class of finite-energy solutions.

For $n\in\mathbb{N}_0$, define the scattering data constants $I^{(n)}[G]$ via
\begin{equation}
    I^{(n)}[G]:=M\int_{v_1}^{v_2}(-1)^n\frac{v^n}{n!} G(v)\dd v.
\end{equation}
Let $n$ denote the smallest natural number such that $I^{(n)}[G]\neq 0$.

Then the solution $\phi$ satisfies, for all $v\geq v_2$ and for all $u<U_0$, and for sufficiently large negative values of $U_0$:
\begin{equation}
    \left|r\phi(u,v)+\frac{I^{(n)}[G]}{|u|^{2+n}}(n+1)!\right|=\mathcal{O}(|u|^{-3-n}\log|u|).\label{scatteringtheorem udecay}
\end{equation}
Moreover, for fixed values of $u$, we have the following asymptotic expansion  as $\mathcal{I}^+$ is approached:
\begin{equation}
    \pv(r\phi)=\sum_{i=0}^{n}\frac{f_i^{(n)}(u)}{r^{3+i}}-(-1)^n(3+n)! I^{(n)}[G] M \frac{\log r-\log|u|}{r^{4+n}}+\mathcal{O}(r^{-4-n}) \label{scatteringtheorem rdecay}
\end{equation}
for some smooth functions $f_i^{(n)}$.
\end{thm}
This theorem shows, in particular, that the solution only remains conformally smooth near $\mathcal{I}^+$ if $G=0$, that is to say, any  smooth compactly supported linear scalar perturbation on $\mathcal{I}^-$ gives rise to a solution which is not conformally smooth.

\begin{proof}
The existence of the scattering solution $\phi$ follows by our previous methods or by the results of~\cite{DRSR}.
The proof of the estimates \eqref{scatteringtheorem udecay} and \eqref{scatteringtheorem rdecay} will be a proof via induction, with the base case having been established in Thm.~\ref{null:thm:scattering}.
The crucial idea is to use \textit{time integrals}.\footnote{These have been used in a similar context in~\cite{Angelopoulos2018Late-timeSpacetimes}.}

To begin, let us state the following two basic facts:
First, we have, for $v\in[v_1,v_2]$, and for any $n\in\mathbb N$:
\begin{align}\label{eq:scatteringproof:fact1}
    \pu\pv(\pu^n (r\phi))=-\left(2M\(1-\frac{2M}{r}\)\frac{G(v)}{r^3}\right)\frac{(n+2)!}{2r^n}+\mathcal{O}(r^{-3-n-1}).
\end{align}
This can easily be established using the methods of the proof of Thm.~\ref{null:thm:scattering}.
Secondly and similarly, we have that, for all $n\in\mathbb N$ and for all $v\geq v_2$,
\begin{align}\label{eq:scatteringproof:fact2}
    \pu\pv(\pv^n (r\phi))=-\left(2M\(1-\frac{2M}{r}\)\frac{r\phi}{r^3}\right)(-1)^n\frac{(n+2)!}{2r^n}+\mathcal{O}(r^{-3-n-1}),
\end{align}
where this can easily be established from the asymptotics for $\pv(r\phi)$ proved in section~\ref{sec:nul:asymptotics} and an inductive argument. Note that both these facts also hold in the non-linear setting.

Let us now initiate the inductive step. We assume that \eqref{scatteringtheorem udecay} holds for some $n-1\geq 0$, and we moreover assume that it commutes with $\pu$, that is to say, we assume that
\begin{equation}
    \left|\pu^m\left(r\phi(u,v)+\frac{I^{(n-1)}[G]}{|u|^{2+n-1}}(n-1+1)!\right)\right|=\mathcal{O}(|u|^{-3-n+1-m}\log|u|)
\end{equation}
for some $n-1\geq 0$, for all $m\in \mathbb N$, and for all $r\phi$ arising from compactly supported scattering data $G$ such that $I^{(k)}[G]=0$ for all $k<n-1$. (That this holds in the base case $n=1$ is an easy consequence of eqns.~\eqref{zztr} and~\eqref{eq:scatteringproof:fact1}.)

Consider now compactly supported scattering data $G$ such that $I^{(k)}[G]=0$ for all $k<n$. These lead to a solution $r\phi$. 
The goal is to show that $r\phi$ can be written as $\T(r\phi^T)$, where $\T=\pu+\pv$, and where $r\phi^T$, the \textit{time integral} of $r\phi$, is another solution coming from compactly supported data $G^T$ such that $I^{(k)}[G^T]=0$ for all $k<n-1$.
To achieve this, we take the obvious candidate for $G^T$:
\begin{equation}
    G^T(v)=\int_{v_1}^vG(v')\dd v'.
\end{equation}
Indeed, by the methods of the proof of Thm.~\ref{null:thm:scattering}, it is easy to see that the solution $r\phi^T$ arising from this satisfies
\begin{equation}
    \T(r\phi^T)(-\infty,v)=\pv(r\phi^T)(-\infty, v)=(G^T)'(v)=G(v).
\end{equation}
Therefore, since $\T$ also commutes with the wave equation\footnote{This is the only property used in this proof that fails to hold in the coupled case.}, we indeed have $\T(r\phi^T)=r\phi$ by uniqueness.

It is left to show that $I^{(k)}[G^T]=0$ for all $k<n-1$.
But this is an easy consequence of the fact that 
\begin{equation}
    \int_{v_1}^{v_2}\frac{v'^k}{k} G(v')\dd v'=-\int_{v_1}^{v_2} v'^{k-1}\int_{v_1}^{v'}G(v'')\dd v''\dd v',
\end{equation}
where we used that $I^{(k)}[G]=0$ for all $k<n$.
In particular, the above equation implies
\begin{equation}\label{equick}
    I^{(n-1)}[G^T]=I^{(n)}[G],
\end{equation} and, similarly,  that $I^{(k)}[G^T]=0$ for all $k<n-1$.
From the induction assumption, it now follows that, for all $m\in\mathbb N$,
\begin{equation}
     \left|\pu^m\left((r\phi^T(u,v)+\frac{I^{(n-1)}[G^T]}{|u|^{2+n-1}}n!\right)\right|=\mathcal{O}(|u|^{-3-n+1-m}\log|u|).
\end{equation}

Finally, if we now write
$
    r\phi(u,v)=\T(r\phi^T)=\pu(r\phi^T)+\pv(r\phi^T)
$, then, as a consequence of the wave equation, the $\pv(r\phi^T)$-term goes like $r\phi^T/|u|^2$ and is therefore sub-leading (for $v\geq v_2$). We thus obtain, for all $v\geq v_2$, that
\begin{equation}
    \left|r\phi(u,v)+\frac{I^{(n)}[G]}{|u|^{2+n}}(n+1)!\right|=\mathcal{O}(|u|^{-3-n}\log|u|),
\end{equation}
where we also used \eqref{equick}.
One proceeds similarly for higher $\pu$-derivatives. (One can appeal to the wave equation \eqref{eq:scatteringwaveequation} to deal with the arising $\pu\pv$-terms.) This completes the inductive proof of eq.~\eqref{scatteringtheorem udecay}.

We proceed exactly in the same way for the proof of \eqref{scatteringtheorem rdecay}: 
We again make the inductive assumption that \eqref{scatteringtheorem rdecay} holds for some $n$ and moreover commutes with $\pv^m$ for all $m$. By this, we mean the following: 
We assume that, for some fixed $n$ and for all $m$,
\begin{equation}
   \pv^m ( \pv(r\phi))=\sum_{i=0}^{n}\frac{f_i^{(n,m)}(u)}{r^{3+i+m}}-(-1)^n(3+n)! I^{(n)}[G] M \pv^m\left(\frac{\log r}{r^{4+n}}\right)+\mathcal{O}(r^{-4-n-m}) 
\end{equation}
for all solutions $r\phi$ arising from compactly supported scattering data $G$ that have $I^{(k)}[G]=0$ for all $k<n$.
 Here, the $f_i^{(n,m)}$ are again some smooth functions. That this holds in the base case is a consequence of eq.\ \eqref{eq:scatteringtheorem:log} from Theorem~\ref{null:thm:scattering} combined with eq.~\eqref{eq:scatteringproof:fact2} and an inductive argument.

Then, in order to the inductive step, we assume that $r\phi$ is a solution arising from compactly supported scattering data $G$ with $I^{(k)}[G]=0$ for all $k<n$, and we anew write $r\phi$ as a time derivative, $r\phi=T(r\phi^T)$, and compute
\begin{equation}\label{equick2}
    \pv(r\phi)(u,v)=\pv^2(r\phi^T)+\pu\pv(r\phi^T)=\pv^2(r\phi^T)-2M\left(1-\frac{2M}{r}\right)\frac{r\phi^T}{r^3}.
\end{equation}
It is then a simple exercise to write down the asymptotics for the second term by plugging in the asymptotics for $\pv(r\phi^T)$ into
\begin{equation}
    r\phi^T(u,v)=r\phi^T(u,\infty)-\int_v^\infty \pv(r\phi^T)(u,v')\dd v'.
\end{equation}
Leaving the details to the reader, one hence finds that the second term in~\eqref{equick2} only produces $\log$-terms at later orders than $\pv^2(r\phi^T)$ does, so the leading-order logarithmic contributions to the asymptotics of $\pv (r\phi)$ are determined by $\pv^2(r\phi^T)$.
A similar argument works for higher derivatives. 

This concludes the proof.
\end{proof}
\subsection{Linear scattering  on extremal Reissner--Nordstr\"om}\label{sec63}

Finally, we remark that, by the "mirror symmetry" of the exterior of the extremal Reissner--Nordstr\"om spacetime~\cite{CouchTorrence} discussed in section~\ref{sec:222} of this paper (and the fact that all our results also apply when including a Maxwell field), we can state as an immediate corollary of our Thms.~\ref{null:thm:scattering} and~\ref{scatteringthm2}:
\begin{cor}
Consider the linear wave equation $\nabla^\mu\nabla_\mu \phi=0$ on extremal Reissner--Nordstr\"om ($|e|=M$). 
Put smooth compactly supported spherically symmetric scattering data on $\mathcal{I}^-$ and on $\mathcal{H}^-$. Then, by the results of~\cite{AAGscattering}, there exists a unique scattering solution attaining these data. This solution, in addition to not being conformally smooth near $\mathcal{I}^+$, fails to be \underline{smooth} at the future event horizon $\mathcal{H}^+$ unless it vanishes identically, and one generically has that $\phi$ is not $C^4$ in the variable $r-r_+$.
\end{cor}
This failure of the solution to remain smooth of course comes from the "mirrored" $\log(r-r_+)$-terms of Theorem~\ref{null:thm:scattering} that now appear in the ingoing derivative of $\phi$ instead of in $\pv(r\phi)$. Here, $r_+=M$ is the $r$-value at $\mathcal{H}^+$.

Notice that this is in stark contrast to the Schwarzschild (or sub-extremal Reissner--Nordstr\"om) case, where the solution remains globally smooth in the exterior. This can be traced back to the existence of a bifurcation sphere in Schwarzschild, which does not exist in the extremal case. We refer the reader to section~\ref{sec:222} for a more detailed discussion.

\appendix
\section{Useful curvature computations}\label{sec:app:weyl}
In this appendix, we write down formulae for various curvature coefficients for the spherically symmetric Einstein-Scalar field system and, in particular, derive eq.~\eqref{eq:null:W3434}.
Recall that $g_{uv}=-\frac12\Omega^2$ and $g_{AB}=r^2\gamma_{AB}$. Recall, moreover, that we use capital Latin letters to denote coordinates on the sphere.

We first compute the Christoffel symbols. The only non-vanishing ones are (apart from those related by symmetry and from $\Gamma_{AB}^C$):
\begin{nalign}
\Gamma_{uu}^u&=\pu\log(\Omega^2),&&\Gamma_{vv}^v=\pv\log(\Omega^2),\\
\Gamma_{AB}^u&=-2\Omega^{-2}r\pu r\gamma_{AB}, &&\Gamma_{AB}^v=-2\Omega^{-2}r\pv r\gamma_{AB}, \\
\Gamma_{uB}^A&=\frac{\pu r}{r} \delta_B^A,&&\Gamma_{vB}^A=\frac{\pv r}{r} \delta_B^A.
\end{nalign}
Next, we compute some of the Riemann tensor components, using the definition
\begin{equation}
R^{\mu}_{\nu\xi o}:=\partial_\xi\Gamma^\mu_{\nu o}-\partial_o\Gamma^\mu_{\nu\xi}+\Gamma_{\xi\pi}^{\mu}\Gamma_{\nu o}^\pi-\Gamma_{o\pi}^{\mu}\Gamma^{\pi}_{\nu\xi}.
\end{equation}
We have:
\begin{align}
R_{AuBu}&=g_{AB}\frac{\zeta^2}{r^2},&&R_{AvBv}=g_{AB}\frac{\theta^2}{r^2},\nonumber\\
R_{Avuv}&=R_{Auvu}=0=R_{ABuv},\\
R_{3434}&=-\frac{\Omega^2}{2}\left(\Omega^2\frac{m}{r^3}-2\pu\phi\pv\phi\right).\nonumber
\end{align}
Finally, we compute the Weyl curvature tensor
\begin{equation}
W_{\mu\nu\xi o}:=R_{\mu\nu\xi o}-\frac12(R_{\mu\xi}g_{\nu o}-R_{\mu o}g_{\nu\xi}+R_{\nu o}g_{\mu\xi}-R_{\nu\xi}g_{\mu o})-\frac{R}{6}(g_{\mu\xi}g_{\nu o}-g_{\mu o}g_{\nu\xi})
\end{equation}
by using the Einstein equations (here, $T$ denotes the trace of $T_{\mu\nu}$, $T=g^{\mu\nu}T_{\mu\nu}$)
\begin{equation}
R_{\mu\nu}=2T_{\mu\nu}-g_{\mu\nu}T
\end{equation}
along with the fact that  $T=\frac{4}{\Omega^2}\pu\phi\pv\phi $
to obtain that
\begin{equation}
W_{AuBu}=0=W_{AvBv}=W_{Auvu}=W_{Avuv}=W_{ABuv}
\end{equation}
and
\begin{equation}
W_{uvuv}=-\frac{\Omega^4}{4}\frac{2m}{r^3}+\frac83\Omega^2\pu\phi\pv\phi.
\end{equation}
\section{Constructing the time integral from characteristic initial data}\label{APPendixB}
In this appendix, we prove the analogue of Thm.~\ref{thm:null:asymptotics of dvrphi} for general $p$, albeit only for the case of the linear wave equation. The crucial part of the proof is the construction of time integrals from characteristic initial data. Note that this has already been done in~\cite{Angelopoulos2018ASpacetimes} and in an improved way in~\cite{Angelopoulos:2018uwb}, Proposition~10.1. We will follow the approach of~\cite{Angelopoulos:2018uwb}.
\begin{thm}\label{thm:B}
Consider the same setup as in section~\ref{sec:null}, but for the linear, uncoupled problem with Eddington--Finkelstein double null coordinates $(u,v)$ (i.e.\ $\pv r=1-\frac{2M}{r}=-\pu r=:D$).
Assume, moreover, that $2\leq p\in \mathbb{Z}$.  In view of Corollary~\ref{cor:nullcase:asymptotics}, we then have, for all $v\geq 1$,
\begin{equation}
r\phi(u,v)=\frac{\Phi^-_p}{|u|^{p-1}}+\mathcal{O}(|u|^{-p+1-\epsilon}).
\end{equation}
Moreover, we have the following asymptotic expansion near $\mathcal{I}^+$:
\begin{equation}
    \pv(r\phi)=\sum_{i=0}^{p-3}\frac{f_i^{(p)}(u)}{r^{3+i}}+(-1)^{p-1} (p-1)p  \frac{\Phi^-_p M(\log r-\log|u|)}{r^{3+p-2}}+\mathcal{O}(r^{-3-p+2}) 
\end{equation}
for some smooth functions $f_i^{(p)}$.
\end{thm}
\begin{proof}
Since we have already shown the result for $p=2$, we will restrict to $p>2$.

We want to use the time integral trick from the proof of Theorem~\ref{scatteringthm2}.
It is clear that the only ingredient missing in order to just repeat the proof from Thm.~\ref{scatteringthm2} is to show that we can write 
$$r\phi=(\pu+\pv)(r\phi^T)$$
for some $r\phi^T$ that still solves the wave equation, which we shall call the time integral of $r\phi$.

In order to construct such a time integral, we note that if $\phi^T$ solves the wave equation \eqref{eq:sys:dudvphi}, and if moreover $(\pu+\pv)\phi^T=\phi$, then we have
\begin{align}
\begin{split}
\pu(r^2\pu\phi^T)=\pu(r^2\phi-r^2\pv\phi^T)&=
r\pu(r\phi)-D r\phi-r^2\pu\pv\phi^T+2Dr\pv\phi^T\\&=r\pu(r\phi)-Dr\phi+rDT\phi^T=r\pu(r\phi).
\end{split}
\end{align}
If we further impose that $r^2\pu\phi^T$ and $\phi^T$ vanish at $\mathcal I^-$, then it follows that 
\begin{equation}
r\phi^T(u,v)=r\int_{-\infty}^u\frac{1}{r^2}\int_{-\infty}^{u' }r\pu(r\phi)(u'',v)\dd u''\dd u'.
\end{equation}
The rest of the proof then follows as in the proof of Theorem \ref{scatteringthm2}.
\end{proof}

%\newpage
{\small\bibliographystyle{ieeetr} 
\bibliography{references,references2, references3}}

\end{document}